\documentclass{article}

\usepackage{PRIMEarxiv}

\usepackage[utf8]{inputenc} 
\usepackage[T1]{fontenc}    
\usepackage{hyperref}       
\usepackage{url}            
\usepackage{booktabs}       
\usepackage{amsfonts}       
\usepackage{nicefrac}       
\usepackage{microtype}      
\usepackage{lipsum}
\usepackage{fancyhdr}       
\usepackage{graphicx}       
\graphicspath{{media/}}     

\usepackage{enumitem}
\usepackage[noend]{algpseudocode}
\usepackage[linesnumbered,ruled,vlined]{algorithm2e}
\usepackage{comment}
\usepackage{graphicx}
\usepackage{booktabs}
\usepackage{url}
\usepackage{hyperref}
\usepackage{xcolor}
\usepackage{subcaption}
\usepackage{multirow}
\usepackage{mathrsfs}
\usepackage{svg}
\usepackage{xurl}
\usepackage{amsmath}
\usepackage{amsthm}

\newcommand{\algo}{{\em Ours-NoT}}
\newcommand{\algomu}{{\em Ours-$\mu$T}}
\newcommand{\algodelta}{{\em Ours}}

\newcommand{\eps}{\varepsilon}
\newcommand{\Gc}{G_{\text{core}}}
\newcommand{\Vc}{V_{\text{core}}}
\newcommand{\Ec}{E_{\text{core}}}

\newcommand{\myqed}{\begin{flushright}
    \qed
    \end{flushright}
}

\newcommand{\costeu}{\text{\em cost}_\text{EU}}
\newcommand{\costei}{\text{\em cost}_\text{EI}}
\newcommand{\costed}{\text{\em cost}_\text{ED}}
\newcommand{\costef}{\text{\em cost}_\text{EF}}
\newcommand{\costec}{\text{\em cost}_\text{EC}}
\newcommand{\costcf}{\text{\em cost}_\text{CF}}
\newcommand{\costcu}{\text{\em cost}_\text{CU}}

\newtheorem{fact}{Fact}
\newtheorem{observation}{Observation}
\newtheorem{claim}{Claim}
\newtheorem{definition}{Definition}
\newtheorem{theorem}{Theorem}
\newtheorem{lemma}{Lemma}

\newcommand{\shortauthors}{\emph{et al.}}
  
\title{Dynamic Structural Clustering Unleashed: Flexible Similarities, Versatile Updates and for All Parameters
}

\author{
  Zhuowei Zhao \\
  The University of Melbourne \\
  \texttt{zhuoweiz1@student.unimelb.edu.au} \\
   \And
  Junhao Gan \\
  The University of Melbourne \\
  \texttt{junhao.gan@unimelb.edu.au} \\
     \And
  Boyu Ruan \\
  The Hong Kong Unversity of Science and Technology \\
  \texttt{boyuruan@ust.hk} \\
     \And
  Zhifeng Bao \\
  RMIT University \\
  \texttt{zhifeng.bao@rmit.edu.au} \\
     \And
  Jianzhong Qi \\
  The University of Melbourne \\
  \texttt{jianzhong.qi@unimelb.edu.au} \\
     \And
  Sibo Wang \\
  The Chinese University of Hong Kong \\
  \texttt{swang@se.cuhk.edu.hk} \\
}

\begin{document}
\maketitle

\begin{abstract}

We study structural clustering on graphs in {\em dynamic} scenarios, where the graphs can be updated by {\em arbitrary} insertions or deletions of edges/vertices. The goal is to efficiently compute structural clustering results for any clustering parameters $\eps$ and $\mu$ given {\em on the fly}, for arbitrary graph update patterns, and for all typical similarity measurements. Specifically, we adopt the idea of update affordability and propose an a-lot-simpler yet more efficient (both theoretically and practically) algorithm (than state of the art), named {\em VD-STAR} to handle graph updates.
First, with a theoretical clustering result quality guarantee, {\em VD-STAR} can output high-quality clustering results with up to 99.9\% accuracy. Second, our {\em VD-STAR} is easy to implement as it just needs to maintain certain sorted linked lists and hash tables, and hence, effectively
enhances its deployment in practice. 
Third and most importantly, by careful analysis, {\em VD-STAR} improves the per-update time bound of the state-of-the-art from $O(\log^2 n)$ expected {\em with} certain update pattern assumption to $O(\log n)$ {\em amortized in expectation} {\em without} any update pattern assumption. 
We further design two variants of {\em VD-STAR} to enhance its empirical performance. Experimental results show that our algorithms consistently outperform the state-of-the-art competitors by up to 9,315 times in update time across nine real datasets.

\end{abstract}


\section{Introduction}
Graph clustering is a fundamental problem that aims to group similar vertices of a graph into clusters. 
Various clustering schemes were proposed and have been widely studied, such as modularity-based clustering~\cite{newman2004finding} and spectral clustering~\cite{white2005spectral}. 
In this paper, we focus on a popular scheme named {\em Structural Clustering}~\cite{xu2007scan} that groups the vertices based on their {\em structural similarity} in the graph. 
An important feature of Structural Clustering 
is that it identifies not only clusters of the vertices but also {\em different roles} 
for the vertices (i.e., cores, hubs, and outliers) in the clustering result.

\vspace{1mm}
\noindent
{\bf Structural Clustering.}
Given an undirected graph $G = \langle V, E \rangle$, a specified {\em similarity measurement} $\sigma$ between the neighborhoods of vertices, a similarity threshold parameter $0 < \eps \leq 1$ and an integer parameter $\mu \geq 1$, 
the process of structural clustering starts with identifying a set of special vertices known as the {\em core vertices}. 
A core vertex is a vertex $u\in V$ with at least $\mu$ {\em similar neighbors}, 
and a similar neighbor is a neighbor vertex of $u$ with a similarity score $\geq \eps$.
The core vertices, and the edges to their similar core neighbors, collectively form connected components (CCs), where each CC serves as a {\em primitive cluster}. 
The non-core vertices are then assigned to the corresponding primitive clusters of their similar core neighbors.
For each vertex $u \in V$ which {\em does not} belong to any cluster, 
if $u$ connects to neighbors from two or more clusters, then $u$ is categorized as a {\em hub}, in a sense that $u$ serves as a bridge connecting multiple clusters.
Otherwise, $u$ is considered as an {\em outlier}.

Figure~\ref{fig:clustering_exp} shows an example of structural clustering results with $\eps = 0.5$, $\mu= 4$, using Jaccard similarity as the similarity measurement between vertices. 
Vertices $v_1$ and $v_8$ are core vertices (solid lines indicate edges connecting similar neighbors) and hence, each of them itself forms a primitive cluster. 
Vertices $\{v_1, v_2, v_3, v_4, v_5\}$ form a cluster, 
while $\{v_8, v_9, v_{10}, v_{11}, v_{12}\}$ form another. 
Both $v_7$ and $v_6$ belong to none of the clusters -- $v_7$ is a hub as it connects to neighbors belonging to two different clusters, while $v_6$ is an outlier. 


\begin{figure}
    \centering
    \includegraphics[width=0.8\textwidth]{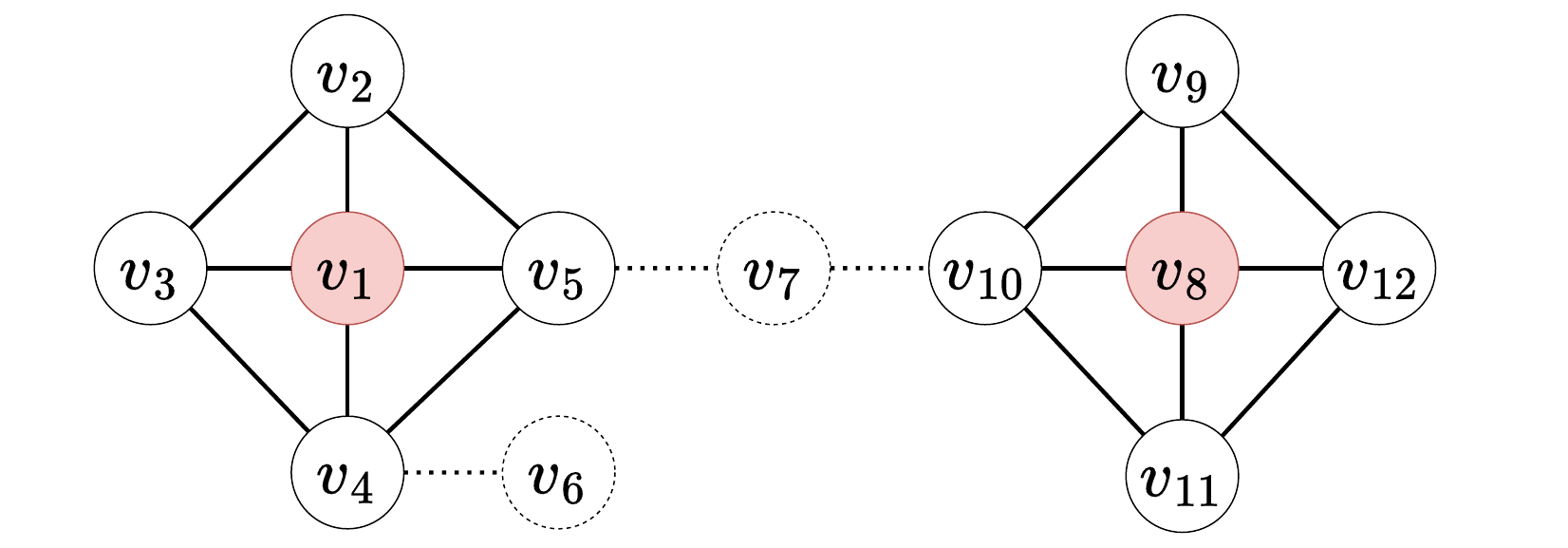}
    \caption{A Structural Clustering Example ($\eps = 0.5$ and $\mu = 5$)}
    \label{fig:clustering_exp}
\end{figure}

\begin{table*}[ttt]
\caption{Comparison with SOTA Methods, where $n$ is the number of vertices, $m$ is the number of edges, $d_{\max}$ is the maximum degree, and $m_{cr}$ is the number of edges in the clustering result graph}
\label{tab:relatedwork_compare}
\resizebox{\textwidth}{!}{%
\begin{tabular}{l|l|c||c|c|c}
\hline
\multicolumn{2}{c|}{Algorithm Features}   & \textbf{Our {\em VD-STAR}}   & BOTBIN~\cite{zhang2022effective} & DynELM~\cite{ruan2021dynamic} & GS*-Index~\cite{wen2017efficient}  \\ \hline
\multirow{2}{*}{Update running time} 

    & Arbitrary update & $O(\log n)$ amortized expected & $O(d_{\max}\log n)$ & $O(\log^2 n)$ amortized  & $O(d_{\max}^2\log n)$ \\ \cline{2-6}

    & Uniformly at random & $O(\log n)$ amortized expected &  $O(\log^2 n)$ expected   & $O(\log^2 n)$ amortized         & $O(d_{\max}^2\log n)$  \\ \hline
\multirow{3}{*}{Similarity measurement}  
    & Jaccard similarity                       & \textbf{\checkmark}  & \checkmark     & \checkmark   & \checkmark               \\ \cline{2-6}
    & Cosine similarity                       & \textbf{\checkmark}  &   \text{\sffamily X}   &    \checkmark      &  \checkmark          \\ \cline{2-6}
    & Dice similarity                        & \textbf{\checkmark}  &   \text{\sffamily X}   &    \checkmark      &   \checkmark        \\ \hline
\multicolumn{2}{l|}{Query running time for $\eps$ and $\mu$ given on the fly} & $O(m_{cr})$ & $ O(m_{cr}) $ & $O(m \log^2 n)$ running from scratch  & $O(m_{cr})$ \\ \hline
\end{tabular}%
}
\vspace{1mm}
\end{table*}

\vspace{1mm}
\noindent
{\bf Applications.} 
Structural clustering is widely applied in genomics and biomarker discovery~\cite{ding2012atbionet}, and over other biology data such as clinical data~\cite{liu2011translating} and protein-protein interaction networks~\cite{martha2011constructing}. Another important application is community detection over social networks~\cite{lim2014linkscan, papadopoulos2012community, fortunato2010community}, to unveil community structures, offering insights into the relationships and interactions among social network users. 
Similarly, in the context of web data~\cite{papadopoulos2009leveraging, papadopoulos2010graph}, structural clustering contributes to the identification and analysis of interconnected communities, enhancing the comprehension of web structures and user behaviors. 
Structural clustering has also been used for fraud detection~\cite{chawathe2019clustering, panigrahi2022ctb} since block-chain became a heat.
Recently, the AI community has turned its attention to structural clustering for improving model training~\cite{fajri2024structural, yang2022self}. Among these applications, dynamic scenarios are particularly significant, as graphs can rapidly expand over time, as seen in clinical data, social networks, web-page hyperlinks, and blockchain, to name a few.

\vspace{1mm}
\noindent
{\bf Related Work and Limitations.}
Structural clustering was first proposed by Xu~\shortauthors~\cite{xu2007scan}.
Since then, it has opened a line of studies. 
Among these, pSCAN~\cite{chang2017mathsf} is the state-of-the-art (SOTA) {\em exact} algorithm for  {\em static} graphs $G$, where no updates to $G$ are allowed.
The running time complexity of pSCAN is bounded by $O(m^{1.5})$, 
where $m$ is the number of edges.
As shown by Chang~\shortauthors~\cite{chang2017mathsf}, this $O(m^{1.5})$ bound is indeed worst-case optimal.

When the graphs are {\em dynamic}, where the graphs can be updated by insertions or deletions of edges, structural clustering 
becomes even more challenging. 
GS*-Index~~\cite{wen2017efficient} is the SOTA for {\em Dynamic Structural Clustering}, which can return
the exact clustering result with respect to parameters $\eps$ and $\mu$ given on the fly for each query.
However, it takes $O(d_{\max}^2 \cdot \log n)$ worst-case time to process each update, where $d_{\max}$ is the maximum degree and $n$ is the number of vertices in the current graph.

DynELM~\cite{ruan2021dynamic} and BOTBIN~\cite{zhang2022effective} are two SOTA {\em approximate} algorithms, yet BOTBIN can {\em only} work for Jaccard similarity.
DynELM can process each update in $O(\log^2 n)$ amortized time for pre-specified parameters $\eps$ and $\mu$. 
In contrast, BOTBIN supports queries with $\eps$ and $\mu$ given on the fly, 
and can process each update in $O(\log^2 n)$ time {\em in expectation} under the assumption that the updates are uniformly at random within each vertex's neighborhood.
This assumption, however, may not always hold for real-world applications, e.g., people tend to follow big names in social networks.
If this is the case, the per-update complexity of BOTBIN degenerates to $O(d_{\max} \cdot \log n)$.


\vspace{1mm}
\noindent
{\bf Our Method.} 
Given the importance of structure clustering, 
we propose an ultimate algorithm, called {\em \underline{V}ersatile \underline{D}ynamic \underline{St}ructur\underline{a}l Cluste\underline{r}ing} ({\em VD-STAR}),
which 
unifies the state of the art and addresses all their 
limitations.
Specifically, we use a sampling method to estimate the similarity which we show to work on all three similarity measurements (Jaccard, Cosine, and Dice) suggested by Xu~\shortauthors~\cite{xu2007scan} (Section \ref{sec:analysis}).
Inspired by the idea of update affordability (Section~\ref{sec:our-algo}) which is not affected by the update pattern, we propose a bucketing technique to handle multiple update affordability for a neighborhood in a more efficient way that improves the complexity bound. 
Here, the update affordability indicates how many updates can an edge afford before it can possibly affect the clustering results; this concept will be detailed in Section~\ref{subsec:update_afford}. 
Additionally, we propose an algorithm to track the update affordability for each edge, which is easy to implement yet efficient (Section \ref{subsec:edgesim}). 
Last but not least, 
we propose a unified framework (Algorithm~\ref{algo:framework}) for GS*-Index, BOTBIN, and our {\em VD-STAR}, with which users just need to implement the specified interfaces to obtain the algorithms.
This also provides flexibility for customization, that is, users 
can swap the implementations between different algorithms to make their own ``new'' solutions. 
As we will see in Section~\ref{sec:opt}, this is actually what we did to design the two variants of {\em VD-STAR}. 
To summarize, as shown in Table~\ref{tab:relatedwork_compare}, our {\em VD-STAR} advances the SOTA in the following aspects:
\vspace{-1mm}
\begin{itemize}[leftmargin = *]
\item It supports all three similarity measurements, whereas the SOTA method, BOTBIN~\cite{zhang2022effective}, is designed to support Jaccard similarity only.
\item It improves the per-update time complexity of the SOTA methods, BOTBIN and DynELM, from $O(\log^2 n)$ expected to $O(\log n)$ amortized in expectation.
\item It supports arbitrary update patterns lifting BOTBIN's strong assumption that the updates have to be uniformly random.
\end{itemize}
\vspace{-1mm}

\noindent
Overall, we make the following contributions: 
\vspace{-1mm}
\begin{itemize}[leftmargin = *]
\item We propose a novel algorithm, called {\em VD-STAR}, 
which addresses all the aforementioned   
challenges, and hence, overcomes all the limitations of the SOTA algorithms, and most importantly, is even more efficient in processing updates! 

%

\item While the theoretical analysis is technical, 
our {\em VD-STAR} can be easily implemented in practice, as our novel algorithm design consists mainly of fundamental data structures: maintaining and scanning a number of sorted lists and hash tables.

\item We conduct extensive experiments on nine real-world graph datasets with up to 117 million edges to compare our algorithms with GS*-Index and BOTBIN, in terms of update efficiency (with varying update distributions),  clustering quality (under different similarity measurements), and query efficiency. 
The experimental results show that our proposed algorithms outperform SOTA algorithms by up to $9,315\times$ regarding update efficiency.
\end{itemize}

\vspace{-3mm}
\section{Preliminaries}
\label{sec:preli}


\subsection{Problem Formulation}
\label{sec:prob}

Consider an undirected graph $G = \langle V, E \rangle$, 
where $V$ is a set of $n$ vertices and $E$ is a set of $m$ edges. 
Vertices $u \in V$ and $v \in V$ are {\em neighbors} if and only if there exists an edge $(u,v) \in E$. 
The {\em neighborhood} of $u$, denoted by $N(u)$, is 
the set of all $u$'s neighbors, namely, $N(u) = \{v \in V | (u,v) \in E\}$, and the {\em degree} of $u$ is defined 
to be $d_u = |N(u)|$. 
Moreover, we use
$N[u] = N(u) \cup \{u\}$ to denote the {\em inclusive neighborhood} of $u$ and let $n_u = |N[u]|$.

\noindent\textbf{Similarity Measurement.} 
The similarity between vertices $u$ and $v$ is denoted by $\sigma(u, v)$. 
Specifically, $\sigma(u, v) = 0$ if there is {\em no} edge between $u$ and $v$; 
otherwise, depending on the application needs, $\sigma(u,v)$ is calculated as one of the following three popular similarity measurements, where
$I(u,v) = |N[u]\cap N[v]|$ and $U(u,v) = |N[u]\cup N[v]| = n_u + n_v - I(u,v)$:
\begin{itemize}[leftmargin =*]
\item {\em Jaccard similarity}: 
$\sigma(u,v) = \frac{I(u,v)}{U(u,v)} = \frac{I(u,v)}{n_u + n_v - I(u,v)}$, or
\item {\em Cosine similarity}: 
$\sigma(u,v) = \frac{I(u,v)}{\sqrt{n_u \cdot n_v}}$, or
\item {\em Dice similarity}: 
$\sigma(u,v) = \frac{I(u,v)}{(n_u + n_v)/2}$.
\end{itemize}
And 
$\sigma(u,v)$ can be computed in $O(1)$ time with $I(u,v)$, $n_u$ and $n_v$.

\vspace{1mm}
\noindent\textbf{Similar Neighbors, Edge Labels and Core Vertices.} 
Given a {\em similarity threshold} $0 < \eps < 1$, vertices $u$ and $v$ are {\em $\eps$-similar neighbors} if 
$\sigma(u,v) \geq \eps$.
An edge $(u, v)$ is labelled as an {\em $\eps$-similar edge} if $u$ and $v$ are $\eps$-similar neighbors;
otherwise, it is labeled as a {\em $\eps$-dissimilar edge}. 
%
A vertex $u$ is a {\em $(\eps, \mu)$-core vertex} 
if $u$ has {\em at least} $\mu$ $\eps$-similar neighbors; 
otherwise, $u$ is a \emph{non-core vertex} with respect to $\eps$ and $\mu$. 

In the rest of this paper, when the context of the parameters $\eps$ and $\mu$ is clear, 
we use {\em sim-edges} to refer to $\eps$-similar edges, and {\em core vertices} to refer to $(\eps, \mu)$-core vertices, respectively.

\noindent\textbf{Core Sim-Graph.} 
An edge $(u, v)$ is a {\em core sim-edge} if $(u,v)$ is a sim-edge and both $u$ and $v$ are core vertices.
%
The {\em core sim-graph} of $G$ is defined as 
$\Gc= \langle \Vc, \Ec \rangle$,
where $\Vc$ is the set of all the core vertices and $\Ec$
is the set of all core sim-edges.

\noindent{\bf Structural Clusters and the Clustering Result.} 
Each connected component (CC) of $\Gc$ is defined as a {\em primitive (structural) cluster}. 
And each primitive cluster $C$, along with the set of all the {\em non-core} vertices that are similar neighbors of some core vertex in $C$, is defined as a {\em structural cluster} (``cluster'' for short). 
The collection of all these clusters represents the {\em Structural Clustering Result} (``clustering result'' for short) on $G$ with respect to the parameters $\eps$ and $\mu$.
\noindent{\bf Clustering Result Graph.}
Given $\eps$ and $\mu$,
let 
$E_{cr}$ be the set of all the sim-edges that are incident on at least one {\em core} vertex. 
Denote by $G_{cr} = \langle V_{cr}, E_{cr} \rangle$ the induced sub-graph of $G$ by $E_{cr}$, where $V_{cr}$ is the set of all the end-vertices of the edges in $E_{cr}$.
$G_{cr}$ is called the {\em Clustering Result Graph} of $G$ with respect to $\eps$ and $\mu$.  
Moreover, we define $n_{cr} = |V_{cr}|$ and $m_{cr} = |E_{cr}|$.
\vspace{-2mm}
\begin{observation}\label{fact:clustering-graph}
Given the clustering result graph $G_{cr}$ with respect to the given parameters $\eps$ and $\mu$, 
the structural clustering result on $G$ can be computed in $O(m_{cr})$ time.
\end{observation}
\vspace{-4mm}
\begin{proof}
By scanning $G_{cr}$, the core sim-graph $\Gc$ can be obtained, as $\Gc$ is a sub-graph of $G_{cr}$.
As a result, all the primitive clusters (i.e., the connected components of $\Gc$) can be computed in $O( |\Vc| + |\Ec|)$ time.
Finally, for each edge $(u,v) \in E_{cr}$ that is incident on a non-core vertex $v$, assign $v$ to the cluster of the core vertex $u$.  
The overall running time is bounded by $O(m_{cr})$.
\end{proof}
\vspace{-2mm}

\noindent
{\bf Problem Definition.}
We consider structural clustering on graph $G$ which 
{\em evolves over time}.
The problem is defined as follows.

\begin{definition}
Consider a pre-specified similarity measurement $\sigma(\cdot, \cdot)$ (either Jaccard, Cosine or Dice); 
given an undirected graph $G = \langle V, E \rangle$ that can be updated by {\em arbitrary} insertions or deletions of edges, 
the {\bf \em Dynamic Structural Clustering for All Parameters} ({\bf DynStrClu-AllPara}) problem asks to: 
\begin{itemize}[leftmargin = *]
\item (i) support updates {\em efficiently}, 
and
\item (ii) return the clustering result upon request in $O(m_{cr})$ time with respect to the parameters $\eps \in (0,1)$ and $\mu \geq 1$ \underline{\em given on the fly}.
\end{itemize}
\end{definition}

\vspace{-2mm}
\noindent
{\bf Affecting Updates and Affected Edges.}
Observe that an update (either an insertion or a deletion) of an edge $(u,v)$ changes the degrees of both $u$ and $v$, 
and hence, the similarities of all the edges incident on $u$ or $v$ are {\em affected}. 
These edges incident on $u$ or $v$ are called {\em affected edges} of the update $(u,v)$,
and this update $(u,v)$ is an {\em affecting update} to these edges.
A main challenge in {\em DynStrClu-AllPara}
is that for an udpate $(u,v)$, there can be $O(d_u + d_v) \subseteq O(n)$ affected edges.
Thus, maintaining the similarities of these affected edges can be expensive.
As we shall see,
how to overcome this technical difficulty is the main distinction among  
the state-of-the-art
solutions and our proposed algorithms.

\noindent{\bf $\rho$-Absolute-Approximation.}
We exploit the notion of $\rho$-absolute approximation for {\em DynStrClu-AllPara}.

\begin{definition}[$\rho$-Absolute-Approximation]
Given a {\em constant} parameter $\rho \in (0, 1)$ and the similarity threshold  
parameter $\eps$,
the label of an edge $(u,v)$ is decided as follows:
\begin{enumerate}[leftmargin=*]
    \item if $\sigma(u, v) > \eps + \rho$, 
$(u,v)$ must be considered as similar;
    \item if $\sigma(u, v) < \eps - \rho$, 
$(u,v)$ must be considered as dissimilar;
\item otherwise, i.e., $\eps - \rho \leq \sigma(u, v) \leq \eps + \rho$, 
$(u,v)$ can be considered as either similar or dissimilar.
\end{enumerate}
\end{definition}
\vspace{-1mm}

Once the edge labels are decided according to the above $\rho$-absolute-approximation,
all the other definitions introduced in 
this section,
immediately follow.
Moreover, it is shown~\cite{zhang2022effective} that clustering result under the notion of $\rho$-absolute-approximation provides
a ``sandwich'' guarantee on the result quality compared to the exact clustering result with the same parameters $\eps$ and $\mu$. 

\begin{fact}[\cite{ruan2021dynamic,zhang2022effective}]
Given parameters $\eps$, $\mu$, and $\rho$, let $\mathcal{C}_{\eps,\mu}$ denote the exact clustering result, and let $\mathcal{C}^\rho_{\eps,\mu}$ denote the clustering result satisfying the $\rho$-absolute-approximation. We have the following properties:
    \begin{itemize}[leftmargin = *]
        \item for every cluster $C_+ \in \mathcal{C}_{\eps+\rho,\mu}$, there is a cluster $\Tilde{C} \in \mathcal{C}^\rho_{\eps,\mu}$ such that $C_+ \subseteq \Tilde{C}$;
        \item for every cluster $\Tilde{C} \in \mathcal{C}^\rho_{\eps,\mu}$, there is a cluster $C_- \in \mathcal{C}_{\eps-\rho,\mu}$ such that $\Tilde{C} \subseteq C_-$.
    \end{itemize}
\end{fact}

\begin{proof}
    Let $G_{cr+} = \langle V_{cr+}, E_{cr+} \rangle$ , $G_{cr-} = \langle V_{cr-}, E_{cr-}\rangle$, and $G_{cr_\rho} = \langle V_{cr_\rho}, E_{cr_\rho}\rangle$ be the clustering results graphs $\mathcal{C}_{\eps+\rho,\mu}$, $\mathcal{C}_{\eps-\rho,\mu}$, and $\mathcal{C}^\rho_{\eps,\mu}$, respectively for given $\mu$, $\eps$, and $\rho$. For an edge $(u,v) \in E_{cr+}$, we have $\sigma(u,v) \geq \eps + \rho$, hence $u$ and $v$ are considered similar under the $\rho$-absolute-approximation. Additionally, as $(u,v) \in E_{cr+}$, at least one of $u$ and $v$ must be a core vertex in $G_{cr+}$. Without loss of generality, we assume $u$ to be a core vertex in $G_{cr+}$. Then, $u$ is also a core vertex in $G_{cr_\rho}$ based on the $\rho$-absolute-approximation. This is because $u$ can only have more similar neighbors under a more relaxed threshold. Therefore, $(u,v)$ must be in $E_{cr_\rho}$, and hence $E_{cr+} \subseteq E_{cr_\rho}$. Symmetrically, for an edge $(u,v) \in E_{cr_\rho}$, it has $\sigma(u,v) \geq \eps - \rho$, hence $u$ and $v$ are considered similar in $G_{cr-}$. Similarly, a core vertex in $E_{cr_\rho}$ must be a core vertex in $G_{cr-}$, and we have $E_{cr_\rho} \subseteq E_{cr-}$.

    Let $V_+$ and $E_+$ be the vertex and edge sets of $C_+ \in \mathcal{C}_{\eps+\rho,\mu}$, respectively. We have $E_+ \subseteq E_{cr+} \subseteq E_{cr_\rho}$ as verified above. Therefore, there is a cluster $\Tilde{C} \in \mathcal{C}^\rho_{\eps,\mu}$ that contains all the vertex in $V_+$, and hence $C_+\subseteq\Tilde{C}$. This proves the first bullet point. For the second bullet point, let $V_\rho$ and $E_\rho$ be the vertex and edge sets of $\Tilde{C} \in \mathcal{C}^\rho_{\eps,\mu}$, respectively. We have $E_\rho \subseteq E_{cr_\rho} \subseteq E_{cr-}$. Therefore, there exists a cluster $C_- \in \mathcal{C}_{\eps-\rho,\mu}$ that contains all the vertex in $V_\rho$, and hence $\Tilde{C}\subseteq C_-$.

\end{proof}

\vspace{-2mm}
\subsection{A Unified Algorithm Framework}

For ease of presentation, 
we 
introduce 
a {\em unified algorithm framework} for solving {\em DynStrClu-AllPara}, 
which is shown in Algorithm~\ref{algo:framework}.
The SOTA exact and approximate algorithms discussed in this paper, as well as our solutions,
can work under this framework.
%
At a high level, these algorithms implement the following data structures:
\begin{itemize}[leftmargin = *]
\item {\bf \em Sorted Neighbor Lists}: for each vertex $u \in V$, a {\em non-increasing sorted} list of $u$'s neighbors by their similarities to $u$; with a slight abuse of notation, we simply use $N(u)$ to refer to this sorted list, and each neighbor $v$ of $u$ in this list is stored along with its similarity to $u$. 
\item {\bf \em EdgeSimStr}: a data structure for maintaining the (approximate) similarities for all the edges;
there are five functions:
\begin{itemize}
\item update($(u,v)$, $op$): given an update of edge $(u,v)$, where $op$ indicates whether this is an insertion or a deletion, update the information maintained in {\em EdgeSimStr} accordingly;
\item insert($(x,y)$): insert an edge $(x, y)$ to {\em EdgeSimStr};
\item delete($(x,y)$): delete an edge $(x, y)$ from {\em EdgeSimStr}; 
\item find($(u,v)$, $op$): return a set $F$ of all the affected edges whose similarities are considered ``invalid'' and thus need to be re-computed;
\item cal-sim($(x, y)$): given an edge $(x, y)$, return $\sigma(x,y)$;  
\end{itemize} 
\item {\bf \em CoreFindStr}: a data structure for finding core vertices; it has two functions:
\begin{itemize}
\item update($u$): given a vertex $u \in V$, update {\em CoreFindStr};
\item find-core($\eps$, $\mu$): return the set $\Vc$ of all the core vertices with respect to the given parameters $\eps$ and $\mu$; 
\end{itemize}
\end{itemize}
%

\noindent
{\bf Running Time Analysis.}
Let
$\costeu$, $\costei$, $\costed$, $\costef$ and $\costec$ denote the running time cost of each invocation of the functions {\em update}, {\em insert}, {\em delete}, {\em find} and {\em cal-sim} in {\em EdgeSimStr}, respectively;
and $\costcf$ and $\costcu$ denote the running time cost of the functions {\em find-core}
and {\em update} in {\em CoreFindStr}, respectively.

\noindent
\underline{\em Query Running Time.}
By Observation~\ref{fact:clustering-graph}, the running time cost for each query is bounded by 
$O(\costcf + m_{cr})$.

\noindent
\underline{\em Per-Update Running Time.}
In the Update Procedure in Algorithm~\ref{algo:framework},
Line~2 takes $O(\costeu)$ time and Lines~3-9 takes $O(\costec + \costei + \costed + \log n)$ time, where $O(\log n)$ is the maintenance cost for the sorted neighbor lists.
Furthermore, the running time cost of Line~10 
is bounded $O(\costef)$ while that of Lines 11-15  is bounded by $O(|F| \cdot (\costec +  \costei + \costed + \log n))$. 
Finally, Lines 17 - 18 can be performed in $O(|F| \cdot \costcu)$ because $|S| \in O(|F|)$. 
Summing all these up, 
the overall running time of each update is bounded by 
\begin{small}
\begin{equation*}\label{eq:query-bound}
O\left(\costeu + \costef +  
(|F| + 1) \cdot (\costec + \costei + \costed + \log n + \costcu)\right)\,.
\end{equation*}
\end{small}

With this algorithm framework, one can just focus on the implementations for {\em EdgeSimStr} and {\em CoreFindStr} of different algorithms.
Substituting the corresponding costs to the above analysis, the query and per-update running time bounds follow.
 
\begin{small}
\begin{algorithm}[t]
\caption{A Unified Algorithm Framework}\label{algo:framework}
\SetKwComment{Comment}{// }{}

\SetKwBlock{update}{Update Procedure:}{end}
\SetKwFunction{esm}{EdgeSimStr}
\SetKwFunction{esf}{EdgeSimStr.find}

\SetKwFunction{esd}{EdgeSimStr.delete}
\SetKwFunction{esi}{EdgeSimStr.insert}
\SetKwFunction{esu}{EdgeSimStr.update}
\SetKwFunction{esc}{EdgeSimStr.cal-sim}

\SetKwFunction{csu}{CoreFindStr.update}

\SetKwBlock{query}{Query Procedure:}{end}
\SetKwFunction{icv}{CoreFindStr.find-core}

\update{
\KwIn{an update of $(u,v)$ flagged by $op \in \{\text{ins}, \text{del}\}$}

\esu{$(u,v)$, $op$}\;
\If{$op == \text{ins}$}    
{
\esc{$(u,v)$}\;
\esi{$(u,v)$}\;
insert $(u,v)$ to $E$ and maintain $N(u)$ and $N(v)$\;
}
\Else{
\esd{$(u,v)$}\;
remove $(u,v)$ from $E$ and maintain $N(u)$ and $N(v)$\;
}
\Comment{identify all the ``invalid'' edges}
$F \leftarrow$ \esf{$(u,v)$, $op$}\;
\For{each $(x,y)\in F$}{
\esd{$(x,y)$}\;
\esc{$(x,y)$}\;
\esi{$(x,y)$}\;
maintain $N(x)$ and $N(y)$\;
}
$S \leftarrow$ $\{\text{end-vertices of all the edges in}~F \} \cup \{u,v\}$\;
	
\For{each $x \in S$}{    
\csu{$x$}\;
}
}

\query{
\KwIn{parameters $0 < \eps < 1$ and $\mu \geq 1$}
\KwOut{a clustering result with respect to $\eps$ and $\mu$}
    $\Vc \leftarrow \emptyset$, $E_{cr} \leftarrow  \emptyset$\;
    
    $\Vc \leftarrow$ \icv{$\eps$, $\mu$}\;
	\For{each $u\in \Vc$}{
	\For{each $v\in N(u)$}{ 
			if $\sigma(u,v) \geq \eps$, add $(u,v)$ to $E_{cr}$; otherwise, break\;
	}
	} 
    \KwRet{the clustering result from $G_{cr}$ induced by $E_{cr}$;}
}
\end{algorithm}
\end{small}

\vspace{-2mm}
\subsection{A SOTA Exact Algorithm}
\label{sec:gsindex}

\begin{figure*}[t]
\begin{tabular}{c|c|c}
    \begin{subfigure}{.2\linewidth}
        \centering
        \includegraphics[width=\textwidth]{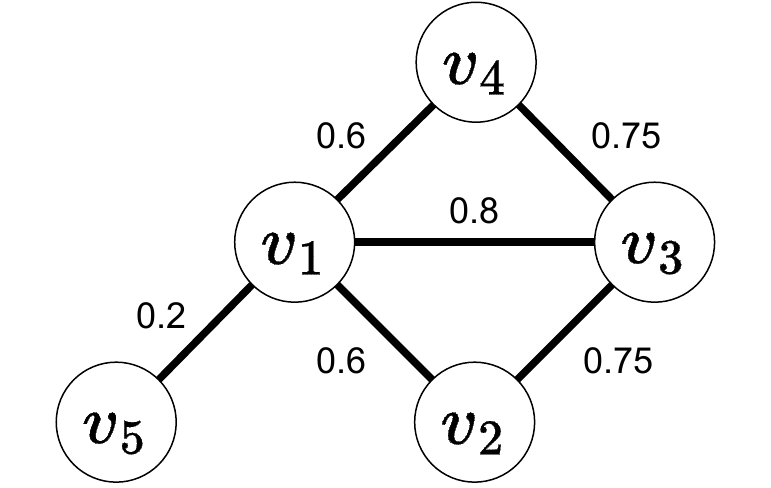}
        \caption{A graph example}
        \label{fig:index_graph}
    \end{subfigure} &
    \begin{subfigure}{.4\linewidth}
        \centering
        \includegraphics[width=\textwidth]{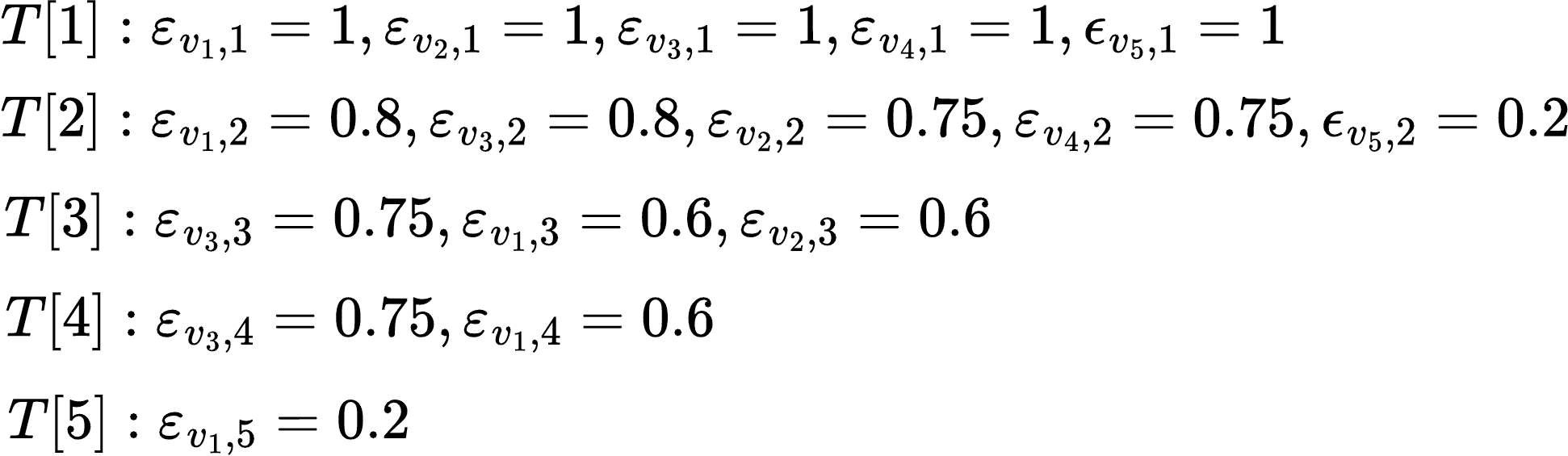}
        \caption{GS*-Index example}
        \label{fig:gs_index}
    \end{subfigure} &
    \begin{subfigure}{.35\linewidth}
        \centering
        \includegraphics[width=\textwidth]{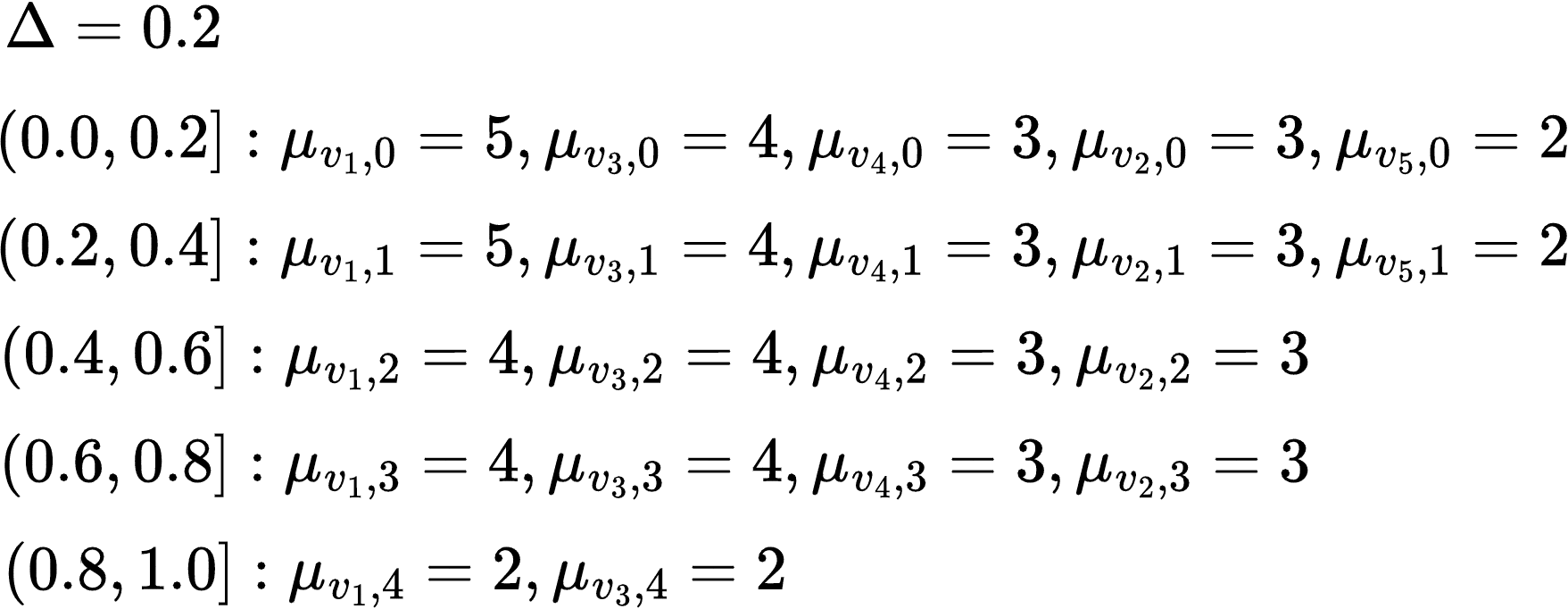}
        \caption{BOTBIN example}
        \label{fig:botbin}
\end{subfigure}
\end{tabular}
\caption{Index schema examples}
\end{figure*}

The {\em GS*-Index}~\cite{wen2017efficient} 
is a state-of-the-art exact algorithm for {\em DynStrClu-AllPara}. 
It implements {\em EdgeSimStr} and {\em CoreFindStr} as follows.

\noindent
{\bf The Implementation of {\em EdgeSimStr}.}
For each vertex $u \in V$, the {\em EdgeSimStr}, maintain $d_u$, the degree of $u$, and $I(u,x)$, the {\em intersection size} of $N[u]$ and $N[x]$, for each neighbor $x \in N(u)$. 
And the functions 
are implemented as follows:
\begin{itemize}[leftmargin = *]
\item update($(u,v)$, $op$): maintain the counters $d_u$ and $I(u,x)$ for each $x \in N(u)$ according to the given update. Perform the same maintenance symmetrically for the end-vertex $v$.
Therefore, $\costeu \in O(d_u + d_v) \subseteq O(d_{\max})$, where $d_{\max}$ is the largest degree of $G$.

\item find($(u,v)$, $op$): 
return the set of all the affected edges of the update $(u,v)$.
Hence, $|F| \in O(d_u + d_v)$ and hence, $\costef \in O(d_{\max})$.

\item cal-sim($(x, y)$): compute $\sigma(u,v)$ with $I(u,v)$, $d_u$ and $d_v$.
Thus, $\costec \in O(1)$. 

\item neither insert($(x,y)$) nor delete($(x,y)$) is used in GS*-Index; hence, $\costei = 0$ and $\costed = 0$. 
\end{itemize}

\noindent
{\bf The Implementation of {\em CoreFindStr}}: The GS*-Index implements the 
{\em CoreFindStr} as a {\em $\mu$-Table}, denoted by $T$, which is essentially an array of $d_{\max}$ sorted lists.
Specifically, for integer $1 \leq i \leq d_{\max}$, $T[i]$ stores a {\em sorted} list of all the vertices $u$ (with degrees $d_u \geq i$) in a {\em non-increasing} order by 
the $i^\text{th}$ largest similarity of $u$ to its neighbors, denoted by $\eps_{u,i}$.
\begin{itemize}[leftmargin = *]
\item find-core($\eps$, $\mu$): 
retrieve all the core vertices 
by scanning the vertices in the sorted list $T[\mu]$
until the first vertex $x$ such that $\eps_{x, \mu} < \eps$ is met or the entire list has been retrieved.
$\costcf \in O(|\Vc| + 1)$.

\item update($x$): maintain the sorted lists $T[1], \cdots, T[d_x]$ for $x$ accordingly. This takes $\costcu \in O(d_{\max}\cdot \log n)$ time because the maintenance on each of these sorted lists takes $O(\log n)$ time. 
\end{itemize}
%
\vspace{-1mm}
\begin{fact}
The GS*-Index can answer each query in $O(m_{cr})$ time and can handle each update in $O(d^2_{\max} \cdot \log n)$ time with space consumption bounded by $O(n + m)$ at all times.
\end{fact}
\vspace{-3mm}
\subsection{A SOTA  Approximate Algorithm}
\label{sec:botbin}

While the space consumption and the query time complexity of GS*-Index are good,
Unfortunately, the $O(d_{\max}^2 \cdot \log n)$ per-update time of GS*-Index is prohibitive as $d_{\max}$ can be as large as $n$.
To remedy this, 
BOTBIN~\cite{zhang2022effective} adopts the notion of {\em $\rho$-absolute-approximation}
for {\em DynStrClu-AllPara}.
It improves the per-update cost from $O(d_{\max}^2 \cdot \log n)$ to roughly $O(\log^2 n)$ {\em in expectation} assuming that the updates are uniformly at random within the neighborhood of each vertex.
However, BOTBIN only works for Jaccard similarity measurement and this expected update bound only holds under an assumption that the updates on $G$ are {\em uniformly at random}.

\noindent
{\bf The Implementation of {\em EdgeSimStr}.}
BOTBIN maintains a {\em bottom-$k$ signature}, denoted by $s(u)$, for each vertex $u \in V$ for some integer parameter $k$.
An $\rho$-absolute-approximate Jaccard similarity, denoted by $\tilde{\sigma}(u,v)$, between any two vertices $u$ and $v$
can be computed with their signatures $s(u)$ and $s(v)$. 
Specifically, it first generates and stores a fixed random permutation $\pi$ of $V$.
For each vertex $u\in V$, 
if $d_u \geq k$, then the {\em signature} of $u$, denoted by $s(u)$, is the set of the $k$ smallest neighbors in $N[u]$ according to 
the permutation order $\pi$; 
otherwise, $s(u) = N[u]$.
With the bottom-$k$ signatures, an $\rho$-absolute-approximate Jaccard similarity, denoted by $\tilde{\sigma}(u,v)$, between any two vertices $u$ and $v$
can be computed with $s(u)$ and $s(v)$. 
is computed as 
$\tilde{\sigma}(u,v) = \frac{|s(u) \cap s(v) \cap s(\{u,v\})|}{k}$, where $s(\{u,v\})$ is the $k$ smallest vertices in $s(u) \cup s(v)$ according to the permutation $\pi$; if $|s(u) \cup s(v)| < k$, then $s(\{u,v\}) = s(u) \cup s(v)$.

\begin{itemize}[leftmargin = *]
\item update($(u,v)$, $op$): update the signatures $s(u)$ and $s(v)$ with respect to the update of $(u,v)$ accordingly; it is known that this can be achieved in $O(\log n)$ time.
Thus, $\costeu \in O(\log n)$.
 
\item find($(u,v)$, $op$): if the signature $s(u)$ is changed due to this given update of $(u,v)$, add all the edges incident on $u$ to $F$; and perform the same symmetrically for $v$.
As a result, either $s(u)$ or $s(v)$ changes, then $\costef \in O(d_{\max})$; otherwise, $\costef \in O(1)$.

\item cal-sim($(x, y)$): return $\tilde{\sigma}(x,y)$ as the similarity of $x$ and $y$. This can be done in $O(k)$ time.

\item neither insert($(x,y)$) or delete($(x,y)$) is used in BOTBIN; hence, $\costei = 0$ and $\costed = 0$.

\end{itemize}

\noindent
{\bf The Implementation of {\em CoreFindStr}.}
BOTBIN implements {\em CoreFindStr} as an array, called {\em $\Delta$-Table} and denoted by $T_{\Delta}$, of $\lceil \frac{1}{\Delta} \rceil$ sorted list of vertices, where $0 < \Delta < 1$ is a {\em constant}. 
Specifically, 
the parameter $\Delta$ partitions value range of $\eps$ into $\lceil \frac{1}{\Delta} \rceil$ intervals, where the $i^\text{th}$ interval is $[i \Delta, (i + 1) \Delta)$ for $i = 0, \cdots, \lceil \frac{1}{\Delta}\rceil-1$.
$T_{\Delta}[i]$ 
is a sorted list of all vertices $u \in V$ in {\em non-increasing} order by
$\mu_{u,i}$, where $\mu_{u,i}$ is 
 the number of neighbors of $u$ have similarities to $u$ at least $i \Delta$.
\begin{itemize}[leftmargin = *]
\item find-core($\eps$, $\mu$): identify $i^* = \lfloor \eps / \Delta \rfloor$; retrieve all the core vertices by scanning and reporting the vertices in the sorted list $T_{\Delta}[i^*]$ until the first vertex $x$ such  
that $\mu_{x,i^*} > \mu$ or the entire list has been retrieved.
Thus, $\costcf \in O(|\Vc| + 1)$.
Note that $\Delta$-Table introduces an additive $\Delta$ error to the overall approximation.

\item update($x$): maintain the $\lceil \frac{1}{\Delta} \rceil \in O(1)$ sorted lists in $T_{\Delta}$ for $x$. 
This takes $\costcu \in O(\lceil \frac{1}{\Delta} \rceil \cdot \log n) = O(\log n)$ time.

\end{itemize}

\noindent\underline{A Running Example.} Figure~\ref{fig:botbin} shows a running example of $\Delta$-Table with $\Delta = 0.2$. Take $[0.2,0.4)$ for instance, where $\mu_{v_1,1} = 5$ meaning that $v_1$ has 5 similar neighbors when $0.2 \leq \eps < 0.4$. Given a query with $\eps = 0.3$ and $\mu = 4$, we know that $v_1$ and $v_3$ are core vertices as only $\mu_{v_1,1}$ and $\mu_{v_3,1}$ are greater than or equal to $4$ in the corresponding sorted linked list.

\noindent
{\bf Theoretical Analysis.}
\cite{zhang2022effective} showed that 
by setting $k \in O(\frac{1}{\rho^2} \cdot \log (n\cdot M) ) = O(\log (n \cdot M))$ and $\Delta = \frac{1}{2}\rho$, BOTBIN guarantees to return a valid $\rho$-absolute-approximate clustering result with high probability, specially at least $1 - \frac{1}{n}$, for every query. This guarantee holds for up to $M$ updates.   

\noindent
\underline{\em Query Running Time.} Since $\costcf$ in bounded by $O(|\Vc| +1)$, the running time of each query is bounded by $O(m_{cr})$. 

\noindent
\underline{\em Per-Update Running Time.}
Substitute the running time cost of each function in the above implementation to Expression~\eqref{eq:query-bound},
the per-update running time of BOTBIN is thus bounded by $O(|F| \cdot \log (n\cdot M)) \subseteq O(d_{\max} \cdot \log (n\cdot M))$.
While this per-update bound is still prohibitive, \cite{zhang2022effective} proved that, as long as the $M$ updates happen {\em uniformly at random} in the neighborhood of each vertex,
then the signature of a vertex $u$ changes with probability $\frac{k}{d_u}$.
Therefore, the expected size of the ``invalid'' affected edge set $F$ is bounded by $O(\frac{k}{d_u} \cdot d_u) = O(\log (n\cdot M))$.
As a result, the expected per-update cost is bounded by $O(\log^2 (n\cdot M))$. 

\vspace{1mm}
\noindent
\underline{\em Space Consumption.} It can be verified that the space consumption of BOTBIN is bounded by $O(n + m)$ at all times.
\vspace{-1mm}
\begin{fact}
BOTBIN can return a valid $\rho$-absolute-approximate clustering result (under Jaccard similarity {\em only}) in $O(m_{cr})$ time with high probability, at least $1 - \frac{1}{n}$, for each query, and this holds for up to $M$ updates.
Furthermore, it can handle each update in $O(d_{\max} \cdot \log (n\cdot M))$ time and the space consumption is bounded by $O(n + m)$ at all times. 

When the $M$ updates occur uniformly at random in the neighborhood of each vertex, then the per-update time is bounded by $O(\log^2 (n \cdot M))$ in {\em expectation}.  
\end{fact}
\vspace{-3mm}

\noindent
{\bf Remark.}
As discussed earlier, BOTBIN has two main limitations:
\begin{itemize}[leftmargin = *]
\item BOTBIN can handle Jaccard similarity only. 

\item The $O(\log^2 (n\cdot M))$ per-update expected running time bound of BOTBIN holds only for random updates. 
As a result, for repeated insertion and deletion of a {\em critical} edge which changes the signatures of its two end-vertices, BOTBIN has to pay $O(d_{\max}\cdot \log (n\cdot M))$ cost for each such update.

\end{itemize}

\section{Our Versatile DynStrClu Algorithm}
\label{sec:our-algo}

Next, we introduce our solution, called {\em \underline{V}ersatile \underline{D}ynamic \underline{St}ructur\underline{a}l Cluste\underline{r}ing} ({\em VD-STAR}), which not only overcomes {\em all} the aforementioned limitations of BOTBIN, but also improves the per-update running time cost to $O(\log n + \log M)$ amortized in expectation. 

Let $n_0$ and $m_0$ be the number of vertices and edges in the graph at the current moment. 
Without loss of generality, we assume that the number of updates $M \leq n_0^2$ since now, because, otherwise, when $M = n_0^2$, we can rebuild everything from scratch 
in $O((n_0 + m_0 + M) \cdot  \log n) = O(M \cdot \log n)$ expected time.
Hence, each of such $M$ updates is charged a $O(\log n)$ amortized expected cost which does not affect the per-update running time bound.   
Moreover, 
it is worth mentioning that the randomness in the running time of {\em VD-STAR} only comes from the use of hash tables.
In this and the next section, we prove this theorem:
\vspace{-1mm}
\begin{theorem}\label{thm:ver-dynstrclu}
Our {\em VD-STAR} algorithm supports all three similarity measurements (Jaccard, Cosine, Dice).
It can return an $\rho$-absolute-approximate clustering result with high probability, at least $1 - \frac{1}{n}$, for each query, and can handle each update in $O(\log n)$ amortized expected time.  
The space consumption of {\em VD-STAR} is bounded by $O(n + m)$ at all times.
\end{theorem}
\vspace{-1mm}

Our {\em VD-STAR} also works under the unified algorithm framework (Algorithm~\ref{algo:framework}). 
%
Specifically, the implementation of {\em VD-STAR} for {\em CoreFindStr} follows that of BOTBIN, 
i.e., the $\Delta$-Table.
Therefore, we will focus on our implementation for {\em EdgeSimStr}.

\vspace{-2mm}
\subsection{Update Affordability and Background}
\label{subsec:update_afford}

We adopt the notion of $\rho$-absolute-approximation.
Thanks to the approximation, 
{\em VD-STAR} is allowed to just maintain {\em approximate} rather than exact similarities. 
It thus creates room for efficiency improvements.
First, the similarity can now be estimated (within an $\rho$-absolute error) via certain sampling techniques efficiently.
Second, each edge can now afford a certain number of affecting updates before its estimated similarity exceeds the $\rho$-absolute-error range from the last estimation.
Such a number of affecting updates is called the {\em update affordability} of the edge. 
\vspace{-1mm}
\begin{definition}[Update Affordability]
\label{def:update-affordability}
For any edge $(u,v)$, 
consider the moment when an $\frac{1}{2} \rho$-absolute-approximate similarity $\tilde{\sigma}(u,v)$ is just computed;
the {\em update affordability} of $(u,v)$ is 
a {\em lower bound} on the number of 
affecting updates, denoted by $\tau(u,v)$,
such that 
$\tilde{\sigma}(u,v)$ remains a {\em valid} $\rho$-absolute approximation to the exact similarity,
in the sense that $|\tilde{\sigma}(u,v) - \sigma(u,v)| \leq \rho$,
at any moment within $\tau(u,v)$ affecting updates. 
\end{definition}
\vspace{-1mm}
The concept of update affordability was first proposed and 
exploited in~\cite{ruan2021dynamic} for solving the {\em DynStrClu} problem with {\em pre-specified} parameters $\eps$ and $\mu$. 
We borrow this concept and extend it to work for {\em versatile} similarity measurements (Jaccard, Cosine, and Dice).
As we show in Section~\ref{sec:update-affordability}, 
the extension of the concept of update affordability from Jaccard to Cosine similarity is challenging and non-trivial, mainly because of the {\em non-linear} denominator $\sqrt{n_u\cdot n_v}$ in Cosine similarity.
%
To avoid distraction,
we defer the proof of the following claim to Section~\ref{sec:update-affordability}.
\vspace{-1mm}
\begin{claim}
\label{claim:update-affordability}
For any edge $(u,v)$ with $n_u \leq n_v$, 
the update affordability $\tau(u,v) \geq \frac{1}{4} \rho^2 n_v \in \Omega(d_{\max}(u,v))$ holds for any of Jaccard, Cosine and Dice similarity measurements,
where $d_{\max}(u,v) = \max\{d_u, d_v\}$.
\end{claim}
\vspace{-1mm}
By the definition of update affordability,
when an approximate similarity $\tilde{\sigma}(u,v)$ of an edge $(u,v)$ is just computed,
$\tilde{\sigma}(u,v)$ will remain valid for the next at least $\tau(u,v)-1$ affecting updates.
And hence, 
in order to ensure a valid $\rho$-absolute-approximate similarity for every edge, 
one may need to {\em re-compute} $\tilde{\sigma}(u,v)$ {\em no later} than the arrival of the $\tau(u,v)^\text{th}$ affecting update for each $(u,v) \in E$. 

However, observe that (i) the update affordability can be different for different edges, and 
(ii) an update of edge $(u,v)$
would ``consume'' one affordability for each of its $d_u + d_v \in O(d_{\max})$ affected edges.
Therefore, simply tracking the ``remaining'' update affordability for each affected edge 
can be as expensive as $\Omega(d_{\max})$.
It is challenging to identify the set $F$ of all invalid edges when an update arrives, without touching each of the affected edges.

Ruan~\shortauthors~\cite{ruan2021dynamic} overcome this technical challenge by adopting the {\em Distributed Tracking} technique~\cite{cormode2011algorithms,huang2012randomized,keralapura2006communication} to track the {\em exact moment} when the $\tau(u,v)^\text{th}$ affecting update for each edge occurs. 
They proved that their algorithm can achieve an $O(\log^2 n)$ amortized time for 
processing each update.
Next, we show a {\em simpler yet more efficient} solution 
to identify
the invalid edge set $F$ just in $O(1)$ amortized expected time for each update.
\vspace{-2mm}
\subsection{Our Implementation of {\em EdgeSimStr}}
\label{subsec:edgesim}

\noindent
{\bf Rationale of Our Algorithm.}
The basic idea of our solution for identifying invalid edges is as follows.
For each edge $(u,v) \in E$, once $\tilde{\sigma}(u,v)$ is just computed, we compute 
its update affordability $\tau(u,v)$.
Instead of tracking the exact moment when the $\tau(u,v)^\text{th}$ affecting update arrives,
our algorithm aims to just identify an {\em arbitrary} moment when there have been at least $\frac{1}{4} \lfloor \tau(u,v) \rfloor_2$ affecting updates, where $\lfloor \tau(u,v) \rfloor_2$ is the {\em largest power-of-two} integer that is no more than $\tau(u,v)$, namely, $\lfloor \tau(u,v) \rfloor_2 = 2^{\lfloor \log_2 \tau(u,v) \rfloor}$. 
And such a moment is called a {\em checkpoint moment} of edge $(u,v)$.
Clearly, $\lfloor \tau(u,v) \rfloor_2 \geq \frac{1}{2} \tau(u,v)$.
When a checkpoint moment of $(u,v)$ is identified, there must have been at least $\frac{1}{8} \tau(u,v) \in \Omega(\tau(u,v))$ affecting updates, which
is already ``good enough'' for our theoretical analysis.

To capture the checkpoint moments, 
our algorithm, for each edge $(u,v)$, allocates an {\em affordability quota}, denoted by $q(u,v) = \frac{1}{4}\lfloor \tau(u,v) \rfloor_2$, to the vertices $u$ and $v$.
Once 
an arbitrary moment when at least $q(u,v)$ affecting updates incident on either $u$ or $v$ are ``observed'' since the quota is allocated,
$(u,v)$ is then reported as an invalid edge.
The challenge is how to capture a checkpoint moment for each edge $(u,v)$ before its $\rho$-absolute-approximate $\tilde{\sigma}(u,v)$ becomes invalid, without touching each of the affected edges for every affecting update. 

\noindent
{\bf The Data Structure for {\em EdgeSimStr}.}
For each $u\in V$, {\em VD-STAR} maintains the following information for {\em EdgeSimStr}:
\vspace{-1mm}
\begin{itemize}[leftmargin = *]
\item a counter $c_u$ that records the number of affecting updates incident on $u$ up to date;
initially, $c_u \leftarrow 0$;
\item a {\em sorted bucket linked list} $\mathcal{B}(u)$, where:
\begin{itemize}
\item each bucket $B_i$ has a {\em unique index} $i$ (for $0 \leq i \leq \lceil \log_2 n \rceil$);
\item bucket $B_i$ stores all the neighbors $w \in N(u)$ such that the affordability quota  $q(u, w) = 2^i$;
\item all the {\em non-empty} buckets $B_i$ (which contain at least one neighbor $w \in N(u)$) are materialized in the sorted linked list $\mathcal{B}(u)$ in an {\em increasing order} by their indices $i$. 
\item each non-empty bucket $B_i \in \mathcal{B}(u)$ maintains a counter $\bar{c}_u(B_i)$ that records the counter value $c_u$ of $u$ when $B_i$ is last visited; initially, $\bar{c}_u(B_i)$ is set as the value of $c_u$ when $B_i$ is materialized and added to $\mathcal{B}(u)$; 
\end{itemize} 
\end{itemize}
\vspace{-1mm}

\noindent
{\bf Implementation of {\em EdgeSimStr}.update.}
In our {\em VD-STAR}, the update function of {\em EdgeSimStr} just increases
the counters $c_u$ and $c_v$ by one, respectively, i.e.,
$c_u \leftarrow c_u +1$ and $c_v \leftarrow c_v + 1$, 
recording that there is one more affecting update on them.

\vspace{1mm}
\noindent
{\bf Implementation of {\em EdgeSimStr}.insert.}
The detailed implementation is shown in Algorithm~\ref{algo:insert}.
To insert an edge $(u,v)$ to {\em EdgeSimStr}, 
our algorithm first computes the affordability quota $q(u,v)$,
and then inserts $v$ (resp., $u$) into the corresponding bucket in $\mathcal{B}(u)$ (resp., $\mathcal{B}(v)$).
If the bucket does not exist, then a bucket is created and inserted into the sorted bucket linked list accordingly.

\begin{small}
\begin{algorithm}[t]
\caption{Our Implementation of {\em EdgeSimStr}.insert}\label{algo:insert}
\SetKwComment{Comment}{/* }{ */}

\KwIn{an insertion of edge $(u,v)$ to {\em EdgeSimStr}}  
$\tau(u,v) \leftarrow \frac{1}{4} \rho^2 \max\{n_u, n_v\}$\; 
$q(u,v) \leftarrow \frac{1}{4} \cdot \lfloor \tau(u,v) \rfloor_2$\;
$i \leftarrow \log_2 (q(u,v))$\;
\If{$B_i$ does not exist in $\mathcal{B}(u)$}{
	create $B_i$ and insert $B_i$ to $\mathcal{B}(u)$\;
	set $\bar{c}_u(B_i) \leftarrow c_u$\;
}
insert $v$ to $B_i$\;
perform the above steps for $v$ symmetrically\; 
\end{algorithm}
\end{small}

\vspace{1mm}
\noindent
{\bf Implementation of {\em EdgeSimStr}.delete.}
This function removes $v$ (resp., $u$) from its corresponding bucket in $\mathcal{B}(u)$ (resp., $\mathcal{B}(v)$). If the bucket becomes empty, then the bucket is removed from the bucket list.
The pseudo code is shown in Algorithm~\ref{algo:delete}.

\begin{small}
\begin{algorithm}[t]
\caption{Our Implementation of {\em EdgeSimStr}.delete}\label{algo:delete}
\SetKwComment{Comment}{// }{}

\KwIn{a deletion of edge $(u,v)$ from {\em EdgeSimStr}} 

remove $v$ from its corresponding bucket $B_i$\;
\If {$B_i$ becomes empty}
{remove $B_i$ from $\mathcal{B}(u)$\;
}

perform the above steps for $v$ symmetrically\; 

\end{algorithm}
\end{small}

\vspace{1mm}
\noindent
{\bf Implementation of {\em EdgeSimStr}.find.}
Algorithm~\ref{algo:find} gives 
implementation details.
Observe that all the neighbors $w$ of $u$ are stored in a sorted list of {\em power-of-two} buckets of their corresponding affordability quotas.
When the counter $c_u$ is increased by one, 
it suffices to scan the sorted bucket list to check all the non-empty buckets $B_i$ such that 
the current $c_u$ has passed across their corresponding power-of-two values $2^i$, 
because they were last visited when $c_u = \bar{c}_u(B_i)$ (see Line~4 in Algorithm~\ref{algo:find}).
For each of such buckets $B_i$,
our algorithm reports and adds the edge $(u,w)$ to $F$ for each $w \in B_i$ such that 
$w$ is visited in $B_i$ for the {\em second time}.
The same process is performed for $v$ symmetrically. 
The correctness of this implementation is proved in Section~\ref{sec:correctness}.

\begin{small}
\begin{algorithm}[t]
\caption{Our Implementation of {\em EdgeSimStr}.find}\label{algo:find}
\SetKwComment{Comment}{// }{}
        
\KwIn{either an insertion or a deletion of edge $(u,v)$} 
\KwOut{a set $F$ of potentially invalid edges}

$F \leftarrow \emptyset$\;
$B_i \leftarrow$ the first bucket in $\mathcal{B}(u)$, where $i$ is the index of $B$\;
\While{$B_i$ is not NULL}{
\If{$\lfloor \frac{c_u}{2^i} \rfloor > \lfloor \frac{\bar{c}_u(B_i)}{2^i} \rfloor$}{
\Comment{check this bucket $B_i$}
\For{each $w \in B_i$}{
\If{$w$ is visited in $B_i$ for the second time}{
add $(u,w)$ to $F$\;
}
\Else{
flag $w$ as it has been visited for once\;
}
}
$\bar{c}_u(B_i) \leftarrow c_u$;
$B_i \leftarrow B_i.\text{next}$\;
}
\Else{
stop the scan of $\mathcal{B}(u)$ and break\;
}
}

perform the steps from Line 2 for $v$ symmetrically\; 
\KwRet{$F$ as the set of invalid edges}

\end{algorithm}
\end{small}

\begin{small}
\begin{algorithm}[t]
\caption{Our Implementation of {\em EdgeSimStr}.cal-sim}\label{algo:calsim}
\SetKwComment{Comment}{// }{}

\KwIn{an edge $(x,y)$} 
\KwOut{an $\frac{1}{2} \rho$-absolute-approximate similarity $\tilde{\sigma}(x,y)$}

\If{$n_x \leq \frac{1}{4} \rho^2 n_y$ or $n_y \leq \frac{1}{4} \rho^2 n_x$}
{
\KwRet{$\tilde{\sigma}(x,y) = 0$\;}
}
$L\leftarrow $ the number of samples as required by Lemma~\ref{lmm:correct-sim}\; 
$X \leftarrow 0$\;

\For{$ i = 1, 2, \cdots, L$}{
flip a coin $z$ such that $\text{Pr}[z = 1] = \frac{n_x}{n_x + n_y}$ and $\text{Pr}[z = 0] = \frac{n_x}{n_x + n_y}$\;

\If{z = 1}{
uniformly at random pick a vertex $w \in N[x]$\;
}
\Else{
uniformly at random pick a vertex $w \in N[y]$\;
}
\If{$w \in N[x] \cap N[y]$}{ $X \leftarrow X + 1$\;}
}

$\bar{X} \leftarrow X / L$\;
\text{\bf return } 
\begin{small}
\begin{align*}
\tilde{\sigma}(x,y)=
\begin{cases}
\frac{\bar{X}}{2 - \bar{X}} & \text{for Jaccard similarity} \vspace{1mm}\\
\frac{n_x + n_y}{2 \sqrt{n_x \cdot n_y}} \cdot \bar{X}  & \text{for Cosine similarity} \vspace{1mm}\\
\bar{X} & \text{for Dice similarity} 
\end{cases}
\end{align*}
\end{small}
\end{algorithm}
\end{small}

\vspace{1mm}
\noindent
{\bf Implementation of {\em EdgeSimStr}.cal-sim.}
See Algorithm~\ref{algo:calsim}.

\begin{figure}
    \centering
    \includegraphics[width=0.85\linewidth]{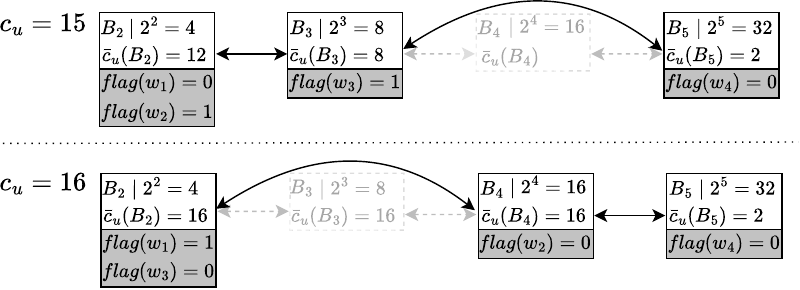}
	\vspace{-2mm}
    \caption{A Running Example of Our {\em EdgeSimStr}}
    \label{fig:running_example}
\end{figure}

\vspace{1mm}
\noindent
{\bf A Running Example.}
Figure~\ref{fig:running_example}
shows a running example of the maintenance of our implementation for {\em EdgeSimStr}. 
At the current status, the affecting update counter of $u$, $c_u = 15$ and there are three non-empty buckets $B_2$, $B_3$ and $B_5$ in the sorted bucket list $\mathcal{B}(u)$, 
where $\bar{c}_u(B_2) = 12$ indicates that when $B_2$ was  
last visited, the value of the counter $c_u$ was $12$.
Moreover, there are two neighbors $w_1$ and $w_2$ in $B_2$, where $w_1$ has not yet been visited while $w_2$ has been visited for once.
When an update incident on $u$ occurs,
$c_u$ is increased by one to $c_u = 16$ and then 
{\em EdgeSimStr}.find (Algorithm~\ref{algo:find}) is invoked and it scans $\mathcal{B}(u)$ from the first bucket $B_2$. 
Since $\lfloor\frac{16}{4}\rfloor > \lfloor\frac{12}{4}\rfloor$ (Line 4), the contents of $B_2$ are checked, where the flag of $w_1$ is set to $1$ indicating that now $w_1$ has been visited once, while the edge $(u, w_2)$ is added to $F$ because $w_2$ is visited for the second time now.
Finally, $\bar{c}_u(B_2) \leftarrow 16$ recording that the ``time'' when $B_2$ was last visited.
This completes the process for $B_2$.
As $\lfloor\frac{16}{8}\rfloor > \lfloor\frac{8}{8}\rfloor$, similarly, $(u, w_3)$ is added to $F$ and $\Bar{c}_u(B_3) \leftarrow 16$. 
The algorithm stops the scanning at $B_5$ because $\lfloor\frac{16}{32}\rfloor = \lfloor\frac{2}{32}\rfloor$. Next,
{\em EdgeSimStr}.delete (Algorithm~\ref{algo:delete}) is invoked for $(u, w_2)$ and $(u, w_3)$ in $F$, and it removes $w_1$ and $w_2$ from $B_2$ and $B_3$, respectively (Line 1). 
As $B_3$ becomes empty, it is then removed from $\mathcal{B}(u)$ (Lines 2-3). 
After the re-calculation of the similarities for $(u, w_2)$ and $(u, w_3)$ with 
{\em EdgeSimStr}.cal-sim (Algorithm~\ref{algo:calsim}), 
{\em  EdgeSimStr}.insert (Algorithm~\ref{algo:insert}) is invoked to insert $w_2$ and $w_3$ to buckets $B_4$ and $B_2$ in $\mathcal{B}(u)$, respectively.

\section{Theoretical Analysis}
\label{sec:analysis}

In this section, we prove the correctness of {\em VD-STAR}, analyze 
the amortized per-update running time and the space consumption.
Putting these results together constitutes a proof for Theorem~\ref{thm:ver-dynstrclu}.

\vspace{-2mm}
\subsection{Correctness}
\label{sec:correctness}
\begin{theorem}\label{thm:correct-sim}
Before and after any update, {\em VD-STAR} maintains a {\em proper} $\rho$-absolute-approximate
similarity $\tilde{\sigma}(u,v)$ for every edge $(u,v) \in E$ with high probability at least $1 - \frac{1}{n}$.
\end{theorem}
\noindent
To prove Theorem~\ref{thm:correct-sim}, it suffices to show these two lemmas:
\begin{lemma}\label{lmm:correct-sim}
By setting $L = \frac{1}{2r^2}\ln (4n^4)$, 
where $r = \frac{1}{4} \rho$ for Jaccard, $r = \frac{1}{4} \rho^2$ for Cosine and $r = \frac{1}{2} \rho$ for Dice similarity,
the approximate similarity $\tilde{\sigma}(u,v)$ returned by Algorithm~\ref{algo:calsim} satisfies  
$|\tilde{\sigma}(u,v) - \sigma(u,v)|\leq \frac{1}{2} \rho$ with probability at least $1 - \frac{1}{2n^4}$.
\end{lemma} 

\begin{lemma}\label{lmm:correct-check}
For any $(u,v) \in E$, its approximate similarity $\tilde{\sigma}(u,v)$ must be recomputed by Algorithm~\ref{algo:calsim} before its update affordability $\tau(u,v)$ is fully consumed, that is, before the arrival of its $\tau(u,v)^\text{th}$ affecting update, since $\tilde{\sigma}(u,v)$ was last computed. 
\end{lemma}
\vspace{-1mm}
\noindent
{\bf Proof of Theorem~\ref{thm:correct-sim}.}
Suppose 
Lemmas~\ref{lmm:correct-sim} and~\ref{lmm:correct-check} hold;
by the definition of update affordability, 
we have that $\tilde{\sigma}(u,v)$ is a correct $\rho$-absolute-approximation of $\sigma(u,v)$ before and after any update for all edges $(u,v) \in E$. 
To see the success probability, as each update can affect at most $2n$ edges, it can trigger at most $2n$ invocations of Algorithm~\ref{algo:calsim}. 
Moreover, there are at most $M \leq n^2$ updates. Therefore, Algorithm~\ref{algo:calsim} is invoked for at most $2n^3$ times. 
According to Lemma~\ref{lmm:correct-sim}, each invocation fails with probability at most $\frac{1}{2n^4}$.
Thus, the whole process succeeds with probability at least $1 - \frac{1}{n}$.
\vspace{-1.5em}
\myqed

%
%


\noindent
{\bf Proof of Lemma~\ref{lmm:correct-sim}.}
%
Consider an edge $(u, v)\in E$; without loss of generality, we assume that $n_u \leq n_v$. 
\vspace{-1mm}
\begin{observation}\label{claim:degree-ratio}
For any edge $(u, v)$ with 
$n_u = \beta \cdot n_v$, where $0 <  \beta \leq 1$, 
we have:
\begin{itemize}[leftmargin = *] 
\item $Jaccard(u,v) = \frac{I(u,v)}{n_u + n_v - I(u,v)} \leq \frac{\beta\cdot  n_v}{n_v}  = \beta$;
\item $Cosine(u,v) = \frac{I(u,v)}{\sqrt{n_u \cdot n_v}} \leq \frac{n_u}{\sqrt{1/\beta}\cdot n_u} = \sqrt{\beta}$; 
\item $Dice(u,v) = \frac{I(u,v)}{(n_u + n_v)/2} \leq \frac{\beta \cdot n_v}{ n_v / 2} = 2 \beta$.
\end{itemize}
\end{observation}
\vspace{-1mm}
Substituting $\beta = \frac{1}{4} \rho^2$ to Observation~\ref{claim:degree-ratio}, that is, $n_u \leq \frac{1}{4}\rho^2 n_v$, 
then the Jaccard, Cosine and Dice similarities of $(u,v)$ are all $\leq \frac{1}{2} \rho$ for any constant $0 \leq \rho \leq 1$. 
Therefore, for any of the above similarity measurements,
$\tilde{\sigma}(u,v) = 0$ is a correct $\frac{1}{2} \rho$-absolute-approximate similarity,
thus, Lines 1-2 in Algorithm~\ref{algo:calsim} are correct. 

Next, we consider the case that $n_v \geq n_u > \frac{1}{4}\rho^2 n_v$ holds. 
In fact, a proof of this lemma for Jaccard similarity is given in~\cite{ruan2021dynamic} by Ruan~\shortauthors. 
We extend their proof to Cosine and Dice similarity. For completeness, 
we prove all of them in the following.

Let $X_i \in \{0, 1\}$ be a random variable for the $i^\text{th}$ iteration in the for-loop in Lines 5 - 12 in Algorithm~\ref{algo:calsim}. Specifically,
$X_i = 1$ if $X$ is increased by one in Line 12; otherwise, $X_i = 0$.
Therefore, 
$\text{Pr}[X_i = 1] = \text{Pr}[X_i = 1 \land z= 1] + \text{Pr}[X_i = 0 \land z = 0] = \frac{n_u}{n_u + n_v} \cdot \frac{I(u,v)}{n_u} + \frac{n_v}{n_u + n_v} \cdot \frac{I(u,v)}{n_v} = \frac{2 I(u,v)}{n_u + n_v}$.
Furthermore, 
since 
$\bar{X} = \frac{X}{L} = \frac{\sum_{i = 1}^L X_i}{L}$, 
we have the expectation $E[\bar{X}] = E[X_i] = \frac{2\cdot I(u,v)}{n_u + n_v}$.
Thus, we have the following for each of the similarity measurements.

\vspace{1mm}
\noindent
\underline{For Jaccard similarity}, 
we have $Jaccard(u,v) = \frac{E[\bar{X}]}{2 - E[\bar{X}]}$, and by Line~15
, $\tilde{\sigma}(u,v) = \frac{\bar{X}}{2 - \bar{X}}$.
Thus,
$\text{Pr}[| \tilde{\sigma}(u,v) - Jaccard(u,v) | > \frac{1}{2} \rho]  = \text{Pr}[\frac{2 \cdot |\bar{X} - E[\bar{X}]|}{(2 - \bar{X}) (2 - E[\bar{X}])} > \frac{1}{2} \rho ] \leq \text{Pr}[|\bar{X} - E[\bar{X}]| > \frac{1}{4} \rho]$, 
where the last inequality is by both $\bar{X}$ and $E[\bar{X}]$ are values in $[0,1]$. 

\vspace{1mm}
\noindent
\underline{For Cosine similarity}, 
we have $Cosine(u,v) = \frac{n_u + n_v}{2 \sqrt{n_u \cdot n_v}} \cdot E[\bar{X}]$, and by Line~16,
$\tilde{\sigma}(u,v) = \frac{n_u + n_v}{2 \sqrt{n_u \cdot n_v}} \cdot \bar{X}$.
Thus, 
$\text{Pr}[| \tilde{\sigma}(u,v) - Cosine(u,v) | > \frac{1}{2} \rho]  
= \text{Pr}[\frac{n_u + n_v}{2 \sqrt{n_u \cdot n_v}} \cdot |\bar{X} - E[\bar{X}]| >\frac{1}{2} \rho ]  = \text{Pr}[|\bar{X} - E[\bar{X}]| > \frac{2 \sqrt{n_u \cdot n_v}}{n_u + n_v} \cdot \frac{1}{2} \rho] \leq \text{Pr}[|\bar{X} - E[\bar{X}]|> \frac{2 \sqrt{1/4 \rho^2 n_v \cdot n_v}}{2 \cdot n_v} \cdot \frac{1}{2} \rho] = \text{Pr}[|\bar{X} - E[\bar{X}]|> \frac{1}{4} \rho^2] $,
where the last inequality is by the fact that $n_v \geq n_u > \frac{1}{4} \rho^2 n_v$.

\vspace{1mm}
\noindent
\underline{For Dice similarity},
we have $Dice(u,v) = E[\bar{X}]$ and by Line~17, $\tilde{\sigma}(u,v) = \bar{X}$.
Therefore, $\text{Pr}[|\tilde{\sigma}(u,v) - Dice(u,v)| > \frac{1}{2} \rho ] = \text{Pr}[|\bar{X} - E[\bar{X}]| > \frac{1}{2} \rho]$.

According to the Hoeffiding Bound~\cite{hoeffding1994probability}, 
by setting $L = \frac{1}{2 \cdot r^2} \ln \frac{2}{\delta}$, 
we have $\text{Pr}[|\bar{X} - E[\bar{X}] > r] \leq \delta$.
As a result, by setting 
$\delta = \frac{1}{2n^4}$, 
$r_j= \frac{1}{4} \rho$, 
$r_c = \frac{1}{4} \rho^2$, and $r_d = \frac{1}{2} \rho$, respectively for Jaccard, Cosine and Dice similarities, 
we can get the corresponding number of samples $L$ to achieve  
$\tilde{\sigma}(u,v)$ being a correct $\frac{1}{2} \rho$-absolute approximation to $\sigma(u,v)$ with high probability at least $1 - \frac{1}{2n^4}$.
%
\vspace{-1.5em}
\myqed

\noindent
{\bf Proof of Lemma~\ref{lmm:correct-check}.}
Recall that for each edge $(u,v)$ right after $\tilde{\sigma}(u,v)$ is computed, 
$(u,v)$ allocates an affordability quota $q(u,v) = \frac{1}{4} \lfloor \tau(u,v) \rfloor_2$ to an entry in a bucket $B_i$ with index $i = \log_2 q(u,v)$ in both the sorted linked bucket lists $\mathcal{B}(u)$ and $\mathcal{B}(v)$.
According to Algorithm~\ref{algo:find}, $(u,v)$ is reported
as an invalid edge in $F$ when the entry in either the bucket in $\mathcal{B}(u)$ or $\mathcal{B}(v)$ 
is visited for the second time.
As a result, the entry of edge $(u,v)$ can be checked for at most three times in total, because at that time, the entry in either bucket must be checked for twice. 
Moreover, since each checking of the bucket $B_i$ is triggered by at most $q(u,v)$ affecting updates,
there can be at most $3 \cdot q(u,v) + q(u,v) - 1 < \lfloor \tau(u,v) \rfloor_2 \leq \tau(u,v)$ affecting updates happened. 
Lemma~\ref{lmm:correct-check} thus follows. 
\vspace{-1.5em}
\myqed

%

By Theorem~\ref{thm:correct-sim} and the fact that {\em VD-STAR} adopts the $\Delta$-Table for {\em CoreFindStr}, 
 this theorem immediately follows, which completes the correctness proof for {\em VD-STAR}.
\vspace{-1mm}
\begin{theorem}
{\em VD-STAR} returns a $(\rho + \Delta)$-absolute-approximate clustering result, with high probability at least $1 - \frac{1}{n}$, 
for any query with respect to the given parameters $\eps$ and $\mu$. 
\end{theorem}

\noindent
{\bf Remark.}
Given any constant {\em target overall approximation} parameter $\rho^*$, by setting $\rho = \Delta = \frac{1}{2} \rho^*$, 
{\em VD-STAR} can achieve $\rho^*$-absolute approximation
without affecting its theoretical bounds.

\vspace{-1mm}
\subsection{Running Time Analysis}
\label{sec:amortized}

\noindent
{\bf Query Running Time.}
As our {\em VD-STAR} adopts the $\Delta$-Table technique for {\em CoreFindStr}, the query running time bound follows immediately from the analysis in Section~\ref{sec:preli}. Thus, we have:
\begin{lemma}
{\em VD-STAR} can answer each query in $O(m_{cr})$ time.
\end{lemma}
\vspace{-1mm}
\noindent
{\bf The Maintenance Cost of an Edge.}
We first analyze the {\em maintenance cost} of each edge $(u,v)$, denoted by $\ell(u,v)$,  between two {\em consecutive} approximation similarity calculations for $(u,v)$. 
Consider the moment when the similarity of an edge $(u,v)$ 
needs to be computed;
according to Algorithm~\ref{algo:framework}, 
the maintenance for $(u,v)$ involves the following operations:
\begin{itemize}[leftmargin = *]
\item a similarity calculation (Algorithm~\ref{algo:calsim}) which takes $\costec$;
\item an invocation of Algorithm~\ref{algo:delete}
to remove the ``old'' quota entries of $(u,v)$ from the buckets in $\mathcal{B}(u)$ and $\mathcal{B}(v)$; this takes $\costed$;
\item an invocation of Algorithm~\ref{algo:insert} to insert the ``updated'' quota entries of $(u,v)$ to buckets in $\mathcal{B}(u)$ and $\mathcal{B}(v)$; this takes $\costei$; 
\item the maintenance of the sorted neighbor lists of $u$ and $v$ due to the change of $\tilde{\sigma}(u,v)$;
this maintenance takes $O(\log n)$ time;
\item the maintenance of $\Delta$-Table for $u$ and $v$ due to the change of their sorted neighbor lists;  
as discussed in Section~\ref{sec:preli}, this cost is bounded by $O(\frac{1}{\Delta} \cdot \log n) = O(\log n)$ since $\Delta$ is a constant;
\item at most three times of visits of the entries of $(u,v)$ in the corresponding buckets before getting reported as an invalid edge; this cost is just $O(1)$. 
\end{itemize}
Summing these costs up, the maintenance cost of $(u,v)$ is:
\begin{equation}\vspace{-1mm}
\label{eq:maintenance-cost}
\ell(u,v) \in O(\costec + \costed + \costei + \log n)\,.
\end{equation}
\vspace{-1mm}
Next, we analyse $\costec$, $\costed$ and $\costei$, respectively. For the cost of similarity calculation, $\costec$, by Algorithm~\ref{algo:calsim},
by Lemma~\ref{lmm:correct-sim}, we know that $L \in O(\log n)$ samples suffice.
Each sample checks if a neighbor $w$ is in $N[u] \cap N[v]$. 
By maintaining a hash table of $N[u]$ and $N[v]$, each of this checking can be performed in $O(1)$ expected time. Therefore, $\costec$ is bounded by $O(L) = O(\log n)$ in expectation. 

To bound the costs $\costei$ and $\costed$ of Algorithms~\ref{algo:insert} and~\ref{algo:delete}, observe that 
inserting and removing an entry from a bucket can be done in $O(1)$ time. 
This can be achieved simply by recording the locations (e.g., the indices in arrays) of the entries in the corresponding buckets. 
The remaining cost are from the operations on the sorted bucket lists $\mathcal{B}(u)$ and $\mathcal{B}(v)$ which include: (i) 
checking if a bucket exists or not, 
(ii) inserting a new bucket,
and (iii) removing an existing bucket.
According to the following Fact~\ref{fact:sorted-linked-list}, each of this operation can be performed in $O(1)$ expected time. And therefore, $\costei + \costed$ is bounded by $O(1)$ in expectation.

\begin{fact}[\cite{gan2024optimal}]\label{fact:sorted-linked-list}
The sorted linked list $\mathcal{B}(u)$ can be maintained with $O(|\mathcal{B}(u)|)$ space and support the following in $O(1)$ expected time: 
\begin{itemize}[leftmargin = *]
\item an insertion or deletion of a bucket to or from $\mathcal{B}(u)$, and  
\item return the pointer of 
the largest bucket $B_i \in \mathcal{B}(u)$ with index $i \leq j$, for 
any given integer index $0 \leq j \leq \lceil \log_2 n \rceil$. 
\end{itemize}
\end{fact}

\noindent
Putting all the above cost bounds to Expression~\eqref{eq:maintenance-cost}, 
we thus have:
\begin{lemma}\label{lmm:maintenance-cost}
The maintenance cost of each edge $(u,v)$ between two consecutive 
similarity calculations of it, 
$\ell(u,v)$, is bounded 
by $O(\log n)$ in expectation.
\end{lemma}

\vspace{-1mm}
\noindent
{\bf Amortized Per-Update Cost.}
Next, we analyze the amortized running time for each update. 
Observe that, for an update of edge $(u,v)$,
according to Algorithm~\ref{algo:framework},
the running time cost of processing this update consists of:
\begin{itemize}[leftmargin = *]
\item a maintenance cost of $\ell(u,v)$ for the update;
\item a cost of Algorithm~\ref{algo:find}, $\costef$, to find a set $F$ of invalid edges; 
\item a maintenance cost of $\ell(x,y)$, for each edge $(x,y) \in F$.
\end{itemize}
By Lemma~\ref{lmm:maintenance-cost}, the update cost of an edge $(u,v)$ is bounded by $O(\log n + \costef + |F|\cdot \log n)$ in expectation. 
As in the worst case, the number of invalid edges, $|F|$, can be as large as $O(n)$, and the update cost can be as expensive as $O(n \log n + \costef)$ in expectation.

Fortunately, by update affordability, there must have been a certain number of affecting updates to trigger an edge$(x,y)$ being reported 
as invalid. 
Therefore, we can {\em charge} the costs of $O(\costef)$  and $O(|F|\cdot \log n)$ respectively to those updates which had contributed to 
them.
The key question is how to make the {\em charging argument} for these costs, specifically, which update is charged at what cost.

For simplicity, for the current update of edge $(u,v)$, we only analyze the part of $u$, because the analysis for the part of $v$ is symmetric.

\vspace{1mm}
\noindent
\underline{\em Amortize $\costef$ to Updates.}
According to Algorithm~\ref{algo:find}, 
we know that $\costef$ consists of two parts: (i) the bucket scanning cost, and (ii) the invalid edges reporting cost which is bounded by $O(|F|)$.
Let $K$ be the number of buckets that are checked (satisfying the if-condition in Line~4 of Algorithm~\ref{algo:find}).
Clearly, the scanning cost is $O(K + 1)$, where the ``$+1$'' term comes from the last bucket which does not satisfy the if-condition.
We thus charge this ``$+1$'' cost to the current update $(u,v)$.
%
Since, for each of the $K$ checked buckets, it must have at least one 
neighbor $w \in N(u)$ visited.
If $w$ is visited for the first time, this bucket checking cost can be charged to the maintenance cost $\ell(u,w)$.  
Otherwise, if $w$ is visited for the second time, this bucket checking cost can be charged to
the reporting cost $O(|F|)$, 
which, in turn, can also be further charged to the
maintenance cost of the edges in $F$, as we analyze next.

\vspace{1mm}
\noindent
\underline{\em Amortize the Maintenance Cost of an Edge to Updates.}
For each edge $(u,w)$ reported from a bucket $B_i$ in $\mathcal{B}(u)$,
according to Line~6 in Algorithm~\ref{algo:find}, 
$w$ is visited for the second time in $B_i$. 
Hence, there must have been at least $q(u,w)$ affecting updates incident on $u$ since $w$ was inserted to bucket $B_i$. 
Therefore, $\ell(u,w)$, the maintenance cost of edge $(u,w)$, can be charged to those at least $q(u,w)$ affecting updates, each of which is charged by a cost at most $\frac{\ell(u,w)}{q(u,w)}$.

Consider the current moment when an update of edge $(u,v)$ arrives;
this update $(u,v)$ 
is then charged (from the part of $u$) by 
at most $\sum_{w \in N(u)} \frac{\ell(u,w)}{q(u,w)}$.

Let $q(u, w^*)$ be the update affordability quota value in the {\em smallest} non-empty bucket $B^*$ in  $\mathcal{B}(u)$ at the current moment, and $w^* \in N(u)$. 
Consider the {\em retrospective degree} of $u$, 
denoted by $d_u^\text{ret}$, 
when $w^*$ was inserted to $B^*$, 
that is, when $q(u, w^*)$ was allocated.
The degree, $d_u$, of $u$ at the current moment satisfies:
$d_u \leq d_u^{\text{ret}} + 2 \cdot q(u, w^*)$.  
This is because, otherwise, 
$w^*$ must have been visited twice in $B^*$, and hence, the edge $(u,w^*)$ must have been reported as invalid.
This is contradictory to the fact that $w^*$ is still in $B^*$, and that $(u, w^*)$ is still considered  
as valid since $q(u, w^*)$ was allocated.

Furthermore, since $q(u,w^*) \geq \frac{1}{8} \tau(u,w^*)$  (see Line~2 in Algorithm~\ref{algo:insert}) and by Claim~\ref{claim:update-affordability}, 
we have $q(u,w^*) \in \Omega(\max\{d_u^{\text{ret}} , d_w^{\text{ret}}\})$.
Thus, $d_u^\text{ret} \in O(q(u, w^*))$; and it turns out that:
$d_u \leq d_u^\text{ret} + 2 \cdot q(u, w^*) \in O(q(u,w^*))$.
Therefore, 
the current update $(u,v)$ is charged by at most 
\begin{equation*}
\sum_{w \in N(u)} \frac{\ell(u,w)}{q(u,w)} \leq \frac{d_u}{q(u, w^*)} \cdot O(\log n) = O(\log n)\, \text{in expectation.}
\end{equation*}
%

\noindent
Putting the above-charged costs and the maintenance cost of the update $(u,v)$ itself together, 
we have:
\begin{lemma}\label{lmm:amortized-cost}
The amortized cost of each update is bounded by $O(\log n)$ in expectation.
\end{lemma}

\vspace{-4mm}
\subsection{Space Consumption}

For each  $u\in V$, the space consumption of (i) the date structures in {\em EdgeSimStr} and {\em CoreFindStr} with respect to $u$, (ii) the hash table of $N[u]$ for similarity calculation, and (iii) the auxiliary data structure for maintaining $\mathcal{B}(u)$ are all bounded by $O(n_u)$.
Hence, we have: 
\begin{lemma}\label{lmm:space-consumption}
The overall space consumption of {\em VD-STAR} is bounded by $O(n + m)$ at all times. 
\end{lemma}

\vspace{-4mm}
\subsection{Proof of the Last Missing Piece: 
Claim~\ref{claim:update-affordability}}
\label{sec:update-affordability}

Next, we give proof for Claim~\ref{claim:update-affordability} to complete our theoretical analysis. 
To show this claim, 
it suffices to prove that the update affordability satisfies $\tau(u,v) \geq t = \frac{1}{4} \rho^2 n_v \in \Omega(d_{\max}(u,v))$, for any edge $(u,v)$ with $n_u \leq n_v$. 
More specifically, in the following, we prove that $\tilde{\sigma}(u,v)$ remains a valid $\rho$-absolute approximation to the exact similarity $\sigma(u,v)$ at any moment within $t$ {\em arbitrary} affecting updates since the last moment when $\tilde{\sigma}(u,v)$ was computed.

In fact, Ruan~\shortauthors~\cite{ruan2021dynamic} give proof for a lemma similar to our Claim~\ref{claim:update-affordability} for Jaccard similarity only. 
Unfortunately, their proof is not immediately applicable to Cosine similarity.
As we show below, overcoming this technical difficulty of proving Claim~\ref{claim:update-affordability} for Cosine similarity
requires a more sophisticated analysis. 

First, we identify the cases when the similarity has the largest increment or decrement on the exact similarity for an affected update. Consider an update of edge $(u, w)$ and an affected edge $(u,v)$, there are four cases for each similarity measurement:

\noindent\underline{For Jaccard similarity}, 
\begin{itemize}[leftmargin = *]
    \item $(u,w)$ is an insertion, \begin{itemize}
        \item if $w \in N(v)$, $\sigma(u,v)$ is increased to $\frac{I(u,v) + 1}{n_u + n_v - I(u,v)}$ 
        \item if $w \notin N(v)$, $\sigma(u,v)$ is decreased to $\frac{I(u,v)}{n_u + n_v - I(u,v) + 1}$
    \end{itemize}
    \item $(u,w)$ is a deletion, \begin{itemize}
        \item if $w \in N(v)$, $\sigma(u,v)$ is decreased to $\frac{I(u,v) - 1}{n_u + n_v - I(u,v)}$ 
        \item if $w \notin N(v)$, $\sigma(u,v)$ is increased to $\frac{I(u,v)}{n_u + n_v - I(u,v) - 1}$
    \end{itemize}
\end{itemize}

\noindent\underline{For Cosine similarity},
\begin{itemize}[leftmargin = *]
    \item $(u,w)$ is an insertion, \begin{itemize}
        \item if $w \in N(v)$, $\sigma(u,v)$ is increased to $\frac{I(u,v) + 1}{\sqrt{(n_u+1) \cdot n_v}}$ 
        \item if $w \notin N(v)$, $\sigma(u,v)$ is decreased to $\frac{I(u,v)}{\sqrt{(n_u+1) \cdot n_v}}$
    \end{itemize}
    \item $(u,w)$ is a deletion, \begin{itemize}
        \item if $w \in N(v)$, $\sigma(u,v)$ is decreased to $\frac{I(u,v) - 1}{\sqrt{(n_u-1) \cdot n_v}}$ 
        \item if $w \notin N(v)$, $\sigma(u,v)$ is increased to $\frac{I(u,v)}{\sqrt{(n_u-1) \cdot n_v}}$
    \end{itemize}
\end{itemize}

\noindent\underline{For Dice similarity},
\begin{itemize}[leftmargin = *]
    \item $(u,w)$ is an insertion, \begin{itemize}
        \item if $w \in N(v)$, $\sigma(u,v)$ is increased to $\frac{I(u,v) + 1}{(n_u + n_v + 1)/2}$ 
        \item if $w \notin N(v)$, $\sigma(u,v)$ is decreased to $\frac{I(u,v)}{(n_u + n_v + 1)/2}$
    \end{itemize}
    \item $(u,w)$ is a deletion, \begin{itemize}
        \item if $w \in N(v)$, $\sigma(u,v)$ is decreased to $\frac{I(u,v) - 1}{(n_u + n_v - 1)/2}$ 
        \item if $w \notin N(v)$, $\sigma(u,v)$ is increased to $\frac{I(u,v)}{(n_u + n_v - 1)/2}$
    \end{itemize}
\end{itemize}
Through factorization, it is not difficult to prove that the first case has the largest increment and the third case has the largest decrement for all three similarity measurements. 
With this, we now prove Claim~\ref{claim:update-affordability} for the following two cases separately.

\vspace{1mm}
\noindent
{\bf Case 1: $n_u \leq \frac{1}{4} \rho^2 n_v$.} 
According to Algorithm~\ref{algo:calsim}, 
we set $\tilde{\sigma}(u,v) = 0$ for all three similarity measurements in this case.
By Observation~\ref{claim:degree-ratio} in the proof of Lemma~\ref{lmm:correct-sim},
we know that the exact similarities can be upper bounded by a function of 
$\beta = \frac{n_u}{n_v}$ for the three measurements.
Specifically, $Jaccard(u,v) \leq \beta$, $Cosine(u,v) \leq \sqrt{\beta}$ and $Dice(u,v) \leq 2\beta$.
At the moment when $\tilde{\sigma}(u,v)$ is set to $0$, we know that the value of $\beta \leq \frac{1}{4} \rho^2$.
Next, we show that after $t = \frac{1}{4} \rho^2 n_v$ arbitrary affecting updates, the value of $\beta$ cannot be greater than $\frac{1}{2}\rho^2$.
And thus, by Observation~\ref{claim:degree-ratio}, the exact similarities are still no more than $\rho$ for all the three similarity measurements, and therefore, $\tilde{\sigma}(u,v) = 0$ is still a valid $\rho$-absolute approximation.

It suffices to consider those affecting updates that increase the value of $\beta$ only.
Since $n_u \leq n_v$, without loss of generality, we assume that there are $0 \leq b \leq t$ decrements on $n_v$ while $t - b$ increments on $n_u$.
Let $n_u' = n_u + (t - b)$ and $n_v' = n_v - b$.
After such $t$ updates, we have: 
$\rho^2 n_v' = \frac{2}{4} \rho^2 n_v + \frac{2}{4} \rho^2 n_v - \rho^2 b \geq 2 n_u + 2 t - 2 b = 2 n_u'\,.$
Therefore, after these $t = \frac{1}{4}\rho^2 n_v$ updates, 
the value of $\beta = \frac{n_u'}{n_v'} \leq \frac{1}{2} \rho^2$ holds.
%
This completes the proof of Claim~\ref{claim:update-affordability} for Case 1.

\vspace{1mm}
\noindent
{\bf Case 2: $n_v \geq n_u > \frac{1}{4} \rho^2 n_v$.}
Again consider the moment when a $\frac{1}{2}\rho$-absolute-approximate $\tilde{\sigma}(u,v)$ is computed and the exact similarity at this moment denoted by $\sigma^*(u,v)$.
Next, we examine 
after $t = \frac{1}{4} \rho^2 n_v$ affecting updates,
the value of $\sigma(u,v)$ cannot be 
increased nor decreased by more than $\frac{1}{2}\rho$.
And therefore, $\tilde{\sigma}(u,v)$ remains a valid $\rho$-absolute approximation to the exact similarity at the current moment.

We first show the increment case. As verified above, affecting updates of edges that increase the intersection size $I(u,v)$ 
of $N[u]$ and $N[v]$ is the most effective way to increase the exact similarity $\sigma(u,v)$.
Without loss of generality, suppose that $n_u$ and $n_v$ are increased by $t - b$ and $b$, respectively, after $t$ affecting updates. 

\vspace{1mm}
\noindent
\underline{For Jaccard similarity}, 
with $t = \frac{1}{4} \rho^2 n_v \leq \frac{1}{2}\rho n_v$, 
the increased exact similarity becomes $\sigma(u,v) = \frac{I(u,v)+t}{(n_u + t - b)  + (n_v + b)  - (I(u,v) +t)} \leq {\sigma^*}(u,v) + \frac{t}{n_u + n_v - I(u,v)} \leq {\sigma^*}(u,v) + \frac{\frac{1}{2}\rho n_v}{n_u + n_v - I(u,v)} \leq {\sigma^*}(x,y) + \frac{1}{2}\rho$.

\vspace{1mm}
\noindent
\underline{For Cosine similarity}, with $t = \frac{1}{4} \rho^2 n_v$,
the increased exact similarity becomes
$\sigma(u,v) = \frac{I(u,v) + t}{\sqrt{(n_u + t - b) \cdot (n_v + b)}} < \sigma^*(u,v) + \frac{t}{\sqrt{n_u \cdot n_v}} <  \sigma^*(u,v) + \frac{1/4 \cdot \rho^2 n_v}{\sqrt{1/4 \cdot \rho^2 n_v \cdot n_v}} = \sigma^*(u,v) + \frac{1}{2} \rho 
$.

\vspace{1mm}
\noindent
\underline{For Dice similarity}, 
with $t = \frac{1}{4} \rho^2 n_v \leq \frac{1}{4}\rho n_v$, 
the increased exact similarity becomes
$\sigma(u,v) = \frac{I(u,v) + t}{(n_u + n_v + t) / 2} \leq \sigma^*(u,v) + \frac{t}{(n_u + n_v) / 2} 
\leq \sigma^*(u,v) + \frac{1/4 \rho n_v}{n_v / 2} = \sigma^*(u,v) + \frac{1}{2} \rho$.

For the decrement case, affecting updates of edges that decrease the intersection size $I(u,v)$ 
of $N[u]$ and $N[v]$ is the most effective way to decrease the exact similarity $\sigma(u,v)$.
Without loss of generality, suppose that $n_u$ and $n_v$ are decreased by $t - b$ and $b$, respectively, after $t$ affecting updates.

\vspace{1mm}
\noindent
\underline{For Jaccard similarity}, 
with $t = \frac{1}{4} \rho^2 n_v \leq \frac{1}{2}\rho n_v$, 
the decreased exact similarity becomes $\sigma(u,v) = \frac{I(u,v)-t}{(n_u - t + b)  + (n_v - b)  - (I(u,v) -t)} \geq {\sigma^*}(u,v) - \frac{t}{n_u + n_v - I(u,v)} \geq {\sigma^*}(u,v) - \frac{\frac{1}{2}\rho n_v}{n_u + n_v - I(u,v)} \geq {\sigma^*}(u,v) - \frac{1}{2}\rho$.

\vspace{1mm}
\noindent
\underline{For Cosine similarity}, with $t = \frac{1}{4} \rho^2 n_v$,
the decreased exact similarity becomes
$\sigma(u,v) = \frac{I(u,v) - t}{\sqrt{(n_u - t + b) \cdot (n_v - b)}} > \sigma^*(u,v) - \frac{t}{\sqrt{n_u \cdot n_v}} >  \sigma^*(u,v) - \frac{1/4 \cdot \rho^2 n_v}{\sqrt{1/4 \cdot \rho^2 n_v \cdot n_v}} = \sigma^*(u,v) - \frac{1}{2} \rho 
$.

\vspace{1mm}
\noindent
\underline{For Dice similarity}, 
with $t = \frac{1}{4} \rho^2 n_v \leq \frac{1}{4}\rho n_v$, 
the decreased exact similarity becomes
$\sigma(u,v) = \frac{I(u,v) - t}{(n_u + n_v - t) / 2} \geq \sigma^*(u,v) - \frac{t}{(n_u + n_v) / 2} 
\geq \sigma^*(u,v) - \frac{1/4 \rho n_v}{n_v / 2} = \sigma^*(u,v) - \frac{1}{2} \rho$.

Therefore, for any of these similarity measurements, the update affordability $\tau(u,v) \geq t = \frac{1}{4} \rho^2 n_v \in \Omega(d_{\max}(u,v))$ holds for Case 2.
This completes the whole proof for Claim 1.

\vspace{-4mm}
\section{Optimizations}
\label{sec:opt}
We introduce two optimizations to enhance the practical performance of our {\em VD-STAR}. 
The idea stems from a crucial observation -- the {\em CoreFindStr} is designed for finding core vertices efficiently in $O(|\Vc| + 1)$ time to achieve the target query time complexity $O(m_{cr})$. 
If we relax this query bound, then it is not necessary to implement the {\em CoreFindStr}.
In this way, we can considerably improve 
the update efficiency 
by not only shaving the maintenance cost for {\em CoreFindStr} but also, importantly, 
releasing the ``approximation budget'': 
recall that {\em VD-STAR} uses a $\Delta$-Table which introduces a $\Delta$-absolute error in the approximation. 
Hence, we can set a larger $\rho$ for {\em EdgeSimStr} 
that achieves the same approximation guarantee.

\vspace{-2mm}
\subsection{{\em VD-STAR} with No {\em CoreFindStr}}
\label{subsec:opt1}

Consider an implementation of our {\em VD-STAR} without {\em CoreFindStr}.
When a query with parameters $\eps$ and $\mu$ arrives, 
to identify all the core vertices,
it suffices to check for each vertex $u \in V$ whether $u$ is a core vertex with $u$'s sorted neighbor linked list $N(u)$.
This can be achieved by scanning $N(u)$ from the beginning and checking whether the similarity between $u$ and the $\mu^\text{th}$ (largest) neighbor is $\geq \eps$ or not.
The time complexity is clearly bounded by $O(\mu)$ for each vertex $u$, and hence, the overall running time of identifying all the vertices is bounded by $O(\mu \cdot n)$.
If the sorted neighbor list $N(u)$
is maintained with a binary search tree,
finding the $\mu^\text{th}$
largest similarity can 
be achieved in 
$O(\log d_{\max})$ time. In this case, the core vertex identification cost is bounded by $O(n \cdot \log n)$.
Therefore, without the {\em CoreFindStr}, our {\em VD-STAR} 
can answer each query in $O(\min\{\mu, \log n\} \cdot n + m_{cr})$ time.

Note that, in practice, this query time complexity is acceptable because: (i) the parameter $\mu$ in practice is often a small constant for which $\mu \cdot n \in O(n)$ often holds, and (ii) for reasonable clustering parameters, $m_{cr}$ often dominates the term $O(\min\{\mu, \log n\} \cdot n)$. 
If either of these cases happens, the query time complexity is still bounded by $O(m_{cr})$ the same as before with {\em CoreFindStr}.  
As we will see in experiments, {\em VD-STAR} with no {\em CoreFindStr}, which is named {\em Ours-NoT}, significantly improves the update efficiency with just a negligible sacrifice in the query efficiency.

\subsection{{\em VD-STAR} with a Small $\mu$-Table}
\label{subsec:opt2}

Recall that {\em CoreFindStr} can be implemented with a $\mu$-Table 
which is used in GS*-Index and does not ``consume'' any approximation budget. 
Inspired by this, our other version of {\em VD-STAR} is to implement {\em CoreFindStr} with a {\em small} $\mu$-Table.
In the sense that, we do not implement the $\mu$-Table {\em in full}
to capture all possible values of the given parameter $\mu$. 
Instead, we just implement it {\em partially} for the $\mu$ values 
up to a small constant, say $15$.
As a result, the maintenance of the small $\mu$-Table would not affect the update time complexity of {\em VD-STAR}. 
In addition, it releases the approximation budget consumed by the $\Delta$-Table implementation, and hence, we can increase the value of $\rho$ for the {\em EdgeSimStr} accordingly.

To answer a query with parameters $\eps$ and $\mu$, 
if $\mu$ is captured by the small $\mu$-Table, then we use the $\mu$-Table to retrieve all the core vertices in $O(|\Vc| + 1)$ and hence, the query time complexity is bounded by $O(m_{cr})$ as desired.
Otherwise, we just run the above version of {\em VD-STAR} with no {\em CoreFindStr} to answer the query.

\vspace{-2mm}
\section{Experiments}
\label{sec:exp}
\subsection{Experimental Settings}

\vspace{1mm}
\noindent\textbf{Datasets.}
We evaluate our algorithms on nine real-world datasets from the Stanford Network Analysis Project~\cite{snapnets} and Network Repository~\cite{nr} which are also used in the baseline papers~\cite{ruan2021dynamic,zhang2022effective,wen2017efficient}. 
Following previous works~\cite{ruan2021dynamic, zhang2022effective}, we treat all graphs as undirected and remove all self-loops. Table~\ref{tab:dataset} summarizes the dataset.

\noindent\textbf{Competitors.} 
We study the performance of our three algorithms:  
{\em VD-STAR}, {\em VD-STAR-NoT} (Section~\ref{subsec:opt1}) and 
{\em VD-STAR-$\mu$T} (Section~\ref{subsec:opt2}),
which are respectively denoted by {\bf \em Ours}, {\bf \em Ours-NoT} and {\bf \em Ours-$\mu$T} for short.
We compare these algorithms 
with 
the SOTA exact and approximate algorithms {\bf GS*-Index}~\cite{wen2017efficient} and {\bf BOTBIN}~\cite{zhang2022effective}.
In {\em Ours-$\mu$T}, only a $\mu$-Table with $\mu_{\max} = 15$ is constructed. 

\noindent
{\bf Experiment Environment.}
All experiments are conducted on a Ubuntu virtual server with a 2 GHz CPU and 1 TB memory. All source codes are in C++ and compiled with -O3 turned on.
The source code of our implementations can be found in~\cite{sourcecode}.

\noindent
\textbf{Default Parameter Settings.} 
By default, the {\em target overall approximation budget}, denoted by $\rho^*$, is set as $\rho^* = 0.02$.
For BOTBIN, 
we set $\Delta = 0.01$ by default as suggested in its paper.
Since the use of $\Delta$-Table would introduce a $\Delta$ error, 
we set $\rho = \rho^* - \Delta = 0.01$ for the {\em EdgeSimStr} in BOTBIN
to achieve an overall $\rho^*$-approximation.
We set the same $\Delta$ and $\rho$ for {\em Ours}.
As both {\em Ours-NoT} and {\em Ours-$\mu$T} do not adopt the $\Delta$-Table,
we set $\rho = \rho^*$ for fair comparison.

\begin{figure*}[t]
\begin{subfigure}{0.8\linewidth}
    \centering
    \includegraphics[width=\textwidth]{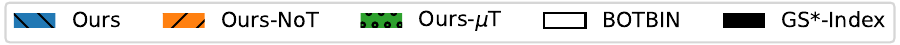}
\end{subfigure}
\begin{subfigure}{0.5\linewidth}
    \includegraphics[width=1\textwidth]{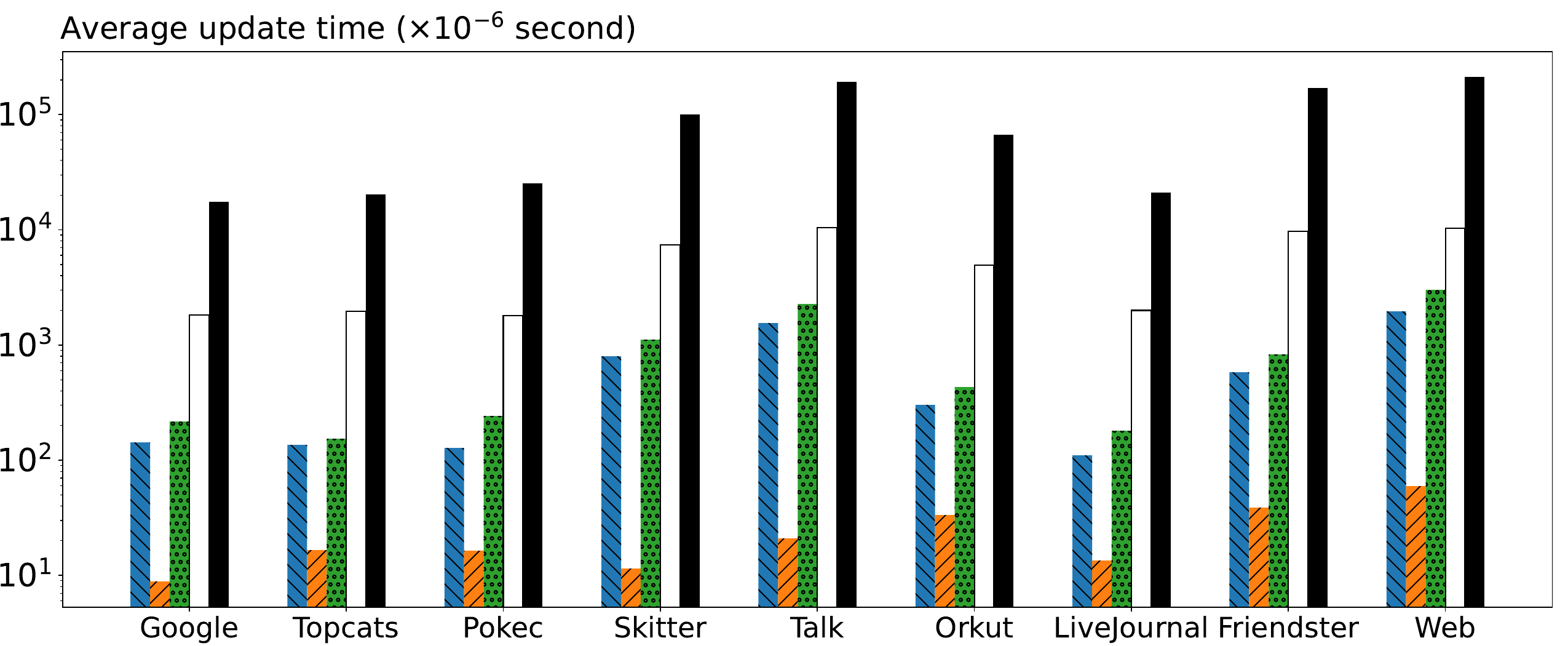}
    \caption{Average update running time}    
    \label{fig:update}
\end{subfigure}%
\begin{subfigure}{0.5\linewidth}
    \includegraphics[width=1\textwidth]{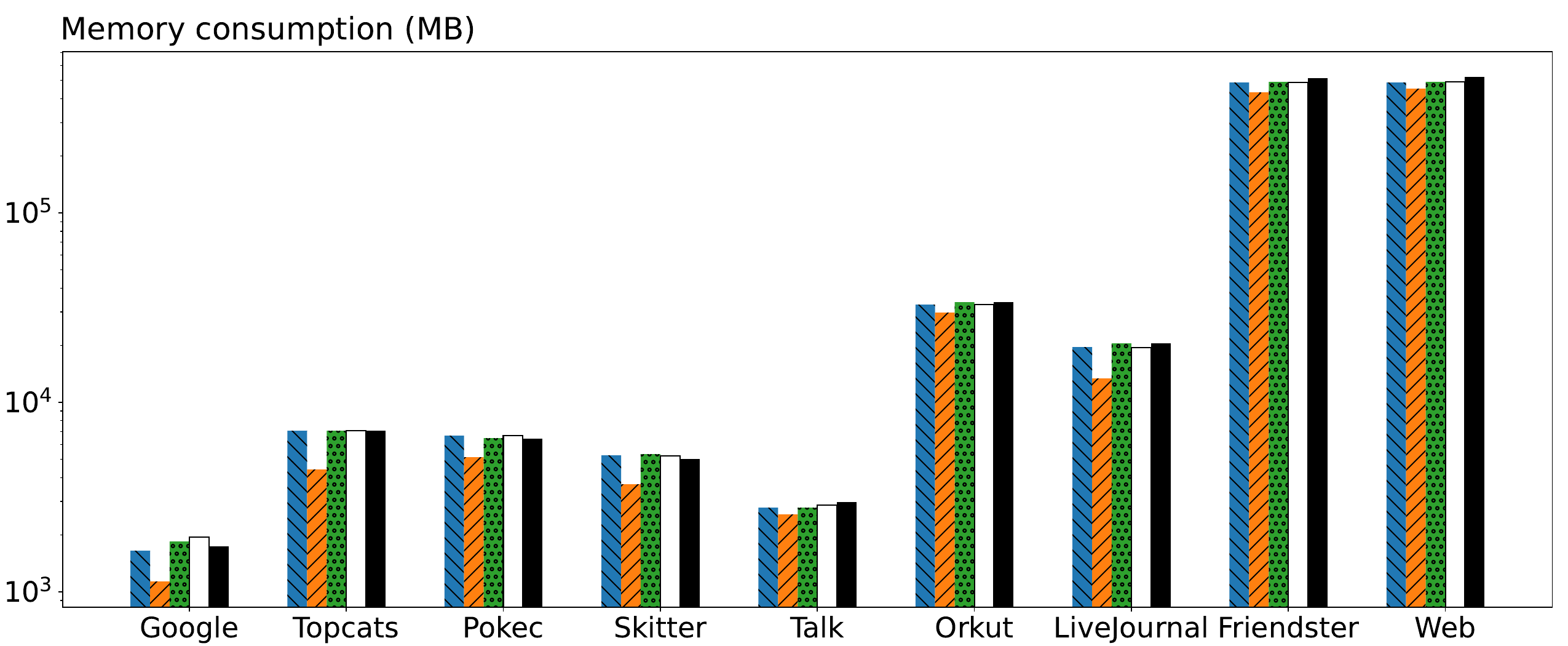}
    \caption{Memory consumption}
    \label{fig:mem}
\end{subfigure}
    \caption{Update processing performance results}
    \label{fig:efficiency}
\end{figure*}

\noindent\textbf{Update Generation.} 
To simulate graph updates in real-world applications, we randomly generate a sequence of edge insertions and deletions for each dataset. 
To testify different scenarios, we vary the ratio $\eta$ of \#deletion to \#insertion by setting the probabilities of an insertion and a deletion to $\frac{1}{1+\eta}$ and $\frac{\eta}{1+\eta}$, respectively. 
To perform an edge deletion, we uniformly at random choose an existing edge and delete it. 
To perform an edge insertion, we employ three strategies:
\vspace{-1mm}
\begin{itemize}[leftmargin=*]
    \item random-random (\textbf{RR}): a non-existent edge is randomly added.
    \item degree-random (\textbf{DR}): Each vertex $u$ has a probability of $\frac{d_u}{2m}$ to be chosen, where $d_u$ is the degree of $u$ and $m$ is the number of edges in the current graph. 
Once $u$ is chosen, the second vertex, $v$, is randomly chosen from those vertices not yet linked to $u$.
    \item degree-degree (\textbf{DD}): Vertex $u$ is chosen as in DR; vertex $v$ is chosen from the vertices not yet linked to $u$ with $\frac{d_v}{2m}$ probability.
\end{itemize}
\vspace{-1mm}
By default, we set $\eta = \frac{1}{10}$.
For each dataset and a configuration of $\eta$ and update generation strategy, we generate $M = 2m^*$ updates, where $m^*$ is the number of edges in the initial graph.

\noindent\textbf{Query Simulation.} To simulate the query process in real-world applications, we randomly generate a query after every 20 updates, with $\eps \in [0.1, 0.5]$ and $\mu \in [2, 2\bar{d}]$ of each graph ($\bar{d}$: the average degree). 
With a total of $M = 2m^*$ updates, $0.1m^*$ queries are tested.

\vspace{-2mm}
\subsection{Study on Update Efficiency}\label{subsec:exp_index}


\subsubsection{\bf Average Update Time with Default Parameters.}
We first study 
the average running time of processing updates. 
As shown in Figure~\ref{fig:update}, we have the following observations:
(1) {\em Ours-NoT}~consistently achieves the best update efficiency. Particularly, it accelerates update processing by as much as 9,315 times (on \texttt{wiki-Talk}) compared with GS*-Index and 647 times (on \texttt{as-skitter}) compared with BOTBIN. GS*-Index uses exact similarity calculation such that it has the highest update time. 
(2) {\em Ours}~is up to 18 times faster in update processing (on \texttt{soc-LiveJournal1})  compared to BOTBIN which also uses a $\Delta$-table. 
(3) {\em Ours-$\mu$T}~ achieves a significant improvement in updating speed, while maintaining competitive query time, as elaborated later in Section~\ref{subsubsec:qe}. It outperforms SOTA methods by up to 193 times on \texttt{soc-Friendster} vs. GS*-Index. 
\vspace{-1mm}

\begin{table}[ht]
\centering
\caption{Dataset Summary}
\scalebox{0.9}{%
\begin{tabular}{l|r|r|r|c}
\toprule
\multicolumn{1}{c|}{\textbf{Datasets}}         & $\mathbf{n(\times10^6)}$     & $\mathbf{m(\times10^6)}$   & \multicolumn{1}{c|}{$\boldsymbol{\Bar{d}}$} & \textbf{Domain} \\\hline
\texttt{web-Google}       & 0.88  & 4.32 & 9.86 & Website hyperlink\\\hline
\texttt{wiki-topcats}     & 1.79 & 25.44 & 28.38 & Website hyperlink \\\hline
\texttt{soc-Pokec}        & 1.63 & 22.30 & 27.36 & Social network \\\hline
\texttt{as-skitter}       & 1.70 & 11.10 & 13.06 & Traceroute graph\\\hline
\texttt{wiki-Talk}        & 2.39 & 4.66 & 3.90& Interaction graph\\\hline
\texttt{soc-Orkut}            & 3.07 & 117.19  & 76.22 & Social network\\\hline
\texttt{soc-LiveJournal1} & 4.85 & 42.85 & 17.69 & Social Network\\\hline
\texttt{soc-Friendster}   & 65.61 & 1,806.07  & 55.13 & Social network\\\hline
\texttt{web-2012} & 90.32 & 1,940.85 & 42.91 & Website hyperlink\\\bottomrule
\end{tabular}}
\label{tab:dataset}
\end{table}

\subsubsection{\bf Memory Consumption.} 
As shown in Figure~\ref{fig:mem}, all methods exhibit minor differences in memory consumption 
owing to that their space consumption 
are all linear to the graph size.
{\em Ours-NoT} has the smallest memory consumption 
because it has no implementation for {\em CoreFindStr}.
{\em Ours} consumes slightly less memory compared with BOTBIN possibly because maintaining $\mathcal{B}(u)$ 
might be more space-efficient than maintaining the bottom-$k$ signatures. 

\vspace{-1mm}
\subsubsection{\bf Impact of Target Overall Approximation $\rho^*$.} 

To test how $\rho^*$ affects the update time, we vary $\rho^*$ from $0.001$ to $0.1$. 
GS*-Index is excluded from this experiment as it is an exact algorithm. 
As shown in Figure~\ref{fig:rho_varying}, 
the update time decreases when $\rho^*$ grows, as expected. 
A larger $\rho^*$ value leads to smaller sample sizes for similarity estimation 
and larger update affordability ($\tau$) for our algorithms. 
As expected, {\em Ours-NoT} has the lowest update time under all $\rho^*$ values tested, 
with speedups reaching up to 700 times over BOTBIN (on \texttt{as-skitter} when $\rho^*=0.01$).

Compared with BOTBIN, \algomu\ and \algodelta\  are also about an order of magnitude faster. 
Remarkably, even with $\rho=0.001$, our \algo~ achieves an average update time of $ 69 \times 10^{-6}$ second on a graph with 1.9 billion edges. 
For BOTBIN and \algodelta, $\Delta$ is set to 0.01 except for $\rho^*=0.01$ and $\rho^*=0.001$ where $\Delta$ is set as half of $\rho^*$ to satisfy the $\rho^*$-absolute-approximation for these two algorithms.

\begin{figure}[t]
    \centering
    \begin{subfigure}{0.7\linewidth}
        \includegraphics[width=\textwidth]{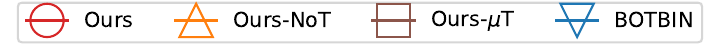}
    \end{subfigure}
    \begin{tabular}{ccc}
        \hspace{-3mm}\begin{subfigure}{0.333\linewidth}
            \includegraphics[width=\textwidth]{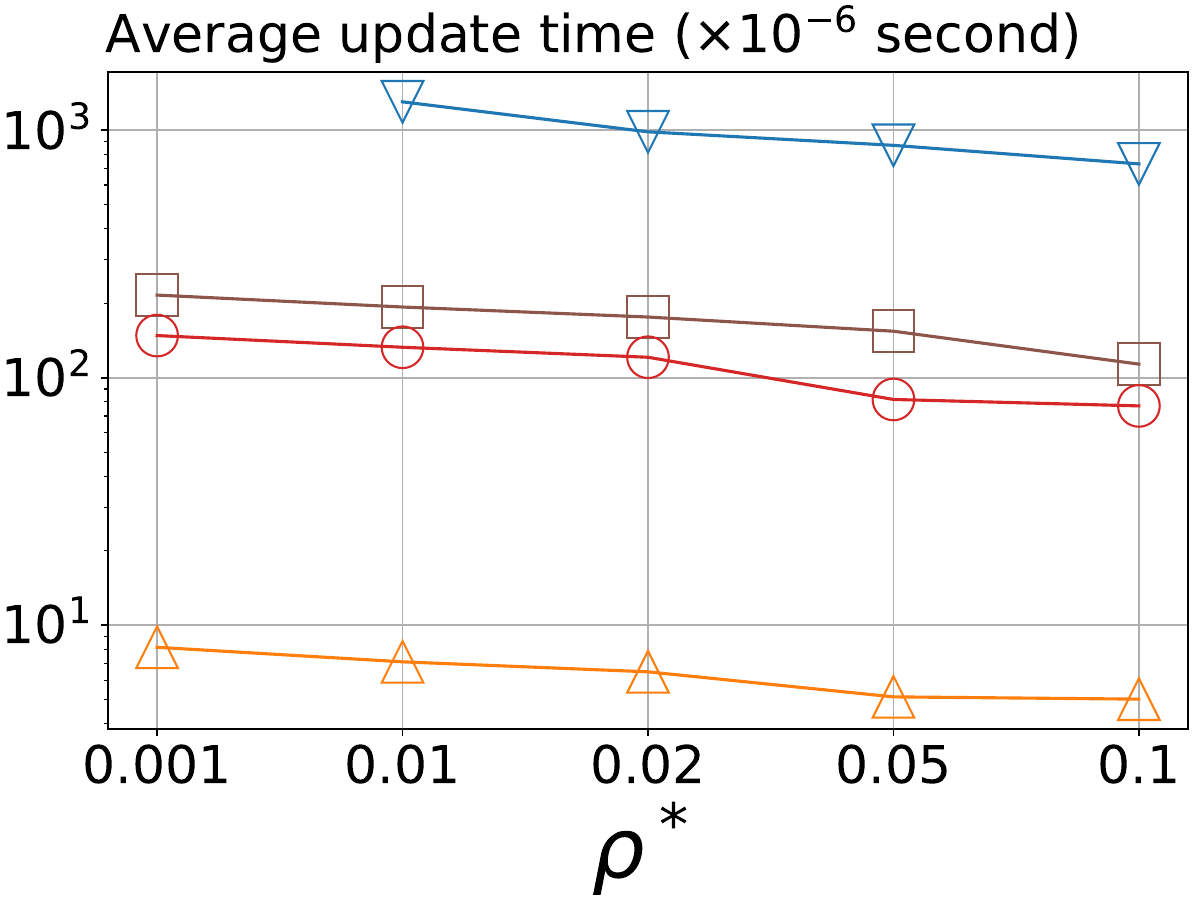}
            \vspace{-6mm}
            \caption{Google}
            \vspace{-1mm}
        \end{subfigure}&
            \hspace{-3mm}\begin{subfigure}{0.333\linewidth}
                \includegraphics[width=\textwidth]{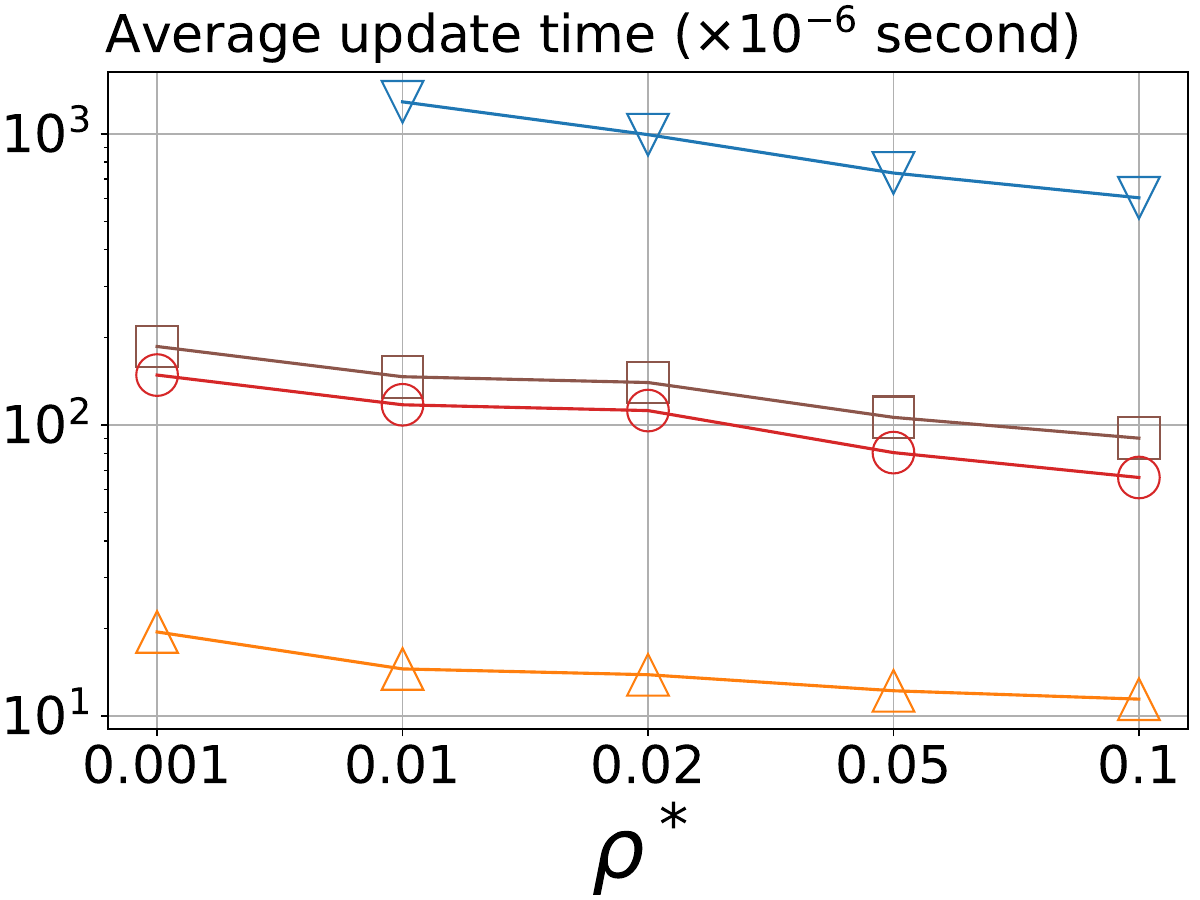}
                \vspace{-6mm}
                \caption{Topcats}
                \vspace{-1mm}
        \end{subfigure}&
            \hspace{-3mm}\begin{subfigure}{0.333\linewidth}
                \includegraphics[width=\textwidth]{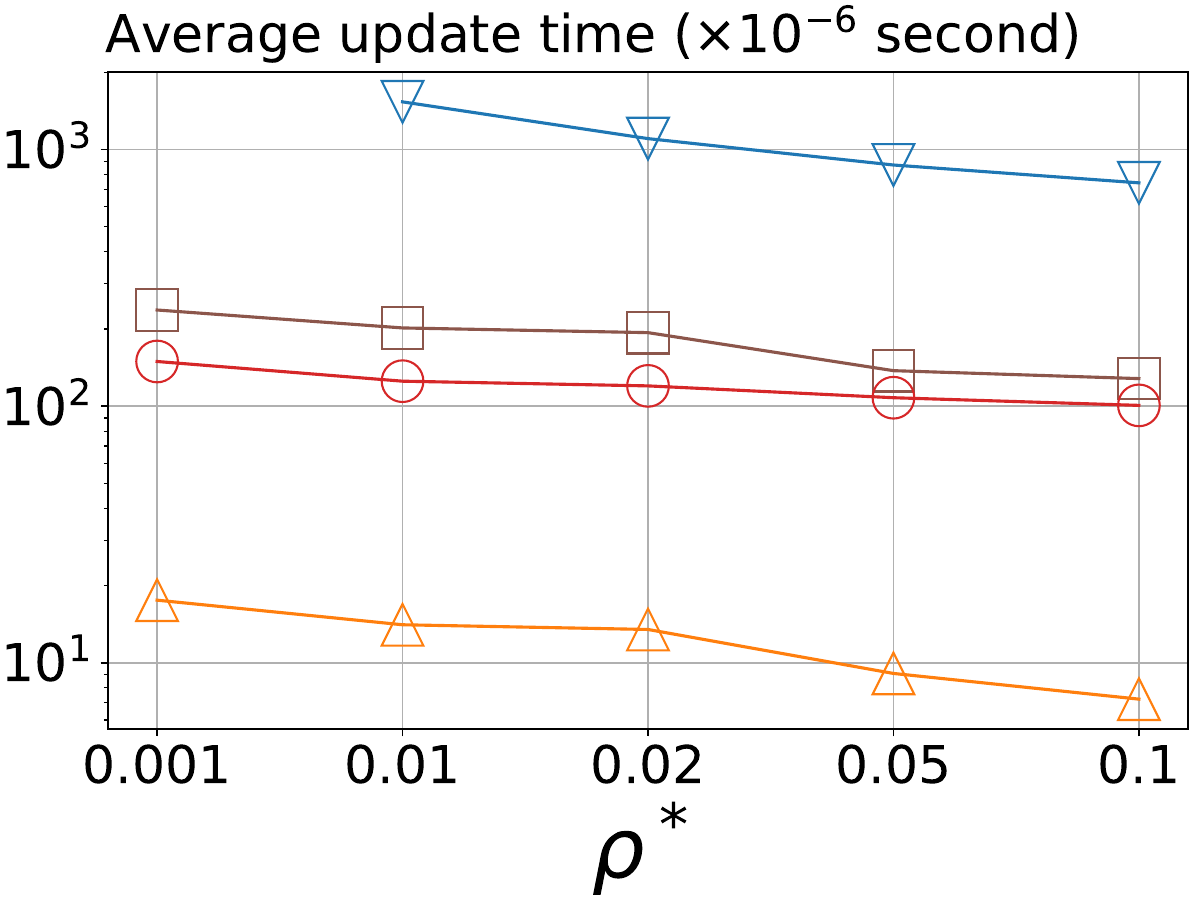}
                \vspace{-6mm}
                \caption{Pokec}
                \vspace{-1mm}
        \end{subfigure}\\
            \hspace{-3mm}\begin{subfigure}{0.333\linewidth}
                \includegraphics[width=\textwidth]{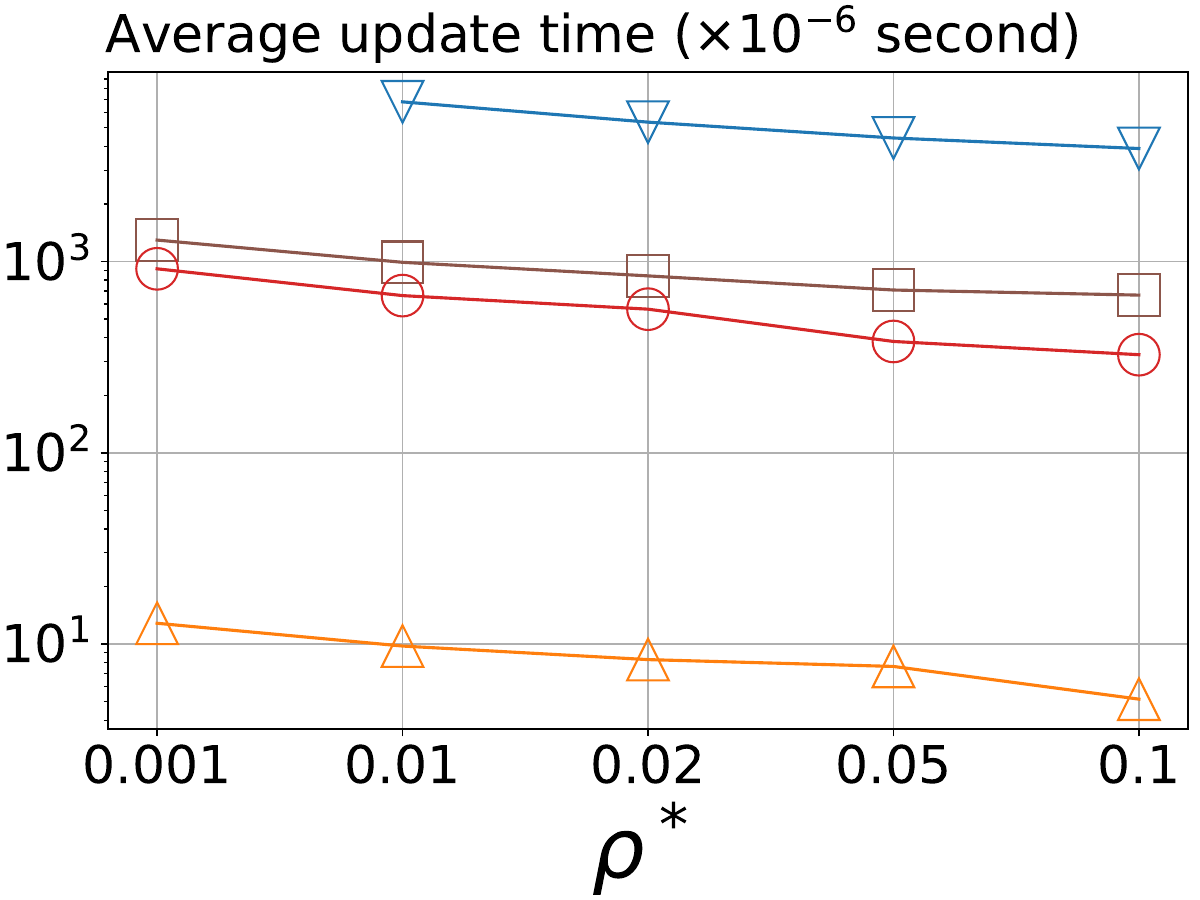}
                \vspace{-6mm}
                \caption{Skitter}
                \vspace{-2mm}
        \end{subfigure}&
            \hspace{-3mm}\begin{subfigure}{0.333\linewidth}
                \includegraphics[width=\textwidth]{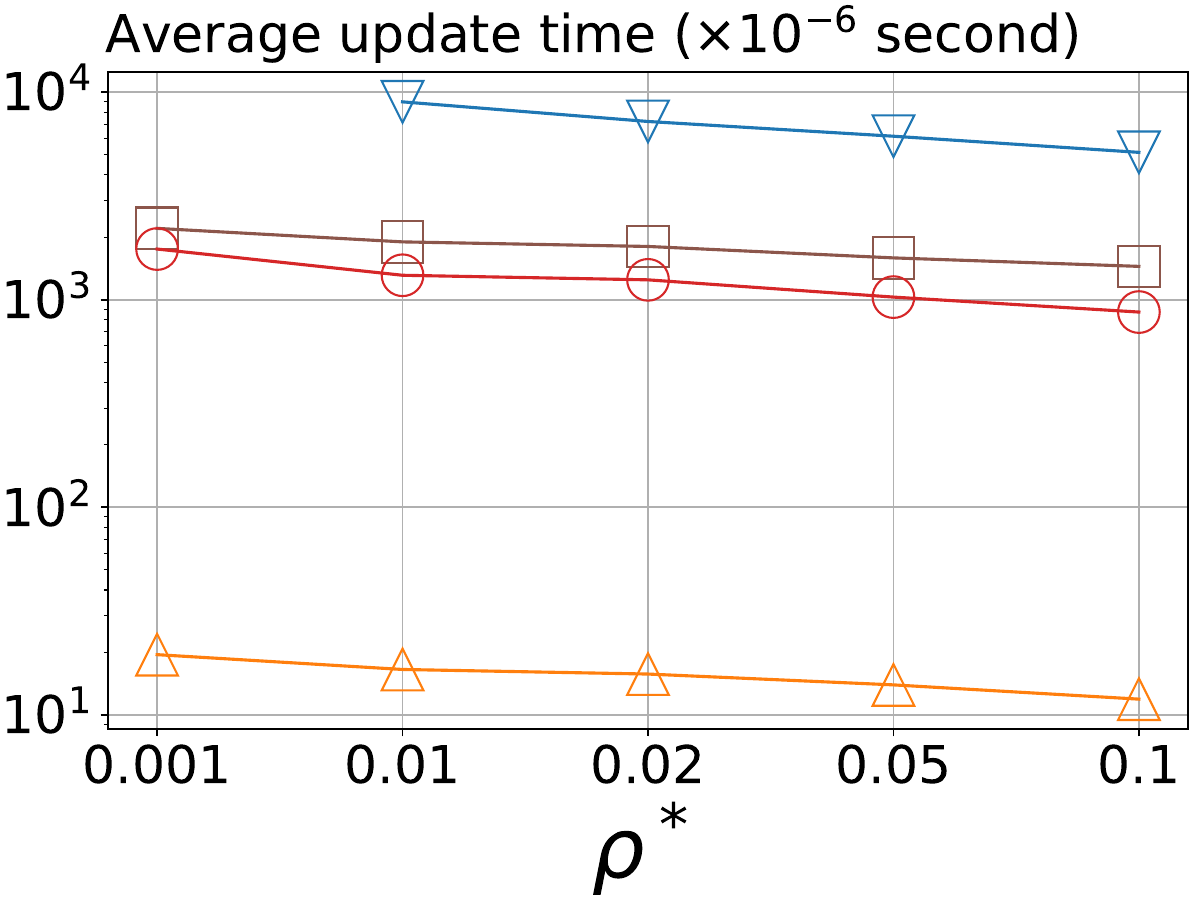}
                \vspace{-6mm}
                \caption{Talk}
                \vspace{-2mm}
        \end{subfigure}&
        \hspace{-3mm}\begin{subfigure}{0.333\linewidth}
            \includegraphics[width=\textwidth]{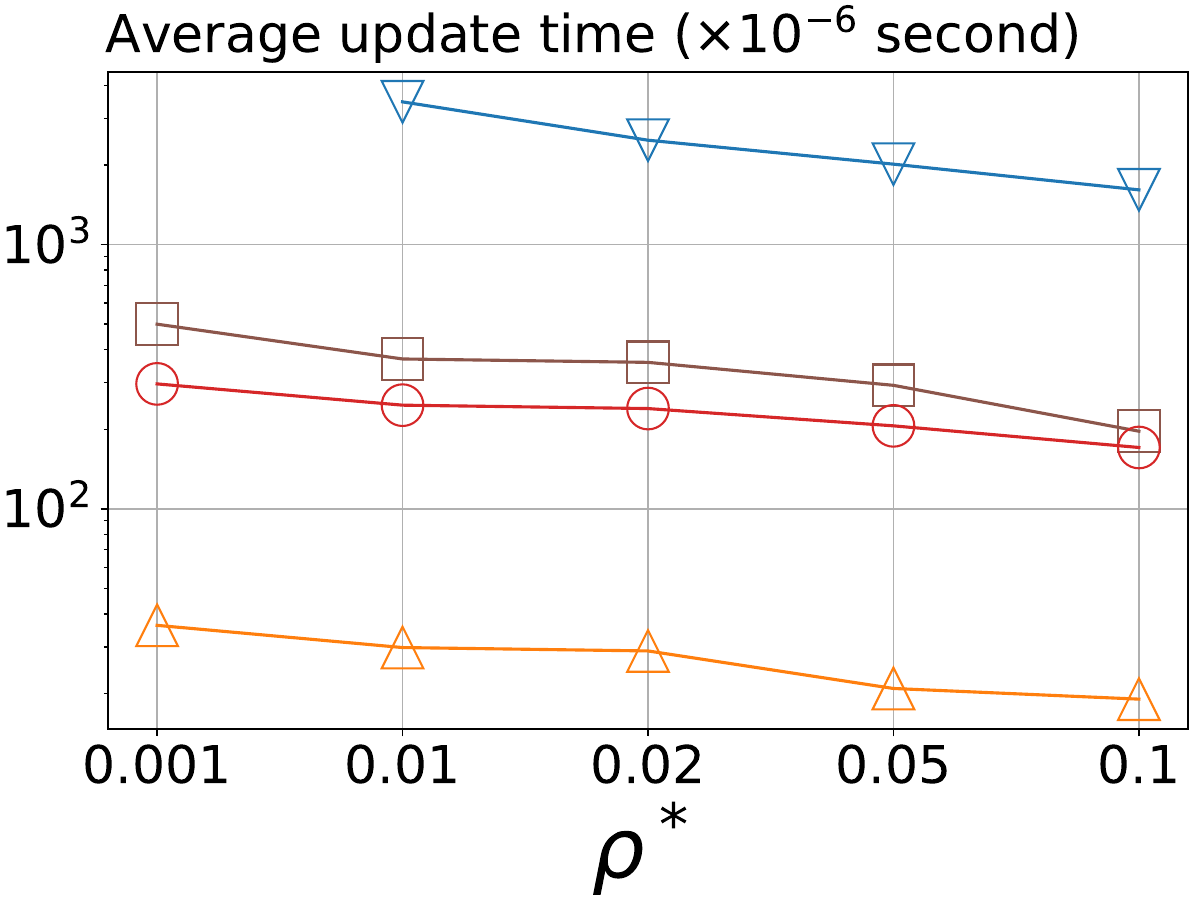}
            \vspace{-6mm}
            \caption{Orkut}
            \vspace{-2mm}
        \end{subfigure}\\
        \hspace{-3mm}\begin{subfigure}{0.333\linewidth}
            \includegraphics[width=\textwidth]{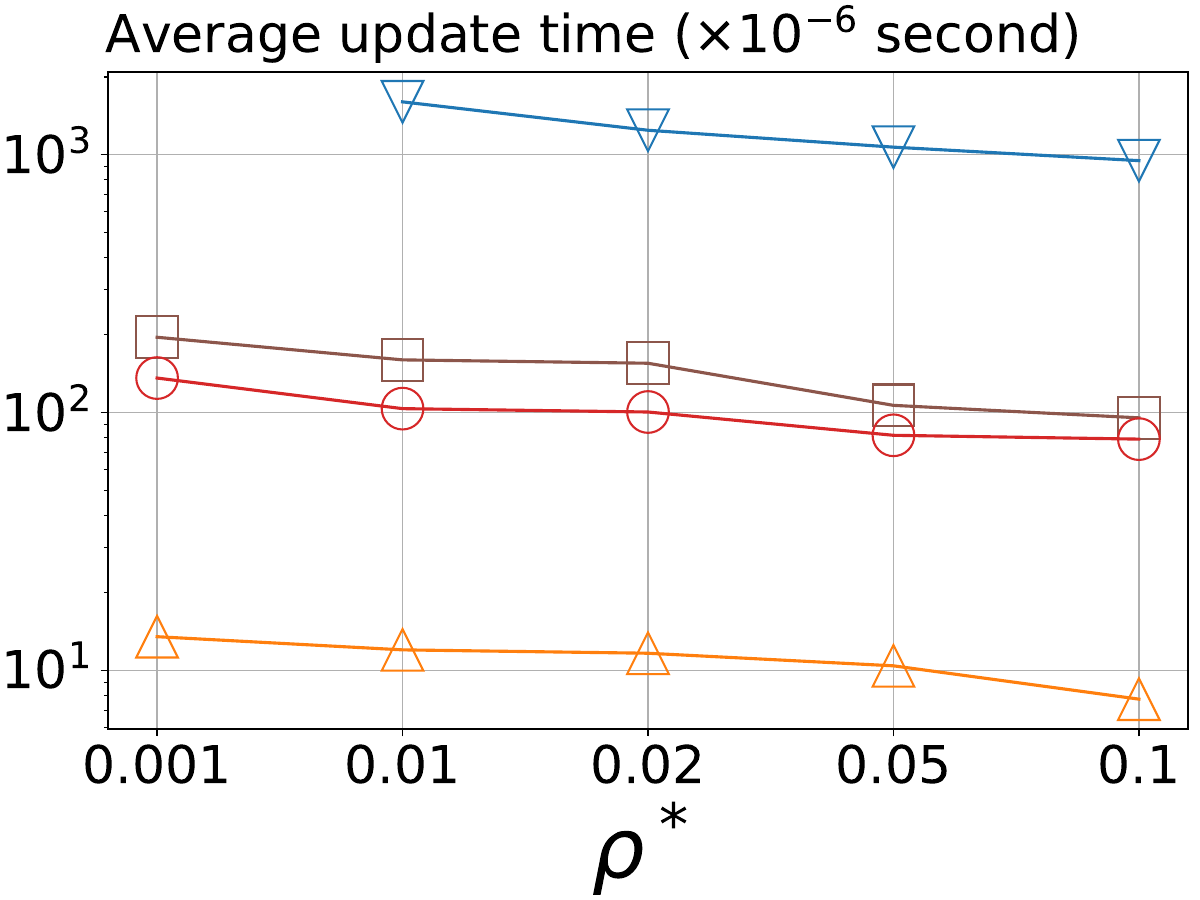}
            \vspace{-6mm}
            \caption{LiveJournal}
            \vspace{-2mm}
        \end{subfigure}&
        \hspace{-3mm}\begin{subfigure}{0.333\linewidth}
            \includegraphics[width=\textwidth]{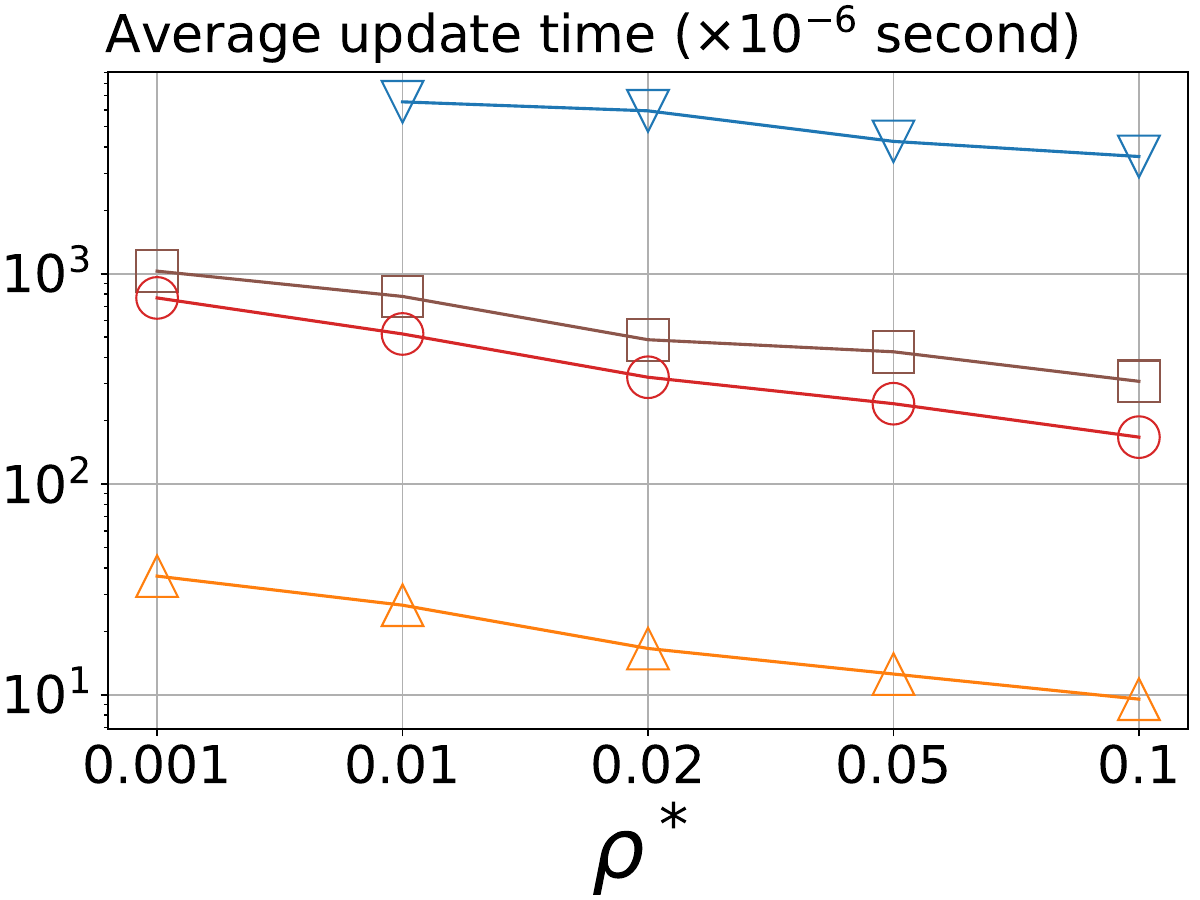}
            \vspace{-6mm}
            \caption{Friendster}
            \vspace{-2mm}
        \end{subfigure}&
        \hspace{-3mm}\begin{subfigure}{0.333\linewidth}
            \includegraphics[width=\textwidth]{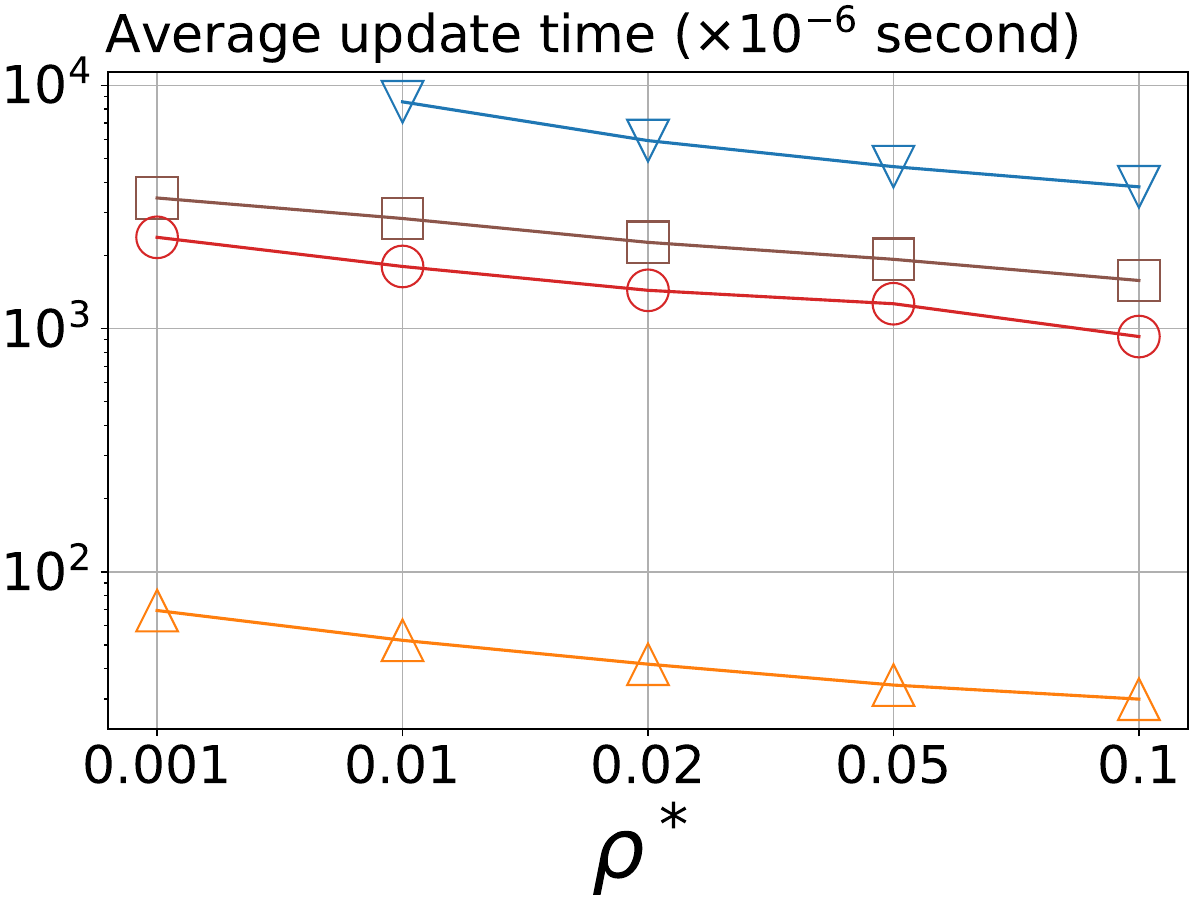}
            \vspace{-6mm}
            \caption{Web}
            \vspace{-2mm}
        \end{subfigure}\\
    \end{tabular}
    \vspace{-3mm}
    \caption{Average update running time vs.  $\rho^*$}
    \label{fig:rho_varying}
\vspace{-4mm}
\end{figure}

\vspace{-1mm}
\subsubsection{\bf Impact of Update Distributions}
\label{subsec:update_dis}


As shown in Figure~\ref{fig:ug_varying}, the average update times of all algorithms increase as updates follow more skewed distributions (from RR to RD and to DD). This trend primarily arises because inserting or deleting a neighbor for a vertex of a larger degree takes a longer time (i.e., $O(\log n_{u})$). Additionally, vertices of larger degrees appear more frequently in the $\mu$-Table. Despite these challenges, our algorithms consistently outperform all competitors, with speedups of up to 6,098 times on \textbf{RR}, 9,315 times on \textbf{DR}, and 4,857 times on \textbf{DD}.

\begin{figure}[t]
    \centering
    \begin{subfigure}{0.9\linewidth}
        \includegraphics[width=\textwidth]{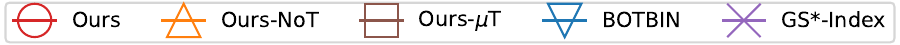}
    \end{subfigure}
    \begin{tabular}{ccc}
        \hspace{-3mm}\begin{subfigure}{0.333\linewidth}
                        \includegraphics[width=\textwidth]{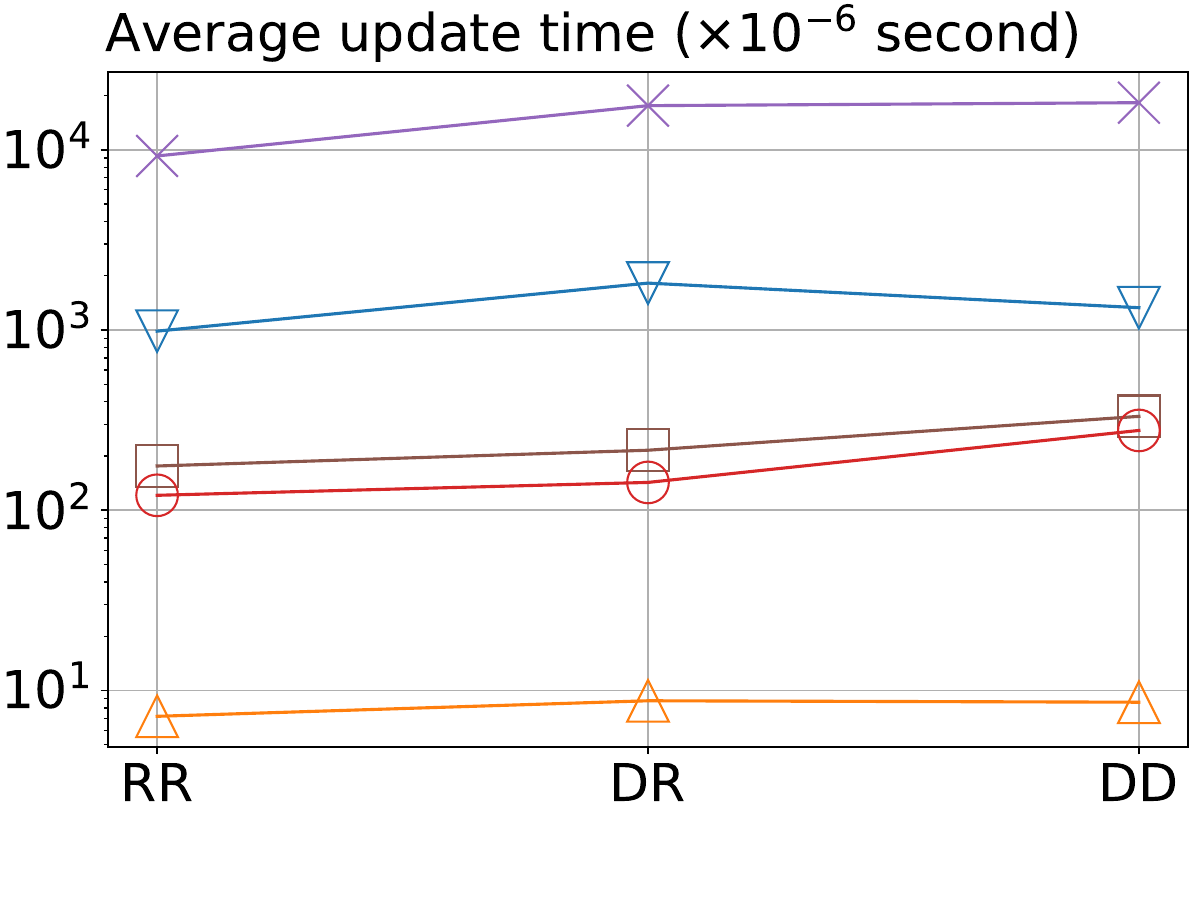}
            \vspace{-8mm}
            \caption{Google}
            \vspace{-1mm}
        \end{subfigure}&
            \hspace{-3mm}\begin{subfigure}{0.333\linewidth}
                \includegraphics[width=\textwidth]{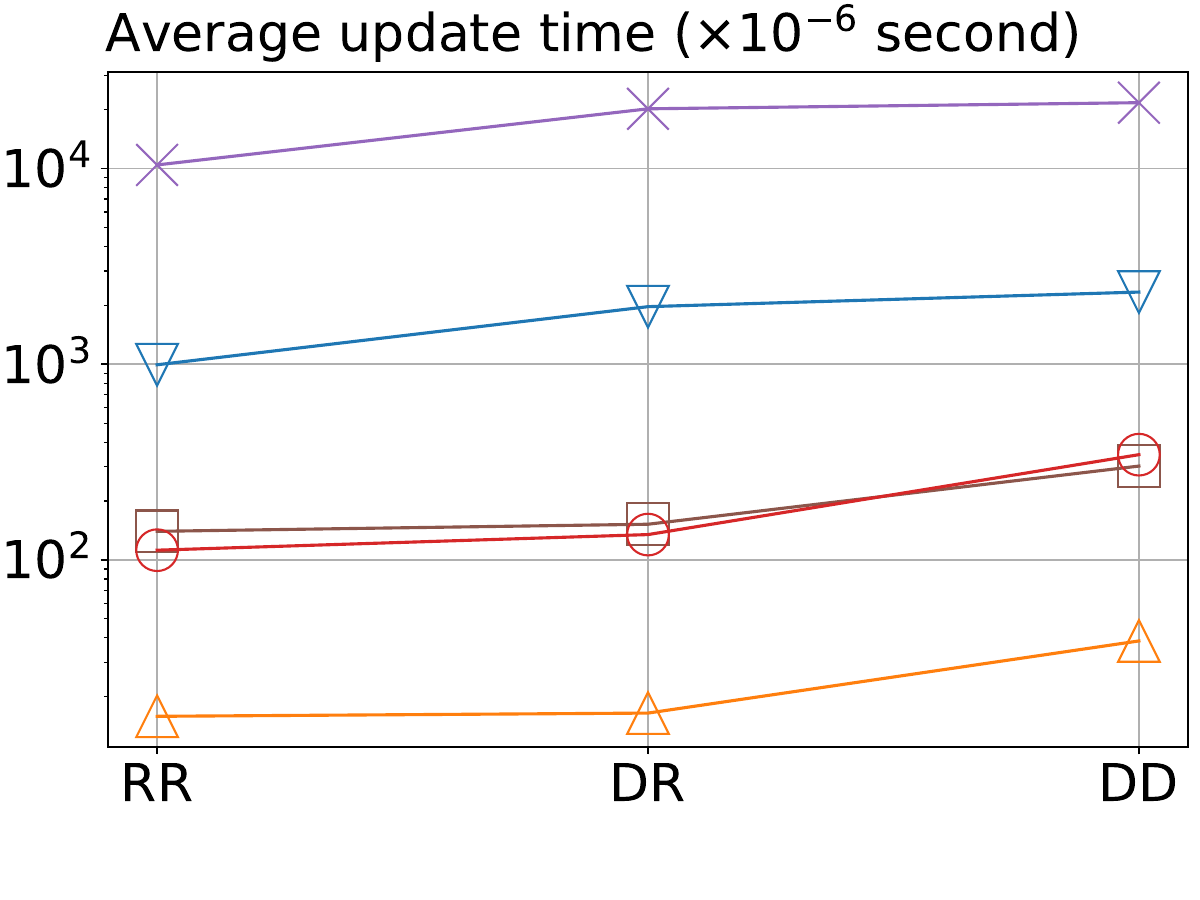}
                \vspace{-8mm}
                \caption{Topcats}
                \vspace{-1mm}
        \end{subfigure}&
            \hspace{-3mm}\begin{subfigure}{0.333\linewidth}
                \includegraphics[width=\textwidth]{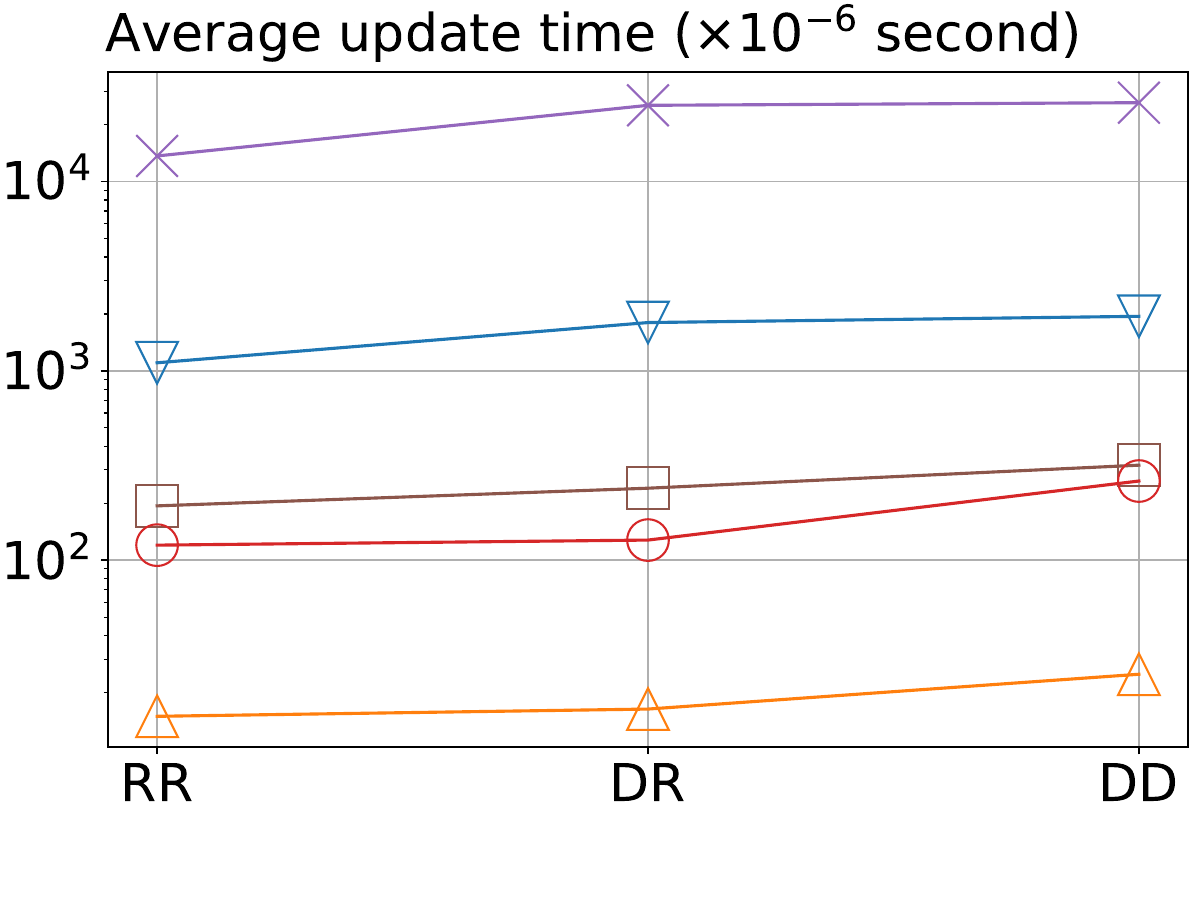}
                \vspace{-8mm}
                \caption{Pokec}
                \vspace{-1mm}
        \end{subfigure}\\
            \hspace{-3mm}\begin{subfigure}{0.333\linewidth}
                \includegraphics[width=\textwidth]{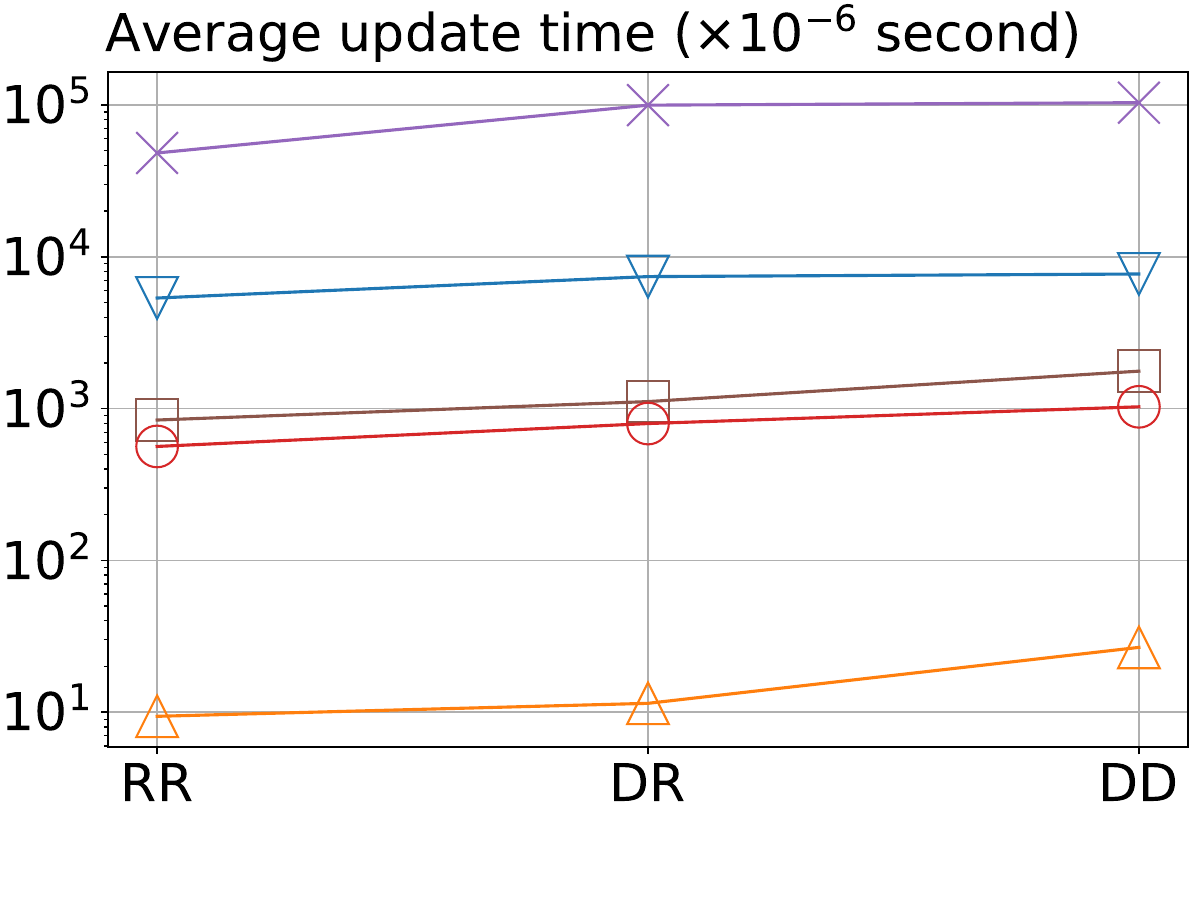}
                \vspace{-8mm}
                \caption{Skitter}
                \vspace{-2mm}
        \end{subfigure}&
            \hspace{-3mm}\begin{subfigure}{0.333\linewidth}
                \includegraphics[width=\textwidth]{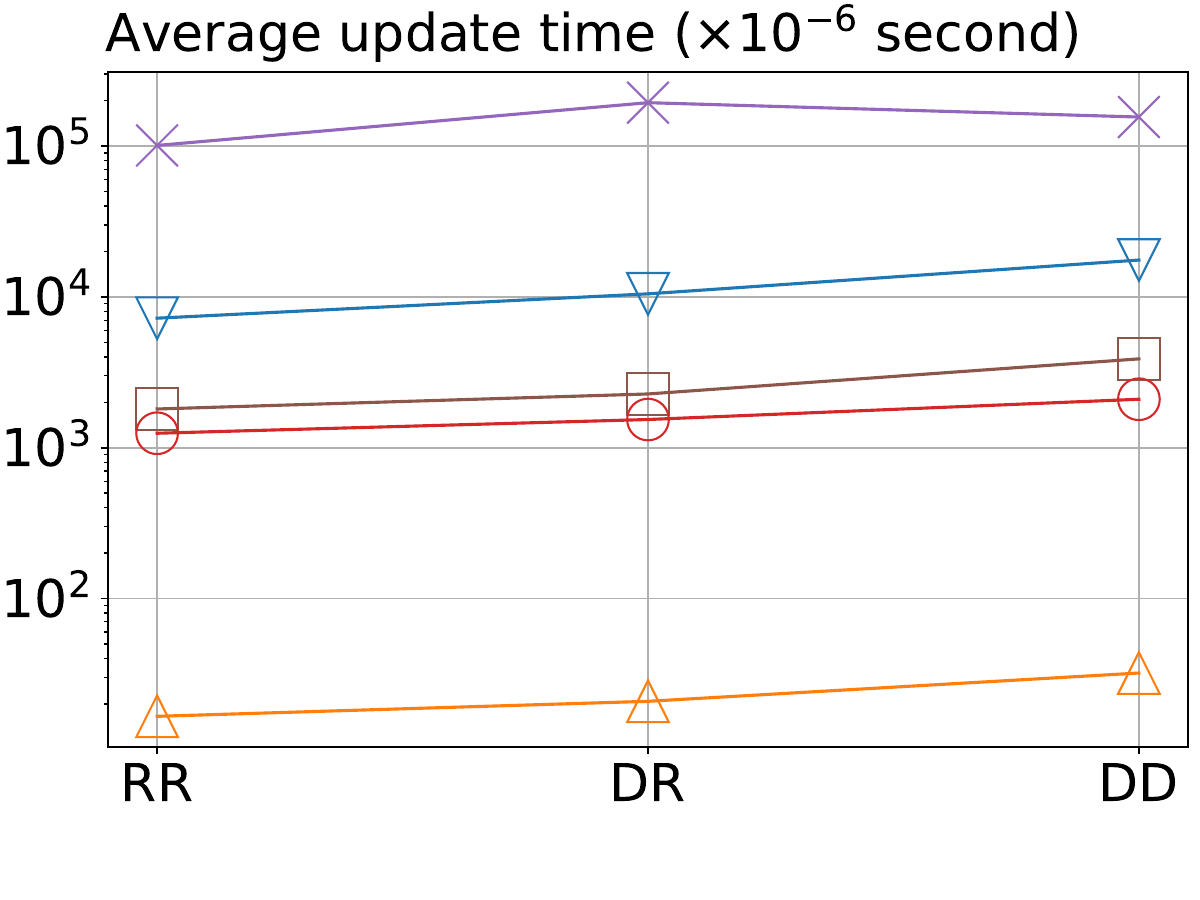}
                \vspace{-8mm}
                \caption{Talk}
                \vspace{-2mm}
        \end{subfigure}&
        \hspace{-3mm}\begin{subfigure}{0.333\linewidth}
            \includegraphics[width=\textwidth]{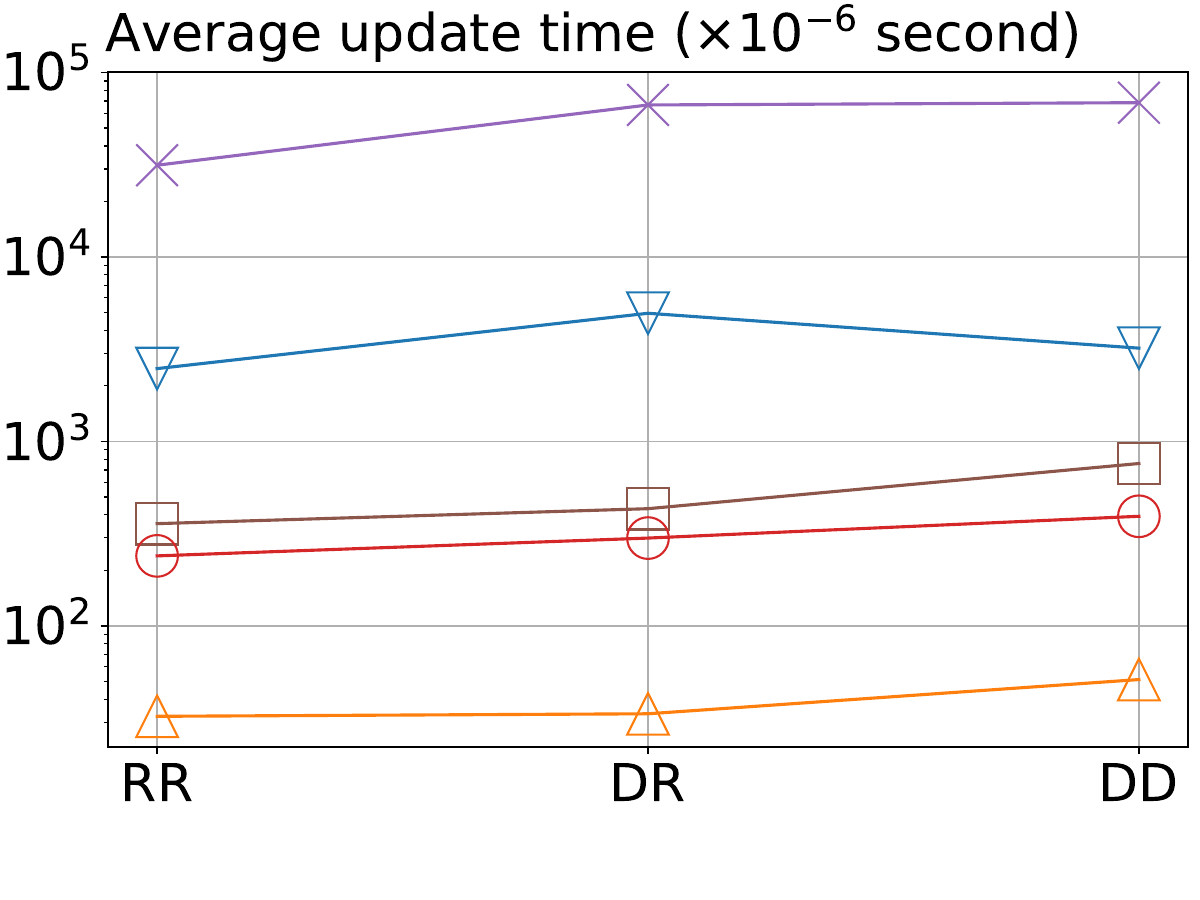}
            \vspace{-8mm}
            \caption{Orkut}
            \vspace{-2mm}
        \end{subfigure}\\
        \hspace{-3mm}\begin{subfigure}{0.333\linewidth}
            \includegraphics[width=\textwidth]{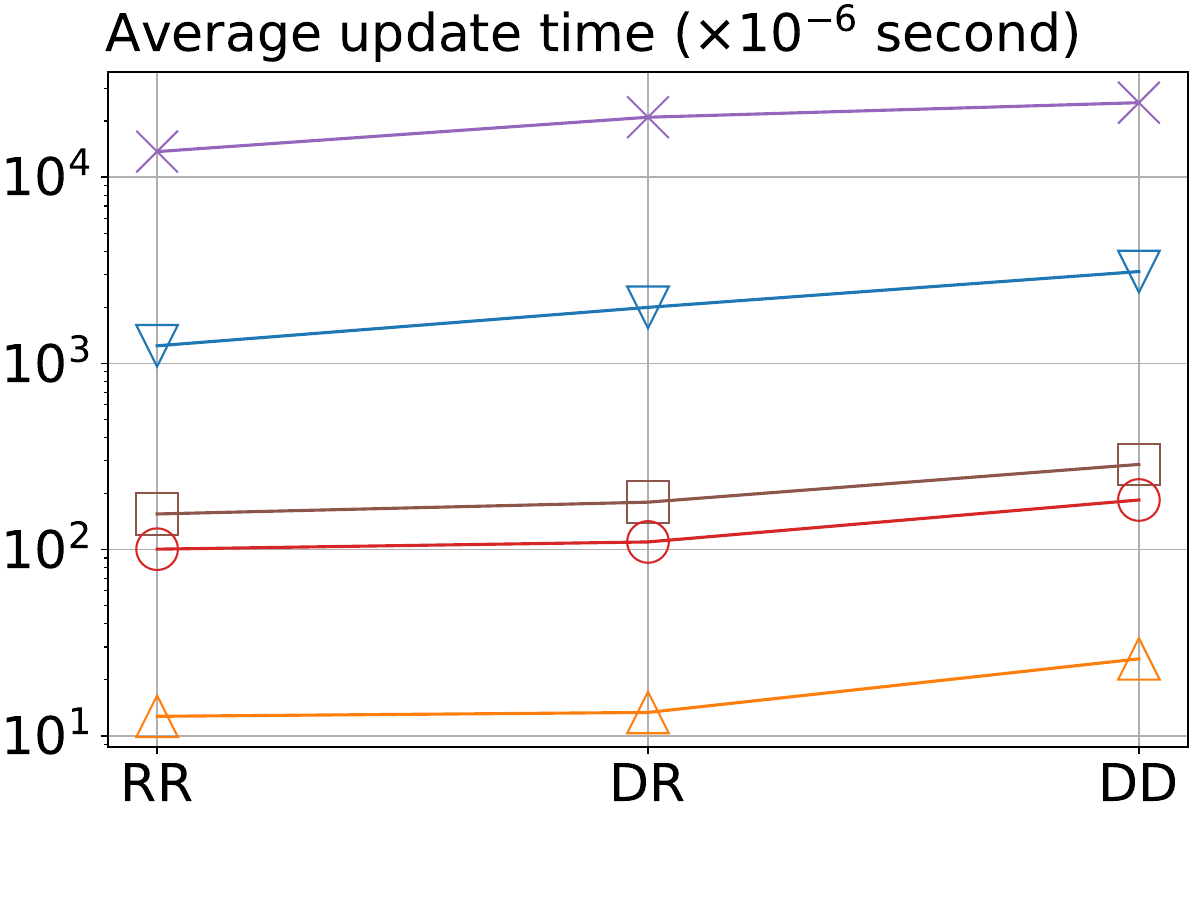}
            \vspace{-8mm}
            \caption{LiveJournal}
            \vspace{-2mm}
        \end{subfigure}&
        \hspace{-3mm}\begin{subfigure}{0.333\linewidth}
            \includegraphics[width=\textwidth]{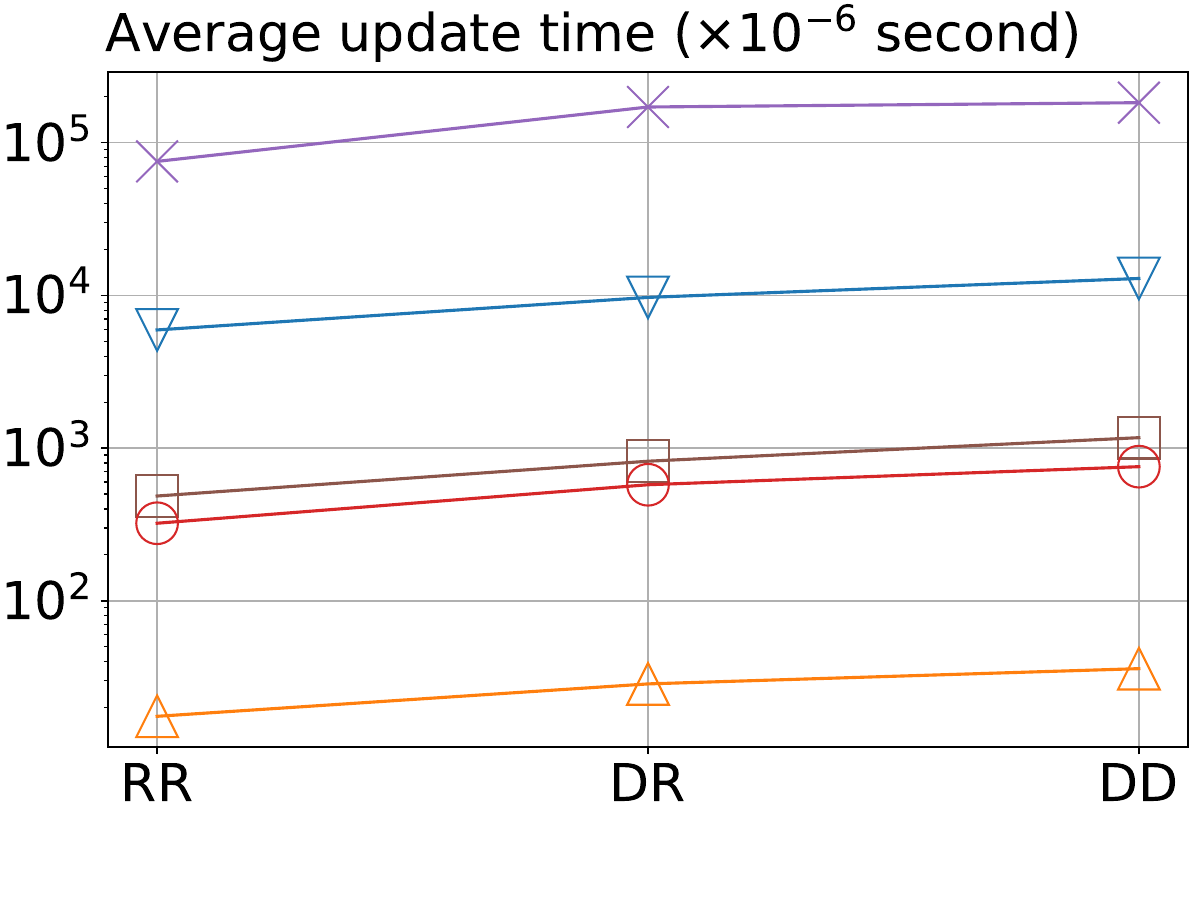}
            \vspace{-8mm}
            \caption{Friendster}
            \vspace{-2mm}
        \end{subfigure}&
        \hspace{-3mm}\begin{subfigure}{0.333\linewidth}
            \includegraphics[width=\textwidth]{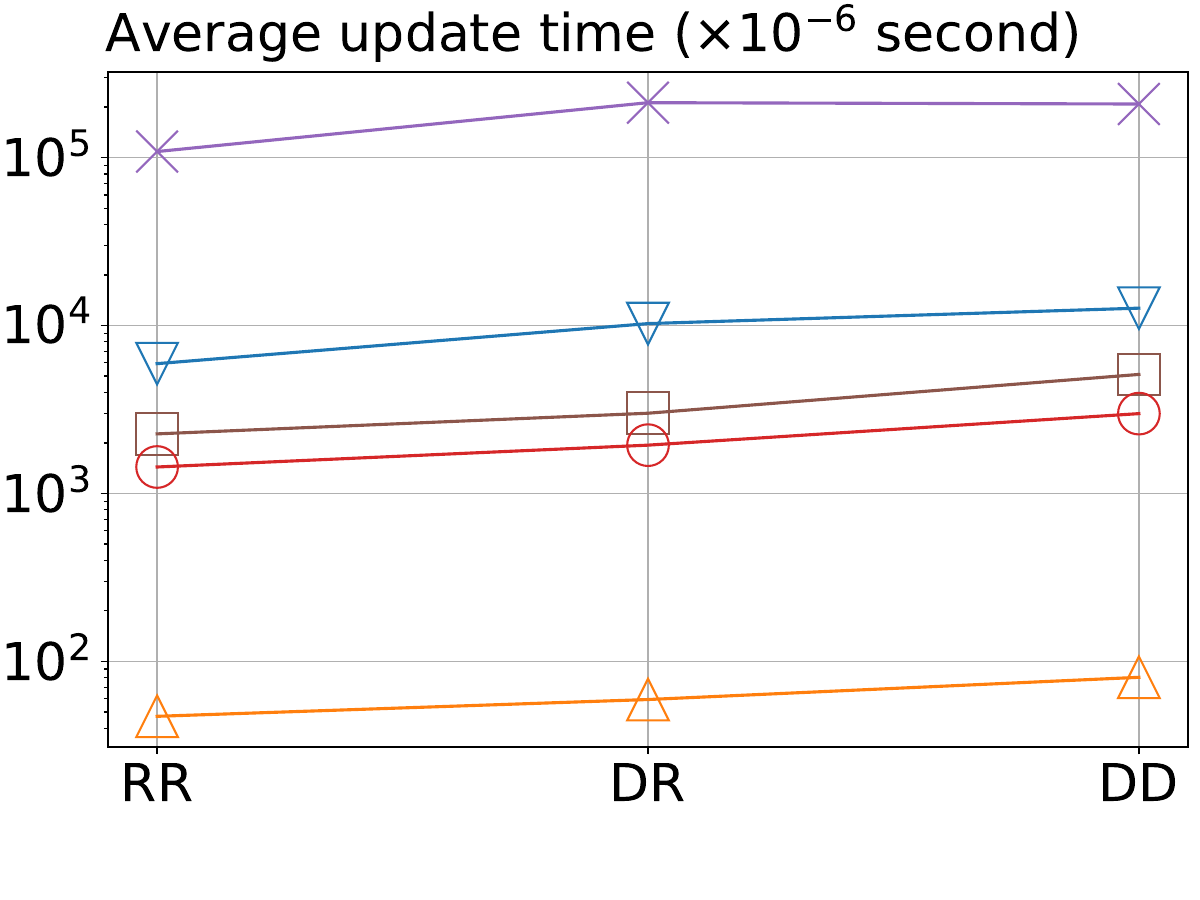}
            \vspace{-8mm}
            \caption{Web}
            \vspace{-2mm}
        \end{subfigure}\\
    \end{tabular}
    \vspace{-3mm}
    \caption{Average update time vs. update distribution}
    \label{fig:ug_varying}
\vspace{-6mm}
\end{figure}

\vspace{-2mm}
\subsubsection{\bf Impact of Deletion-to-Insertion Ratio}
\label{subsec:in_del_r}

\begin{figure}[t]
    \centering
    \begin{subfigure}{0.9\linewidth}
        \includegraphics[width=\textwidth]{figures/update_generation/lines_legend_ug.pdf}
    \end{subfigure}
    \begin{tabular}{ccc}
\hspace{-3mm}\begin{subfigure}{0.333\linewidth}
            \includegraphics[width=\textwidth]{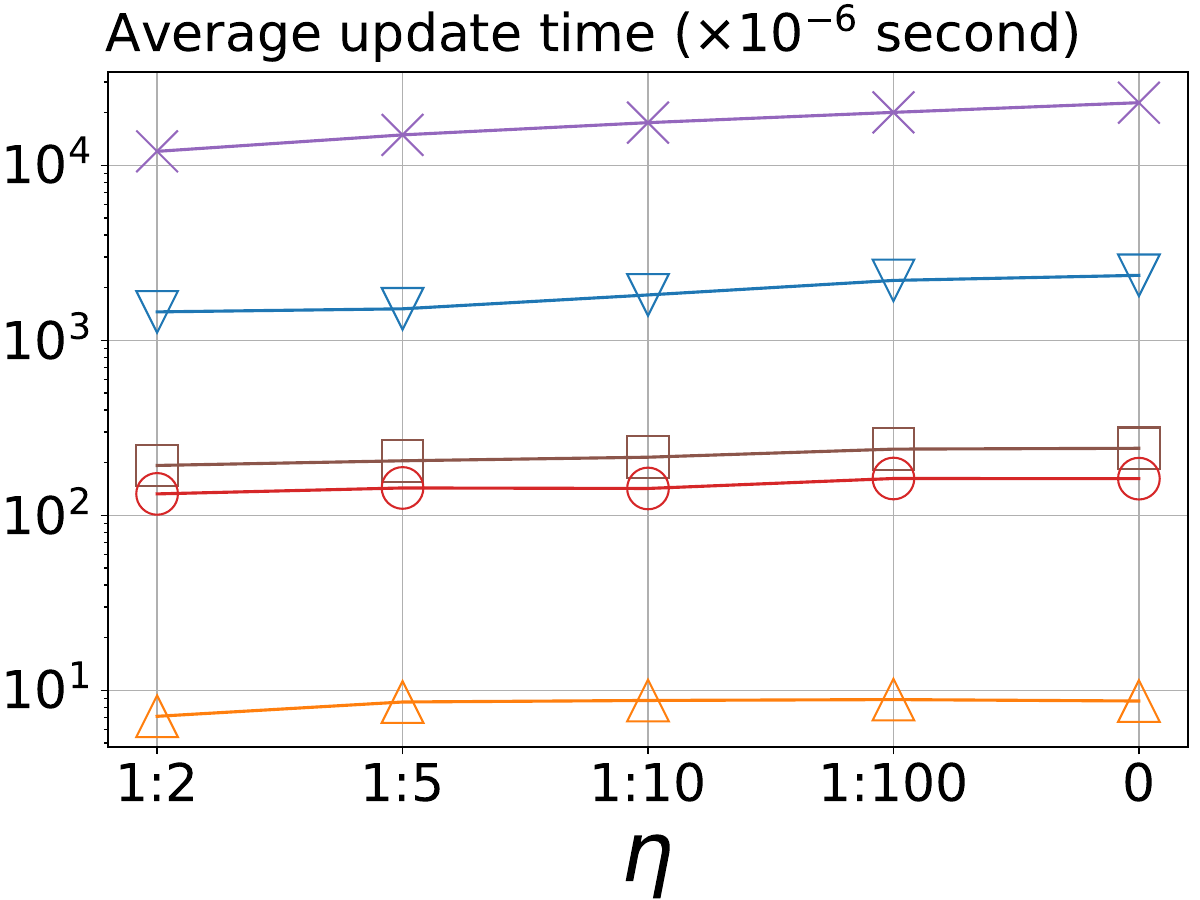}
            \vspace{-6mm}
            \caption{Google}
            \vspace{-1mm}
        \end{subfigure}&
            \hspace{-3mm}\begin{subfigure}{0.333\linewidth}
                \includegraphics[width=\textwidth]{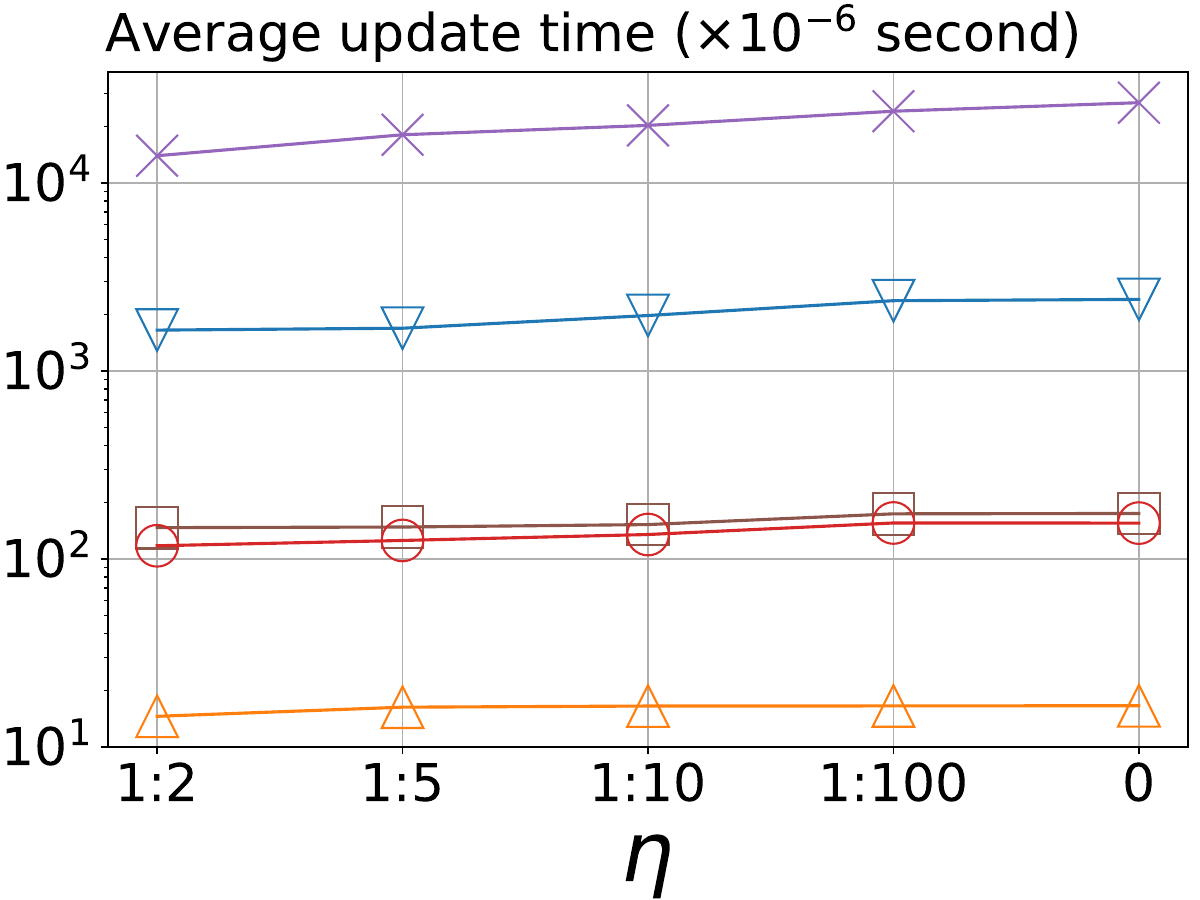}
                \vspace{-6mm}
                \caption{Topcats}
                \vspace{-1mm}
        \end{subfigure}&
            \hspace{-3mm}\begin{subfigure}{0.333\linewidth}
                \includegraphics[width=\textwidth]{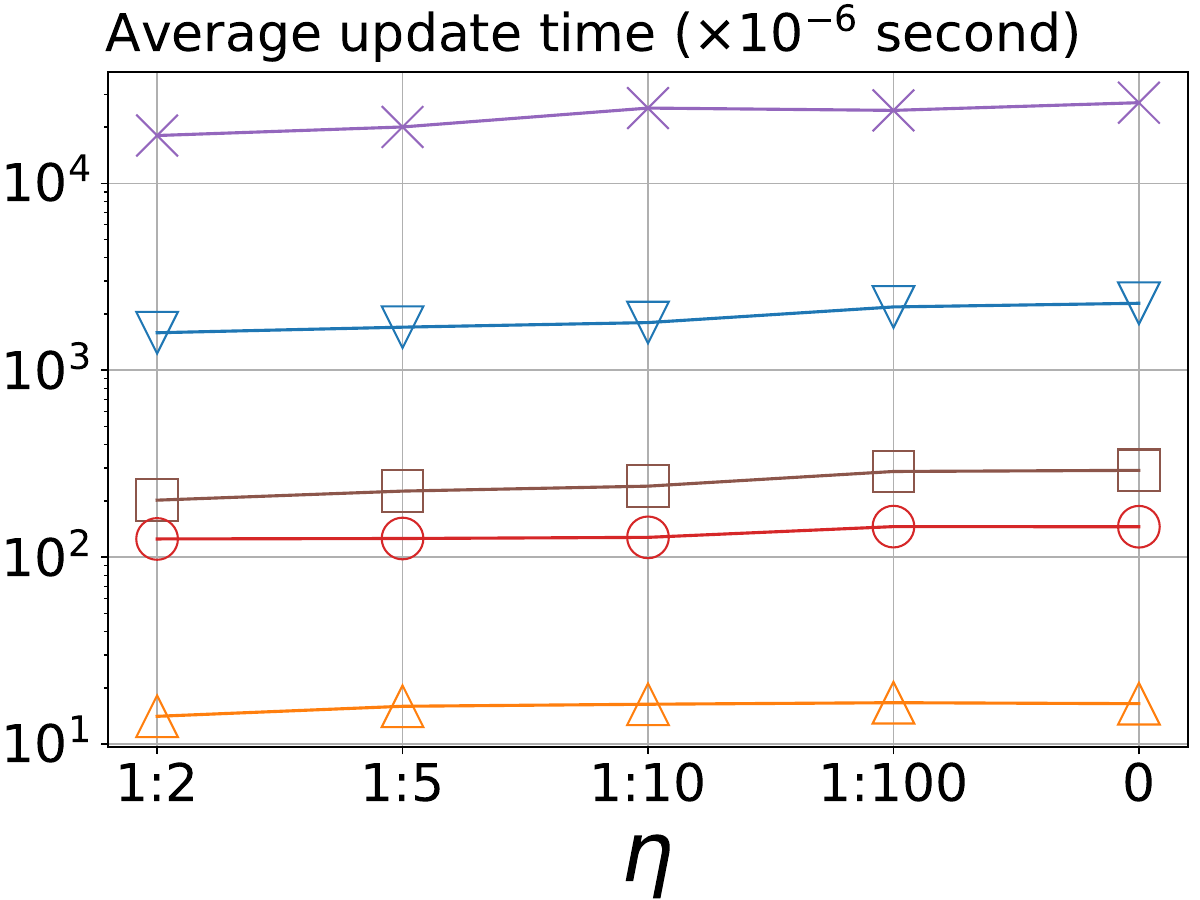}
                \vspace{-6mm}
                \caption{Pokec}
                \vspace{-1mm}
        \end{subfigure}\\
            \hspace{-3mm}\begin{subfigure}{0.333\linewidth}
                \includegraphics[width=\textwidth]{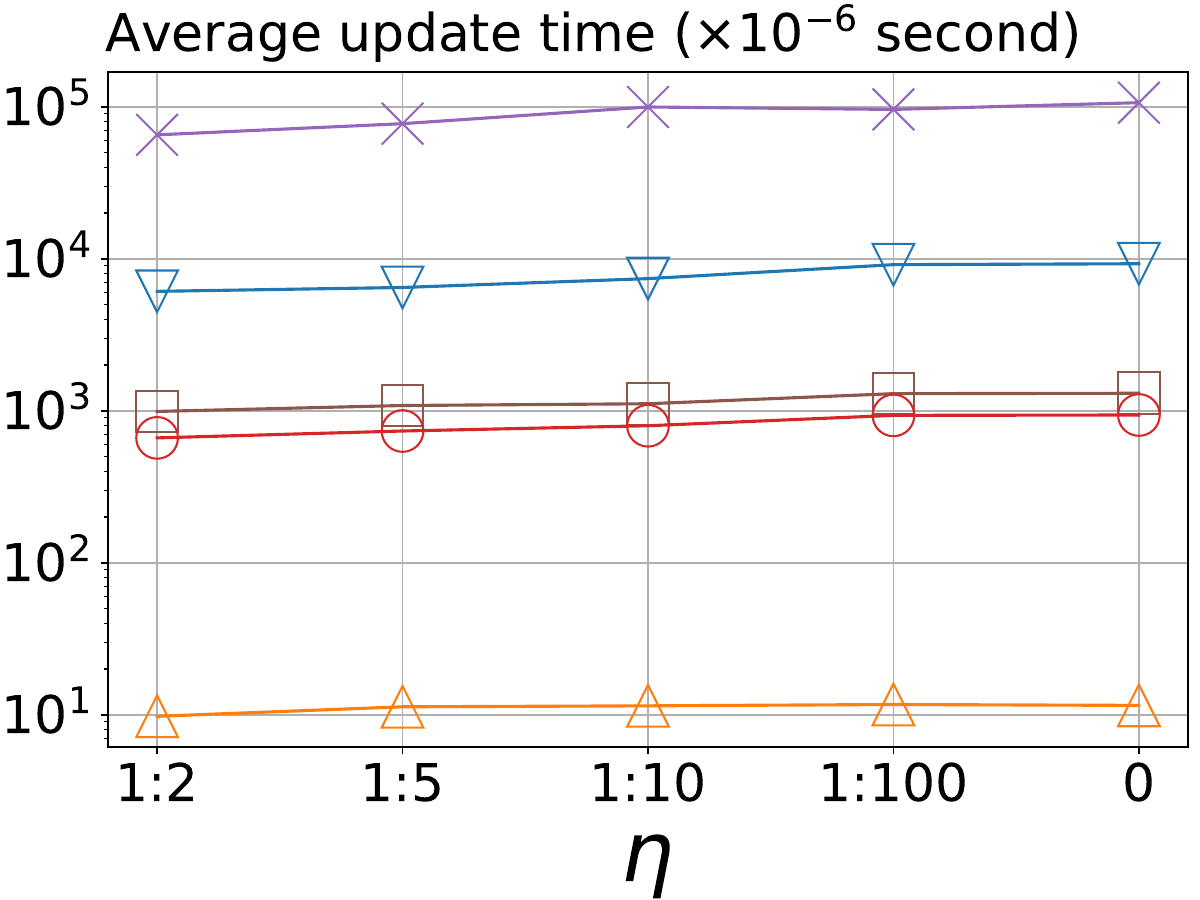}
                \vspace{-6mm}
                \caption{Skitter}
                \vspace{-2mm}
        \end{subfigure}&
            \hspace{-3mm}\begin{subfigure}{0.333\linewidth}
                \includegraphics[width=\textwidth]{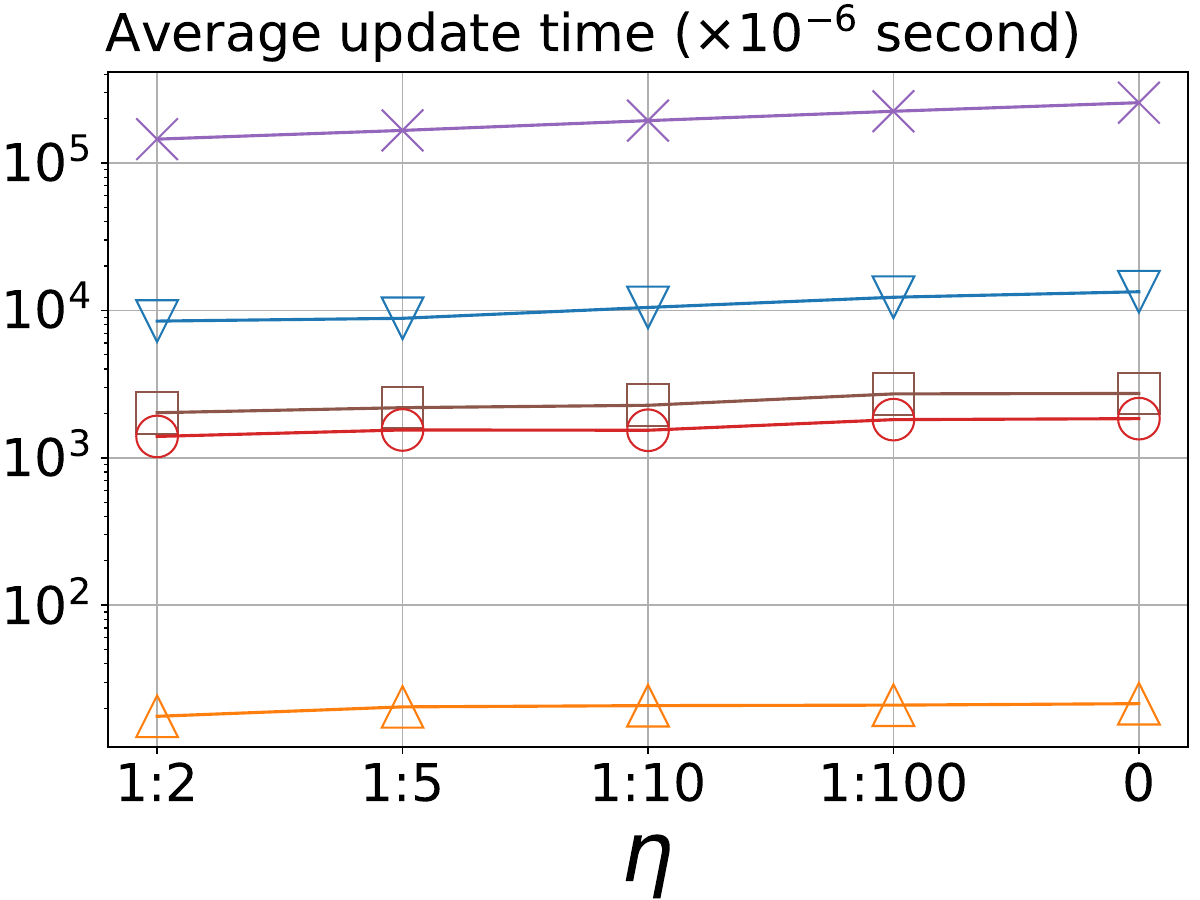}
                \vspace{-6mm}
                \caption{Talk}
                \vspace{-2mm}
        \end{subfigure}&
        \hspace{-3mm}\begin{subfigure}{0.333\linewidth}
            \includegraphics[width=\textwidth]{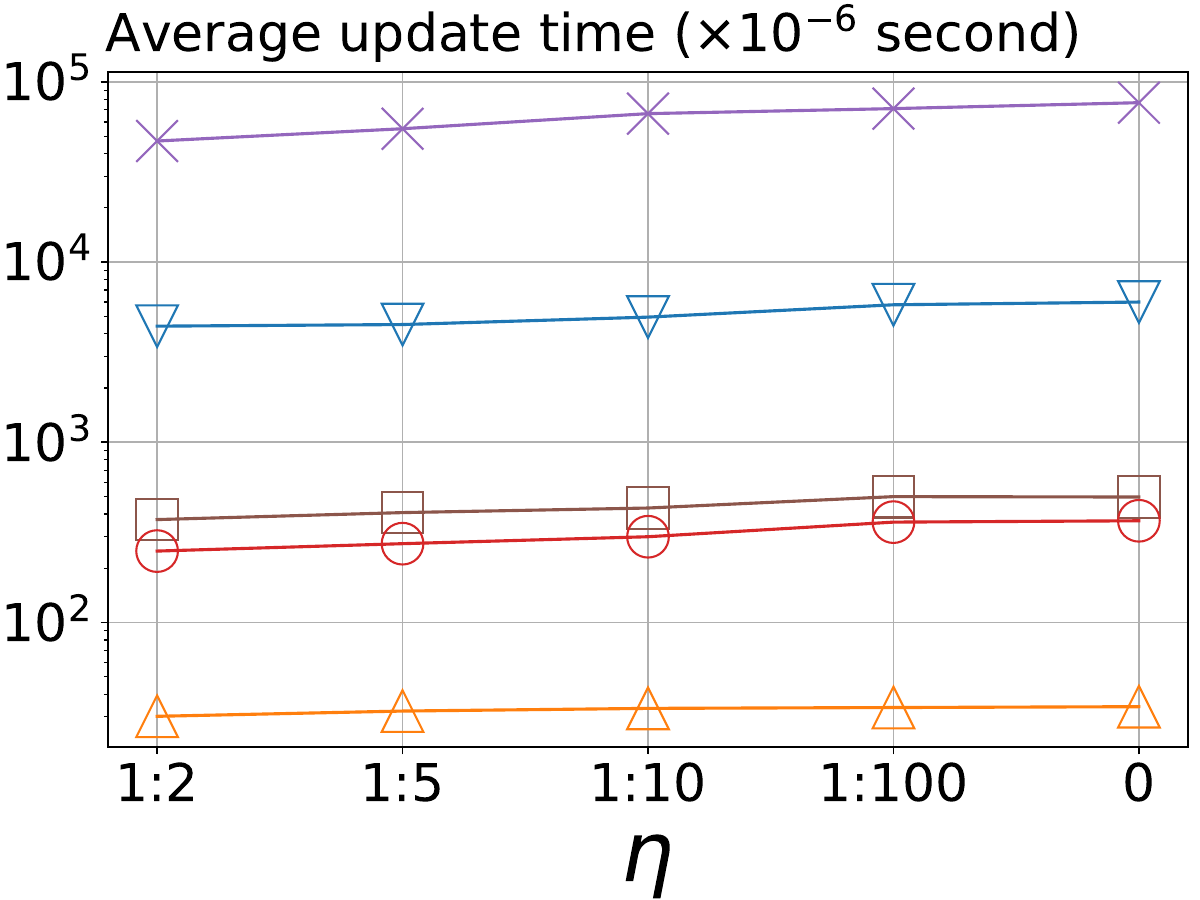}
            \vspace{-6mm}
            \caption{Orkut}
            \vspace{-2mm}
        \end{subfigure}\\
        \hspace{-3mm}\begin{subfigure}{0.333\linewidth}
            \includegraphics[width=\textwidth]{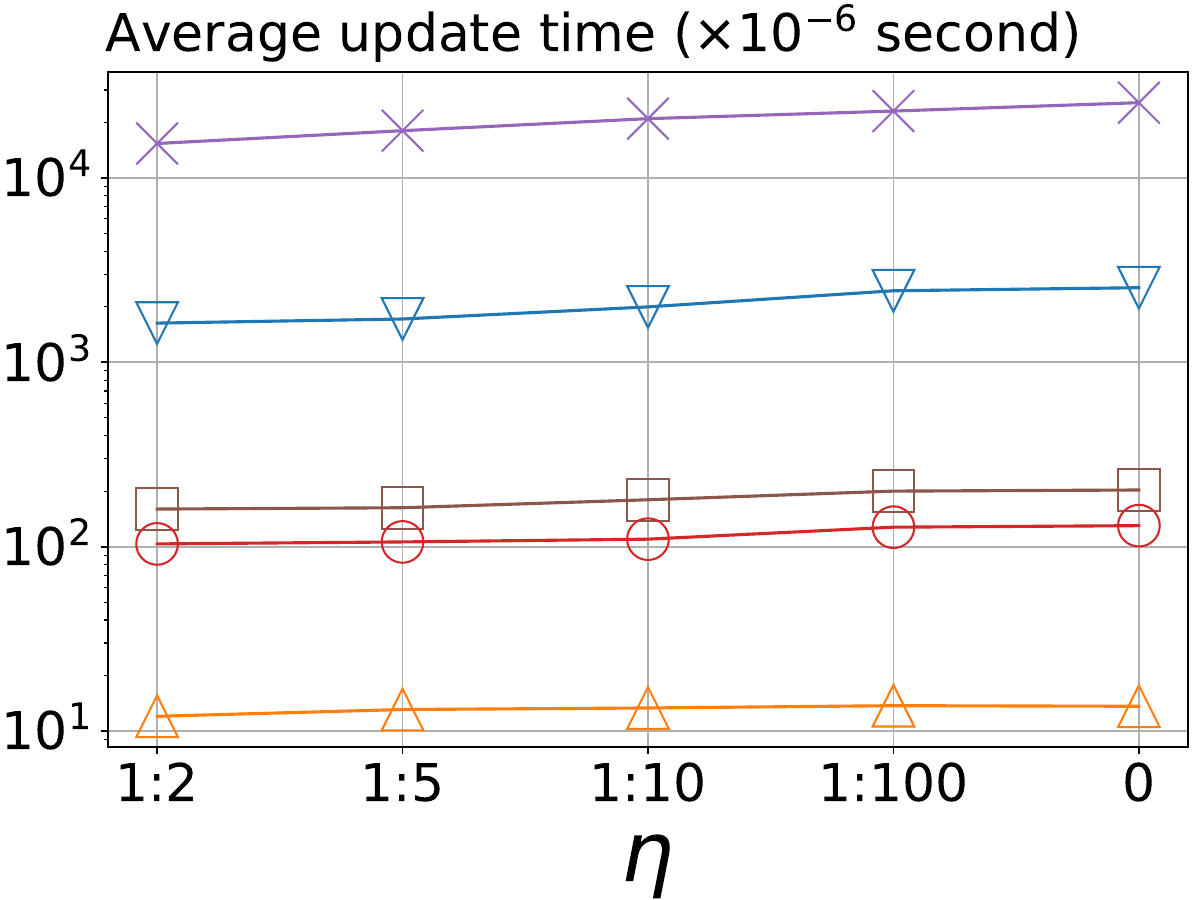}
            \vspace{-6mm}
            \caption{LiveJournal}
            \vspace{-2mm}
        \end{subfigure}&
        \hspace{-3mm}\begin{subfigure}{0.333\linewidth}
            \includegraphics[width=\textwidth]{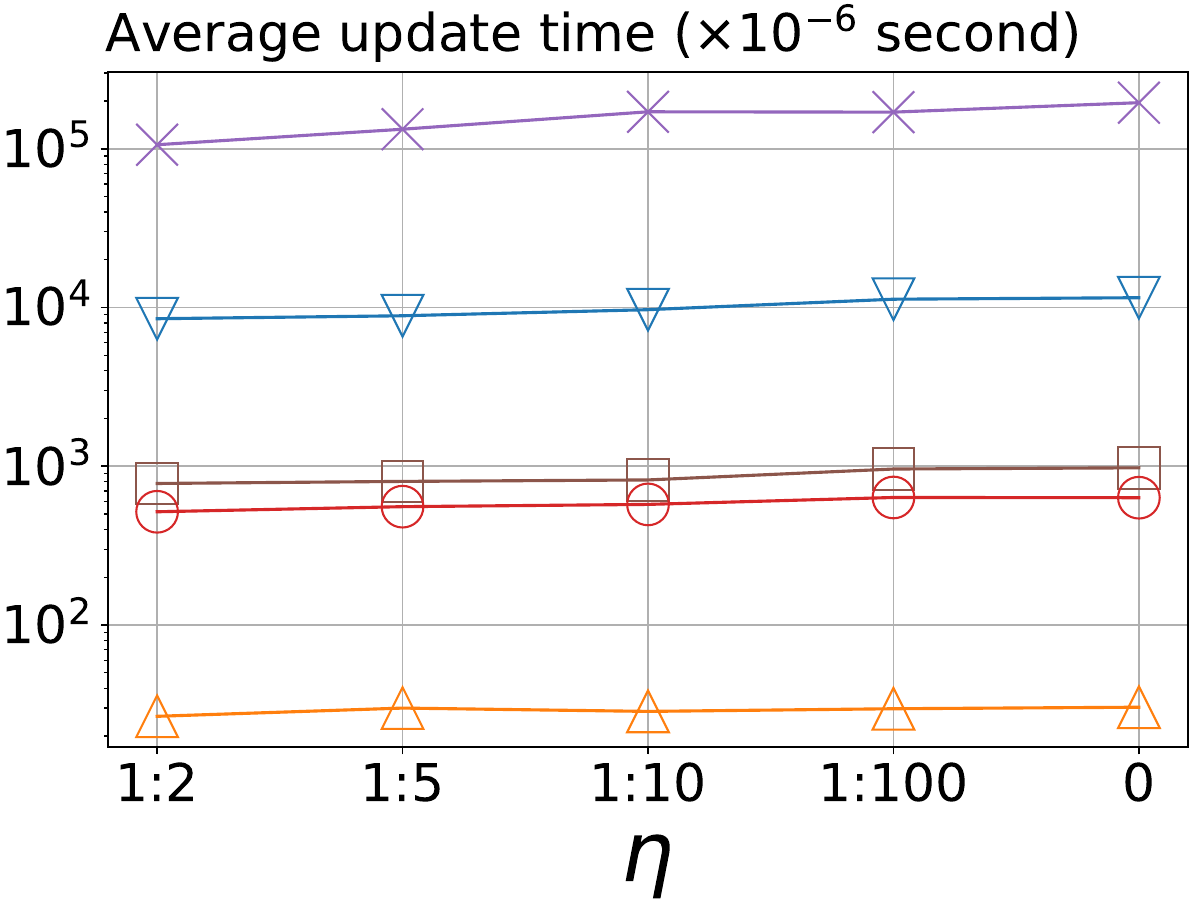}
            \vspace{-6mm}
            \caption{Friendster}
            \vspace{-2mm}
        \end{subfigure}&
        \hspace{-3mm}\begin{subfigure}{0.333\linewidth}
            \includegraphics[width=\textwidth]{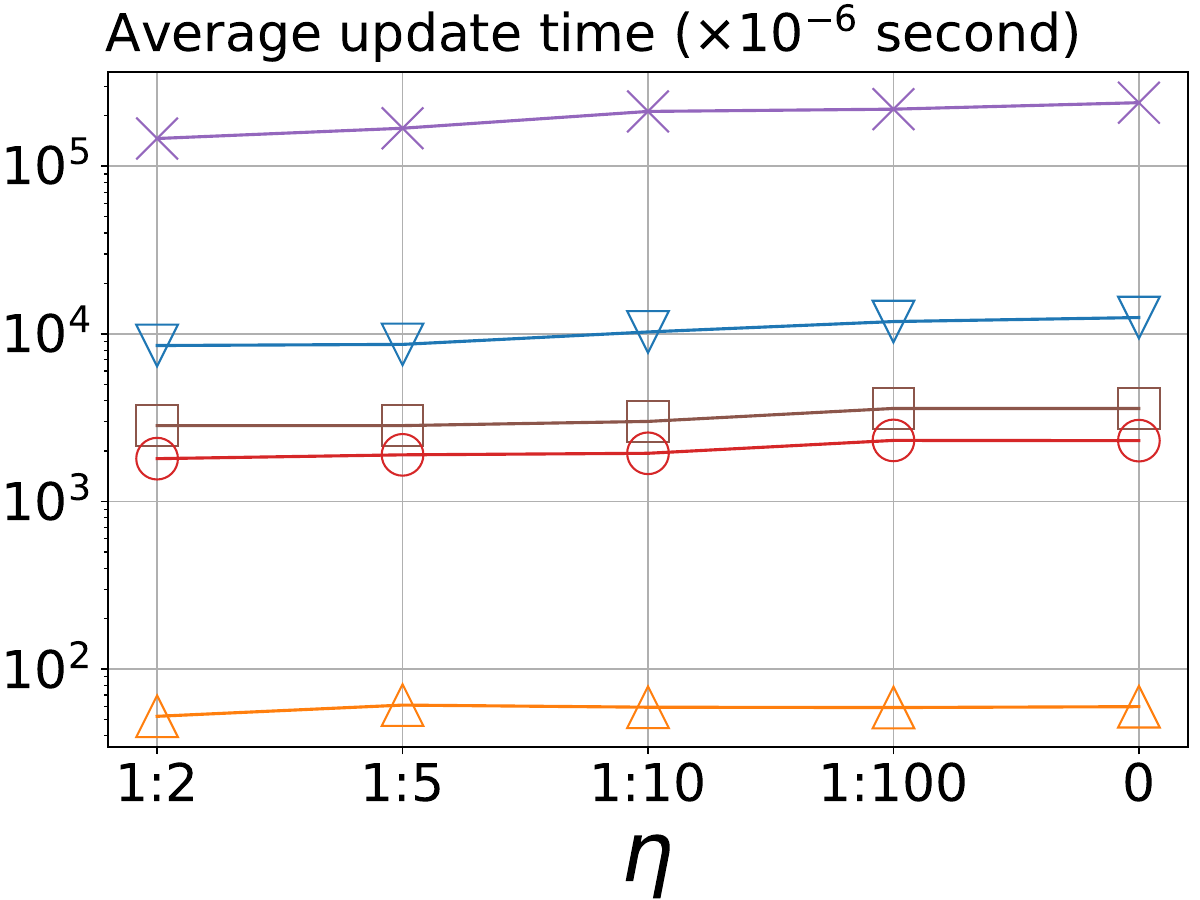}
            \vspace{-6mm}
            \caption{Web}
            \vspace{-2mm}
        \end{subfigure}\\
    \end{tabular}
    \vspace{-4mm}
    \caption{Average update time vs. $\eta$}
    \label{fig:eta_varying}
\end{figure}

Next, we vary deletion-to-insertion ratios, and  Figure~\ref{fig:eta_varying} reports the results. We observe the following: (1)~The average update times increase 
with insertions occurring more frequently
across all algorithms and datasets. 
This is expected since the number of edges grows with an increase in insertions, and the growth accelerates when the deletion-to-insertion ratio $\eta$ decreases. 
(2)~Across all settings tested, all our algorithms outperform the SOTA algorithms, with \algo~ having the lowest update times. (3)~When $\eta = 0$, all updates are insertions. In this case, \algo~ amplifies the speedup from up to 9,315 times to 11,959 times (vs. GS*-Index), and from up to 647 times to 805 times (vs. BOTBIN), compared with $\eta = \frac{1}{10}$. This reaffirms the robustness of our algorithms in terms of scalability.

\vspace{-2mm}
\subsection{Study on Query Efficiency}
\label{subsubsec:qe}
\vspace{-1mm}
The query efficiency results are shown in Figure~\ref{fig:query}.
Several observations are made: (1) GS*-Index, BOTBIN, and \algodelta\ exhibit similar query times as their query time are all bounded by $O(m_{cr})$, linear to the size of the clustering result graph. 
\algomu, utilizing a fixed-size $\mu$-table (in Section~\ref{subsec:opt2}), 
only incurs a cost bounded by $O(\min\{\mu, \log n\} \cdot n + m_{cr})$ when the input $\mu$
exceeds a certain threshold $\mu_{\max}$ which is set as $15$ in this experiment. 
In practice, it can still achieve performance similar to $O(m_{cr})$ methods, 
as evidenced by the results showing very similar query times. 
(2) Across all datasets, our \algo~algorithm's query times are within the same magnitude as the other algorithms. Notably, on \texttt{web-Google} and \texttt{web-topcats}, \algo\ has a lower query time. This is because the size of the clustering result graph is large in these datasets, and hence, $O(m_{cr})$ dominates the query time, making the query overhead of \algo\ 
negligible.
%
(4) In structural clustering problems, the number of the clustering result vertices ($n_{cr}$) has substantial practical implications. If $n_{cr}$ is small, the structural clustering results may lose significance because most vertices are excluded. In cases where $n_{cr}$ approaches $n$, \algo~ introduces a negligible overhead in queries 
while accelerating updates by over 100 times. 
%
\begin{figure}[t]
\begin{subfigure}{0.8\linewidth}
    \centering
    \includegraphics[width=\textwidth]{figures/bar_legend.pdf}
\end{subfigure}
\begin{subfigure}{1\linewidth}
    \hspace{3mm}
    \includegraphics[width=0.9\textwidth]{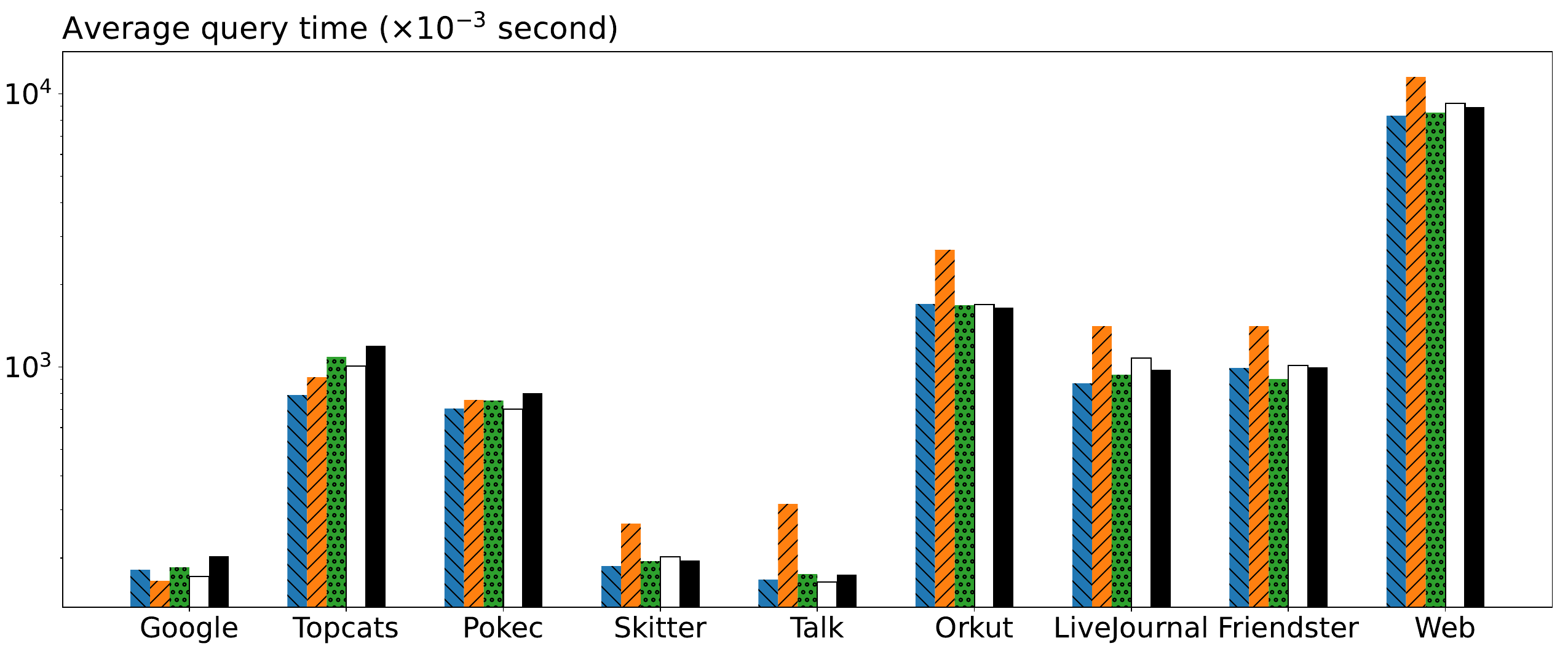}
\end{subfigure}
    \caption{Query processing performance results}
    \label{fig:query}
\end{figure}

\begin{figure*}[t]
\begin{subfigure}{0.8\linewidth}
    \centering
    \includegraphics[width=\textwidth]{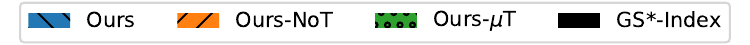}
\end{subfigure}
\begin{subfigure}{0.5\linewidth}
    \includegraphics[width=1\textwidth]{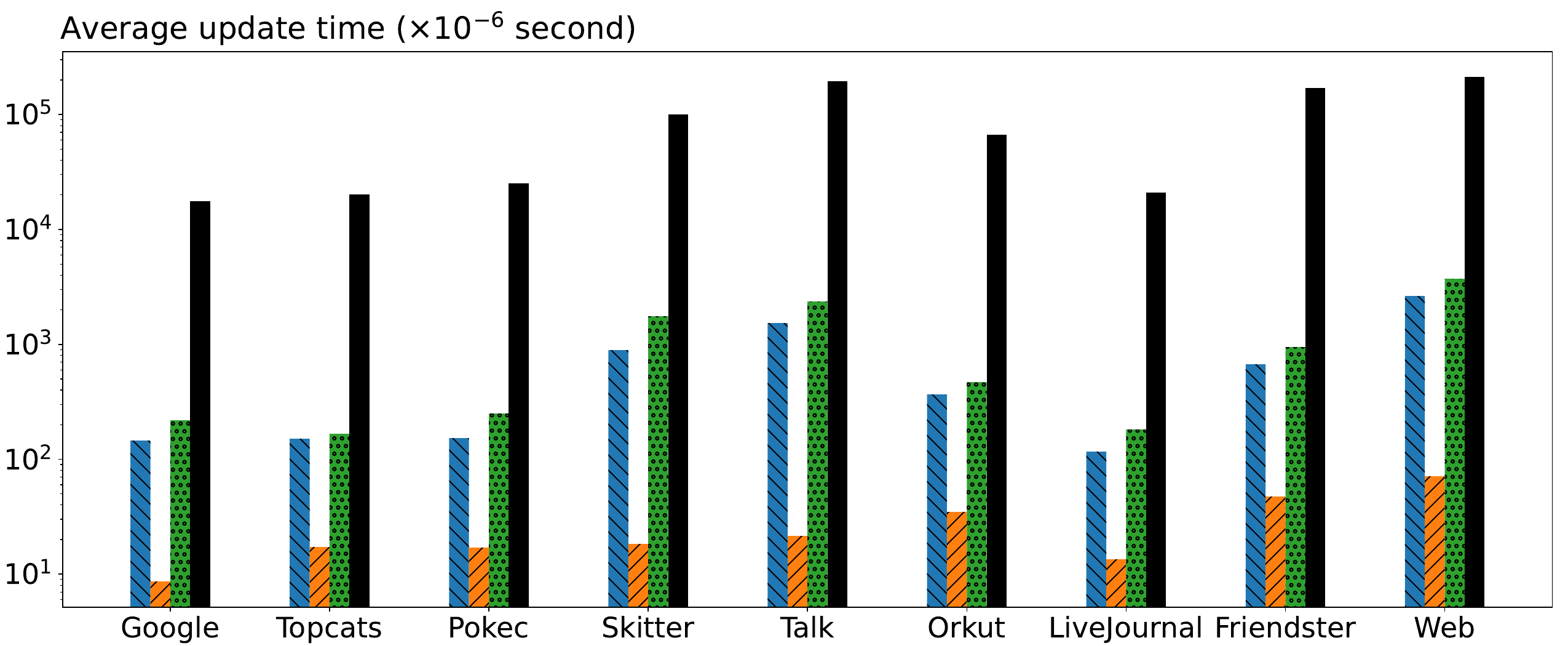}
    \vspace{-6mm}
    \caption{Cosine}    
    \label{fig:cos_update}
\end{subfigure}%
\begin{subfigure}{0.5\linewidth}
    \includegraphics[width=1\textwidth]{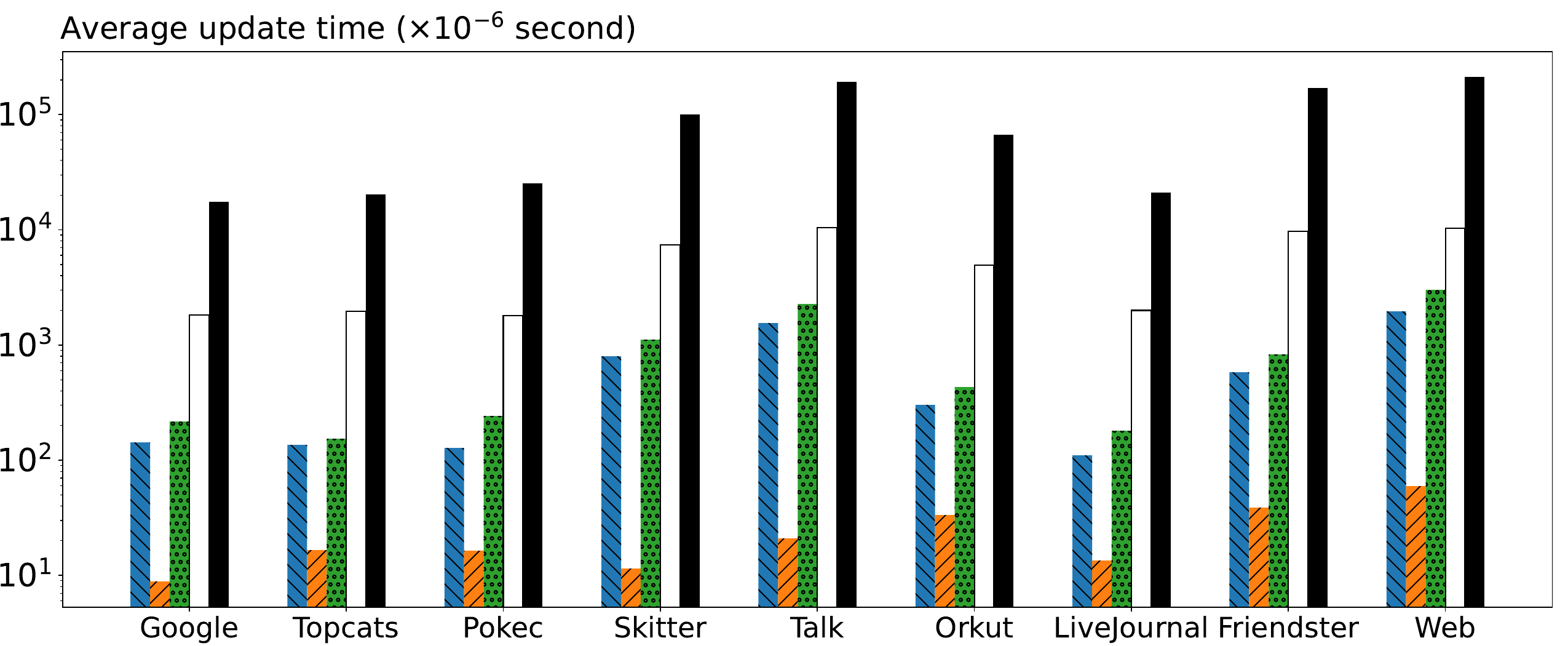}
    \vspace{-6mm}
    \caption{Dice}
    \label{fig:dic_update}
\end{subfigure}
    \caption{Average update running time on Cosine and Dice}
    \label{fig:cos}
    \vspace{5mm}
\end{figure*}

\subsection{Study on Clustering Quality}
%
We look into the clustering quality of both our algorithms and BOTBIN, 
in terms of the \emph{mislabeled rate} (\textbf{MLR}) and \emph{adjusted rand index} (\textbf{ARI})~\cite{hubert1985comparing}. 
MLR is calculated as dividing the number of incorrectly labeled edges by the number of edges, $m$, of the current graph. ARI is widely used to evaluate the clustering quality which outputs a value from 0 to 1, where 1 means that the clusters are exactly the same as the ground truth. 
We evaluate the result quality using the default $\rho^* = 0.02$ and a larger $\rho^* = 0.1$. 
Note that here $\rho^*$ represents an absolute error, 
and 0.1 is already a relatively large error value. 
Both MLR and ARI are measured for each query as described above and their average values are reported in Table~\ref{tab:quality}.

\begin{table}
\caption{Clustering quality results}
\resizebox{1\textwidth}{!}{%
\begin{tabular}{l|cccc|cccc}
\hline
                 & \multicolumn{4}{c|}{$\boldsymbol{\rho=}\;\mathbf{0.02}$}                                                                                & \multicolumn{4}{c}{$\boldsymbol{\rho=}\;\mathbf{0.1}$}                                                                                  \\ \cline{2-9} 
                 & \multicolumn{2}{c|}{\textbf{BOTBIN}}     & \multicolumn{2}{c|}{\textbf{\algodelta}}     & \multicolumn{2}{c|}{\textbf{BOTBIN}}       & \multicolumn{2}{c}{\textbf{\algodelta}}       \\ \cline{2-9} 
                 & \multicolumn{1}{c|}{\textbf{ARI} $\uparrow$}    & \multicolumn{1}{c|}{\textbf{MLR} $\downarrow$} & \multicolumn{1}{c|}{\textbf{ARI} $\uparrow$}    & \textbf{MLR} $\downarrow$ & \multicolumn{1}{c|}{\textbf{ARI} $\uparrow$}    & \multicolumn{1}{c|}{\textbf{MLR} $\downarrow$} & \multicolumn{1}{c|}{\textbf{ARI} $\uparrow$}    & \textbf{MLR} $\downarrow$ \\ \hline
web-Google       & \multicolumn{1}{c|}{\textbf{0.9991}} & \multicolumn{1}{c|}{\textbf{0.14\%}}        & \multicolumn{1}{c|}{0.9990} & \textbf{0.14\% }       & \multicolumn{1}{c|}{0.9695} & \multicolumn{1}{c|}{5.52\%}        & \multicolumn{1}{c|}{\textbf{0.9701}} & \textbf{5.50\% }       \\ \hline
wiki-topcats     & \multicolumn{1}{c|}{0.9990} & \multicolumn{1}{c|}{\textbf{0.06\%}}        & \multicolumn{1}{c|}{\textbf{0.9996}} & \textbf{0.06\%}        & \multicolumn{1}{c|}{\textbf{0.9896}} & \multicolumn{1}{c|}{\textbf{0.86\%}}        & \multicolumn{1}{c|}{0.9887} & 0.92\%        \\ \hline
soc-Pokec        & \multicolumn{1}{c|}{\textbf{0.9957}} & \multicolumn{1}{c|}{0.18\%}        & \multicolumn{1}{c|}{\textbf{0.9955}} & \textbf{0.17\%}        & \multicolumn{1}{c|}{\textbf{0.9676}} & \multicolumn{1}{c|}{\textbf{5.32\%}}        & \multicolumn{1}{c|}{0.9653} & 5.33\%        \\ \hline
as-skitter       & \multicolumn{1}{c|}{0.9995} & \multicolumn{1}{c|}{\textbf{0.19\%}}        & \multicolumn{1}{c|}{\textbf{0.9996}} & \textbf{0.19\%}        & \multicolumn{1}{c|}{0.9826} & \multicolumn{1}{c|}{6.38\%}        & \multicolumn{1}{c|}{\textbf{0.9842}} & \textbf{6.34}\%        \\ \hline
wiki-Talk        & \multicolumn{1}{c|}{0.9987} & \multicolumn{1}{c|}{0.45\%}        & \multicolumn{1}{c|}{\textbf{0.9989}} & \textbf{0.44\%}        & \multicolumn{1}{c|}{0.9716} & \multicolumn{1}{c|}{\textbf{7.29\%}}        & \multicolumn{1}{c|}{\textbf{0.9722}} & 7.30\%        \\ \hline
soc-Orkut            & \multicolumn{1}{c|}{0.9946} & \multicolumn{1}{c|}{\textbf{0.12\%}}        & \multicolumn{1}{c|}{\textbf{0.9954}} & \textbf{0.12\%}        & \multicolumn{1}{c|}{\textbf{0.9548}} & \multicolumn{1}{c|}{\textbf{4.38\%}}        & \multicolumn{1}{c|}{\textbf{0.9548}} & 4.40\%        \\ \hline
soc-LiveJournal1 & \multicolumn{1}{c|}{\textbf{0.9998}} & \multicolumn{1}{c|}{\textbf{0.12\%}}        & \multicolumn{1}{c|}{0.9995} & \textbf{0.12\%}        & \multicolumn{1}{c|}{0.9975} & \multicolumn{1}{c|}{4.56\%}        & \multicolumn{1}{c|}{\textbf{0.9982}} & \textbf{4.56\%}        \\ \hline
soc-Friendster          & \multicolumn{1}{c|}{0.9944} & \multicolumn{1}{c|}{\textbf{0.64\%}}        & \multicolumn{1}{c|}{\textbf{0.9947}} & 0.69\%       & \multicolumn{1}{c|}{0.9636} & \multicolumn{1}{c|}{6.73\%}        & \multicolumn{1}{c|}{\textbf{0.9673}} & \textbf{6.40\%}        \\ \hline
web-2012 & \multicolumn{1}{c|}{\textbf{0.9928}} & \multicolumn{1}{c|}{0.60\%}        & \multicolumn{1}{c|}{0.9926} &  \textbf{0.56\%}       & \multicolumn{1}{c|}{\textbf{0.9609}} & \multicolumn{1}{c|}{\textbf{9.79\%}}        & \multicolumn{1}{c|}{0.9607} & 9.91\%       \\ \bottomrule
\end{tabular}
}
\label{tab:quality}
\end{table}

Both BOTBIN and \algodelta\ have high-quality results, leveraging error bounds to their advantage. Notably, \algodelta\ outperforms BOTBIN on more datasets. 
When $\rho^* = 0.02$, \algodelta\ achieves MLR of less than 0.7\% across all datasets, with ARI values ranging from 0.9926 to 0.9996. 
Even with $\rho^* =0.1$, \algodelta\ maintains an average ARI of at least 0.9548 (on \texttt{soc-Orkut}) and MLR of at most 9.91\% (on \texttt{web-2012}).
\algomu\ and \algo\ have similar results to \algodelta. For brevity, they are not detailed here.
These results again underscore the practical significance of theoretical error bounds with a high success rate on real datasets.

\begin{table}[ht]
\caption{Clustering quality results on all three measurements (The result of Jaccard Similarity with $\rho = 0.02$ is copied and pasted from Table~\ref{tab:quality} to here for easy comparison.)}
\resizebox{1\textwidth}{!}{%
\begin{tabular}{l|c|c|c|c|c|c}
\hline
                & \multicolumn{2}{c|}{\textbf{Jaccard}} & \multicolumn{2}{c|}{\textbf{Cosine}} & \multicolumn{2}{c}{\textbf{Dice}} \\ \hline
\textbf{Datasets}         & \textbf{ARI $\uparrow$}   &\textbf{ MLR $\downarrow$} & \textbf{ARI $\uparrow$}   & \textbf{MLR $\downarrow$} & \textbf{ARI $\uparrow$}   & \textbf{MLR $\downarrow$}\\ \hline
\texttt{web-Google}     & 0.9991&  0.14\%& 0.9599 & 0.29\%     & 0.9990 & 0.14\%   \\ \hline
\texttt{wiki-topcats}   & 0.9990&  0.06\%& 0.9700 & 0.07\%      & 0.9999 & 0.06\%   \\ \hline
\texttt{soc-Pokec}      & 0.9957&  0.18\%& 0.9609 & 0.23\%      & 0.9958 & 0.15\%  \\ \hline
\texttt{as-skitter}     & 0.9995&  0.19\%& 0.9806 & 0.28\%      & 0.9996 & 0.13\%  \\ \hline
\texttt{wiki-Talk}      & 0.9987&  0.45\%& 0.9672 & 0.53\%      & 0.9989 & 0.42\%  \\ \hline
\texttt{Orkut}          & 0.9946&  0.12\%& 0.9673 & 0.14\%      & 0.9958 & 0.13\%  \\ \hline
soc-LiveJournal1        & 0.9998&  0.12\%& 0.9843 & 0.18\%      & 0.9994 & 0.12\%  \\ \hline
\texttt{soc-Friendster} &0.9947& 0.69\%&  0.9653 & 0.94\%   &  0.9937& 0.62\%     \\ \hline
\texttt{web-2012}   & 0.9926& 0.56\%& 0.9551 & 1.00\%      &  0.9958 & 0.57\%    \\ \bottomrule
\end{tabular}
}
\label{tab:all_qua}
\vspace{-5mm}
\end{table}
\subsection{Experiments on Cosine and Dice Similarities}
\noindent{\bf Cosine.}
As Figure~\ref{fig:cos_update} shows, 
the comparative pattern in update running time is similar to that observed under the Jaccard similarity setting (Figure~\ref{fig:update} above). 
However, the update time of \algo\ shows a slight increase compared with the Jaccard similarity setting (Figure~\ref{fig:update} above) due to larger constant factors in the sample size for similarity estimation and smaller constant factors in $\tau$ for update affordability to ensure the complexity bounds, 
as described in Section~\ref{sec:correctness} and Section~\ref{sec:update-affordability}.

The clustering quality results are shown in Table~\ref{tab:all_qua}. 
Our algorithms also show solid results for Cosine similarity-based structural clustering. 
Compared with exact algorithms like GS*-Index (whose ARI is 1 and MLR is 0 and are omitted from the table), our algorithm can achieve up to 0.9843 average ARI (on \texttt{soc-LiveJournal1}) and as low as 0.07\% MLR (on \texttt{wiki-topcats}).

\noindent{\bf Dice.} The results of Dice similarity are shown in Figure~\ref{fig:dic_update} (update running time) and Table~\ref{tab:all_qua} (clustering quality). As expected, our algorithms show similar performance as that on Jaccard.

\vspace{-2mm}
\section{Conclusion}

We proposed an algorithm called {\em VD-STAR} 
for the problem of {\em Dynamic Structural Clustering for All Parameters}.
Our {\em VD-STAR} algorithm can return an $\rho$-absolute-approximate clustering result with high probability for every query in $O(m_{cr})$ time, and can process each update in $O(\log n)$ amortized expected time, 
while its space consumption is bounded by $O(n + m)$ at all times.
The algorithm works well with Jaccard, Cosine and Dice similarity measurements and supports arbitrary updates.
{\em VD-STAR} significantly improves
the state-of-the-art 
approximate algorithm BOTBIN which achieves $O(\log^2 n)$ expected time per-update for random updates 
under Jaccard similarity only.
We evaluate our algorithm on nine real datasets, which shows strong empirical results 
in terms of update and query efficiency, clustering result quality, and robustness in handling various update distributions.
%
%

\clearpage

\bibliographystyle{unsrt} 
\bibliography{refs}

\end{document}